\documentclass[aps,prx,reprint,superscriptaddress,nofootinbib,longbibliography]{revtex4-2}

\usepackage[T1]{fontenc}
\usepackage{lmodern}
\usepackage{amsmath,amssymb,amsthm,mathtools,bm}
\usepackage{booktabs,array}
\usepackage{capt-of}
\usepackage{graphicx}
\usepackage{microtype}
\usepackage{xcolor}
\usepackage{quantikz}
\usetikzlibrary{arrows.meta,positioning,calc}
\usepackage{enumitem}
\usepackage{placeins}
\usepackage[colorlinks=true,citecolor=blue,linkcolor=blue,urlcolor=blue]{hyperref}
\hypersetup{
pdftitle={A Single Fixed Shallow Circuit for Classical Shadows of Arbitrary n-Qubit States},
pdfauthor={Yu Wang and Xiuwu Zhu},
pdfsubject={Minimal global-observable shadows from one fixed shallow measurement circuit},
pdfkeywords={classical shadows, informationally complete POVM, SIC-POVM, Naimark dilation, quantum circuits}
}

\newtheorem{theorem}{Theorem}
\newtheorem{proposition}[theorem]{Proposition}
\newtheorem{lemma}[theorem]{Lemma}
\newtheorem{corollary}[theorem]{Corollary}
\theoremstyle{remark}

\newcommand{\C}{\mathbb C}
\newcommand{\F}{\mathbb F}

\newcommand{\Tr}{\operatorname{Tr}}
\newcommand{\id}{\mathbb I}
\newcommand{\abs}[1]{\left\lvert#1\right\rvert}
\newcommand{\norm}[1]{\left\lVert#1\right\rVert}
\newcommand{\discussionheading}[1]{\par\addvspace{0.6\baselineskip}\noindent{\normalfont\bfseries #1}\nobreak\enspace\ignorespaces
}
\providecommand{\ket}[1]{\lvert#1\rangle}
\providecommand{\bra}[1]{\langle#1\rvert}
\providecommand{\braket}[2]{\langle#1\vert#2\rangle}
\providecommand{\ketbra}[2]{\lvert#1\rangle\!\langle#2\rvert}

\begin{document}

\title{A Single Fixed Shallow Circuit for Classical Shadows of Arbitrary
\texorpdfstring{\(n\)}{n}-Qubit States}

\author{Yu Wang}
\email{ming-jing-happy@163.com}
\affiliation{Hetao Institute of Mathematics and Interdisciplinary Sciences, Shenzhen, Guangdong 518017, China}

\author{Xiuwu Zhu}
\affiliation{Beijing Key Laboratory of Topological Statistics and Applications for Complex Systems, Beijing Institute of Mathematical Sciences and Applications, Beijing 101408, China}

\date{September 6, 2026}

\begin{abstract}
Classical shadows extract many properties of quantum states from reusable
classical records, typically obtained using randomized measurement settings.
Although individual settings may be shallow, switching among them can
introduce control, calibration, and reconfiguration costs beyond conventional
circuit and sampling metrics.  Here we construct, for every $n$, a fixed
quantum analyzer whose Born outcomes replace externally sampled settings as
labels for reusable shadow snapshots.  The analyzer combines a freshly
prepared $n$-qubit fiducial register $A$ with parallel Bell readout of $A$ and
the unknown system $S$, producing one $2n$-bit record per copy.  The same
circuit realizes a rank-one minimal informationally complete measurement for
tomography: a single fixed setting replaces the $3^n$ local-Pauli measurement
settings commonly used for complete reconstruction, while retaining the
minimum $d^2$ outcomes with $d=2^n$.  The resulting Pauli-diagonal frame
admits an analytic inverse.  For fixed Hermitian observables, the Haar-averaged
conditional variance is bounded by a dimension-independent multiple of the
squared Hilbert--Schmidt norm of the centered observable, whereas the
state-uniform bound has a dimension-dependent coefficient.  For $n\geq3$,
worst-state Pauli variances remain bounded in the axial sector while growing
with dimension in the mixed sector.  Fiducial preparation is completed before
system contact: exact unitary preparation uses $n-1$ arbitrary two-qubit gates,
no work qubits beyond $A$, and logarithmic depth under all-to-all connectivity.
The unknown system undergoes only one parallel system--ancilla entangling
layer, followed by local Hadamards and readout.  Measurement-assisted
preparation achieves $O(1)$ adaptive quantum depth per attempt using $n+1$
extra qubits in the heralded route, or deterministic $O(1)$ adaptive quantum
depth using $O(n\log n)$ extra qubits, all counted beyond $A$.
These statistical and
hardware tradeoffs show that part of the measurement-setting randomness and
control complexity traditionally supplied shot by shot can instead be compiled
into the preparation and circuitry of a fixed, shallow, reusable quantum
analyzer.
\end{abstract}

\maketitle

\begin{figure*}[t!]
\centering
\includegraphics[width=0.80\textwidth]{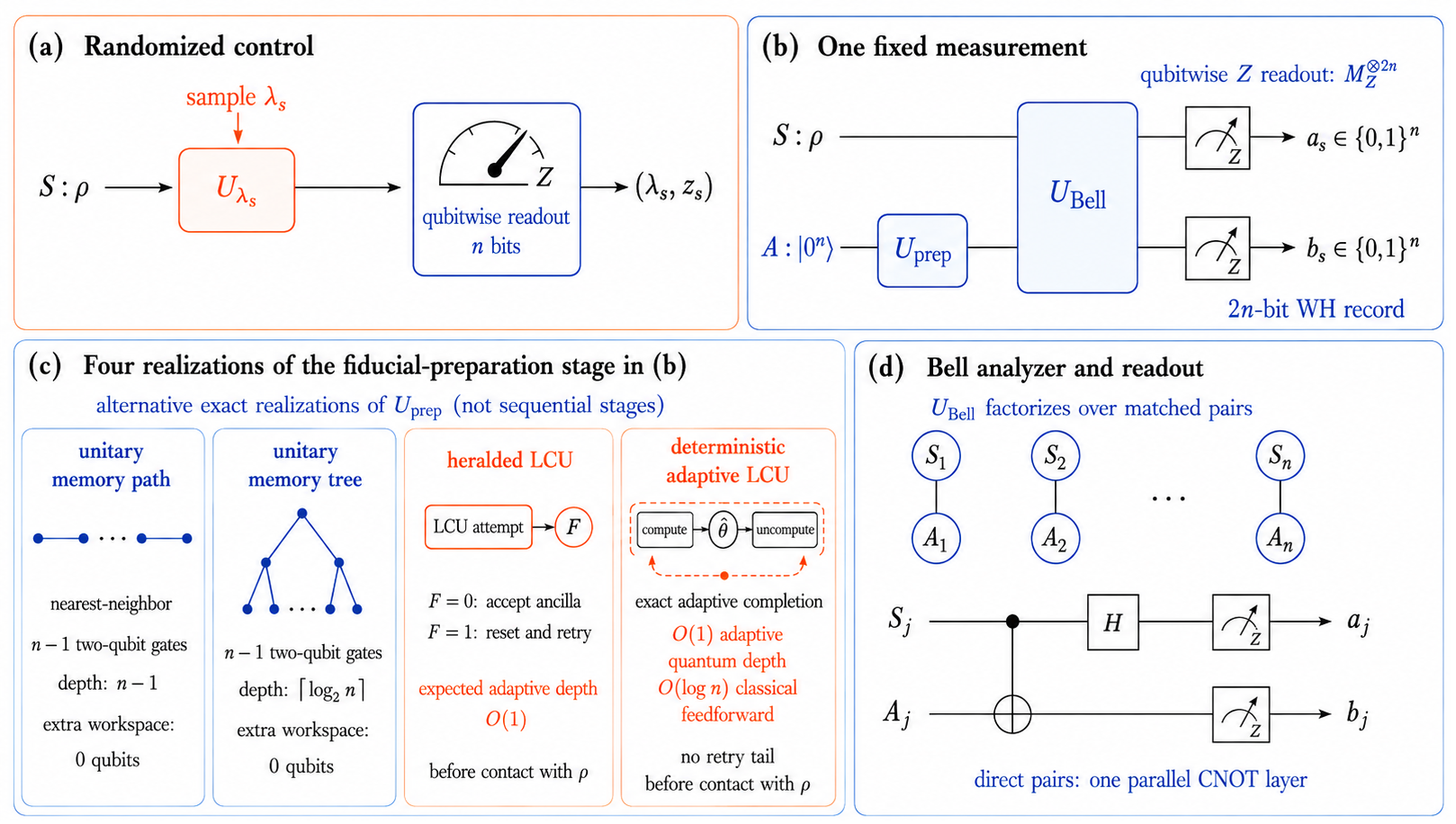}
\caption{Architecture and implementation hierarchy.
(a) Conventional randomized shadows sample a new external setting
$\lambda_s$ on every shot and store it together with the $n$-bit Born outcome
$\bm z_s$.  Local-Pauli shadows sample from $3^n$ bases, while global-Clifford
sampling uses an ensemble of size $2^{\Theta(n^2)}$.  More generally,
complete tomography of an arbitrary $n$-qubit state using only
ancilla-free orthonormal-basis measurements requires at least $2^n+1$ distinct
measurement bases.
(b) The fixed architecture instead repeats the same $U_{\rm prep}$ and
$U_{\rm Bell}$; qubitwise $Z$ readout returns $\bm a_s$ and $\bm b_s$, which
form one $2n$-bit WH record.
(c) The unitary memory path, unitary memory tree, heralded LCU, and
deterministic adaptive LCU are four alternative exact realizations of
$U_{\rm prep}$, not consecutive stages.  Heralding retries before contact
with the unknown state, whereas deterministic adaptive LCU removes the retry
tail by measurement and feedforward.
(d) A matched analyzer applies $\operatorname{CNOT}_{S_j\to A_j}$, then
$H_{S_j}$ and two $Z$-basis readouts.  Direct matched couplers place all $n$
CNOTs in one parallel layer.  Sections~\ref{sec:bell}
and~\ref{sec:fiducial-preparation} prove the induced POVM and resource claims.}
\label{fig:main}
\end{figure*}

\section{Introduction}

Extracting useful information from large quantum systems is a central
challenge in quantum information science.  Meeting it requires both
statistical efficiency and measurement architectures that can be
implemented and controlled at scale
\cite{EisertCertification2020,KlieschRoth2021}.  Classical shadows provide a powerful
framework: repeated measurements on independent copies produce reusable
classical records from which many observables can be estimated after data
acquisition \cite{Huang2020}.  In ensemble-based realizations, each record
combines an externally selected setting \(\lambda\) with a Born outcome
\cite{Elben2023}.  The setting specifies a basis, circuit instance, or
evolution parameter supplied through experimental control; the outcome is
generated by the quantum measurement.  Physically realizing those settings
can require basis switching, pulse reconfiguration, and calibration.
Measurement-setting control is therefore an experimental resource not fully
captured by qubit count, gate count, circuit depth, or sample complexity.
Can part of this randomness and control be compiled into a fixed, shallow,
reusable quantum analyzer?

On the statistical side, ensemble choice determines observable-dependent
sample costs, as quantified by shadow norms for global-Clifford and
local-Pauli measurements \cite{Huang2020}.  Ensemble symmetries make
reconstruction and bounds tractable \cite{Hu2023,Bu2024}, while target-aware
designs optimize sampling and schedules
\cite{Hadfield2022,HuangDerandom2021}, respect interaction-range constraints
\cite{XuPauli2026}, or exploit selected Pauli families
\cite{IppolitiEntangled2024,Wu2026} and fermionic structure
\cite{Wan2023}.  This design freedom extends beyond unitary ensembles to
generalized measurements \cite{Nguyen2022,GarciaPerez2021} and their
reconstruction maps \cite{Caprotti2024}.

Implementation costs involve quantum execution, classical reconstruction,
and experimental control.  Shallow circuits trade entangling depth against
variance \cite{Ippoliti2023,Bertoni2024}, and approximate unitary designs
give logarithmic-depth global-shadow guarantees with controlled bias and
variance corrections \cite{Schuster2025}.  Practical choices also depend
on scrambling and gate errors \cite{HuShallow2025}, platform-specific pulse and readout
imperfections \cite{Notarnicola2023}, and reconstruction costs addressed
by tensor networks for suitable finite-depth ensembles \cite{Akhtar2023}.
Hamiltonian dynamics can replace gate randomization
\cite{HuYou2022,ZhouZhang2024}, although a fixed Hamiltonian can retain a
sampled evolution-time label \cite{LiuHaoHu2024}.  Circuit reuse reduces
reconfiguration overhead by sharing settings across shots, with statistical
tradeoffs \cite{HelsenWalter2023}.  Reducing sampled-circuit depth does not
itself remove external setting selection and realization or their
platform-dependent calibration costs.

A complementary route encodes the measurement identity directly in the
outcome space.  In a fixed informationally complete (IC) positive
operator-valued measure (POVM), the Born outcome directly identifies the
effect and its dual-frame estimator \cite{Acharya2021,Innocenti2023};
fixed system--ancilla quenches already provide analog protocols for such
outcome-only acquisition \cite{Tran2023,McGinleyFava2023}.
An early circuit for single-qubit symmetric informationally complete (SIC)
tomography uses two auxiliary qubits \cite{Rehacek2004}.
Products of single-qubit SIC measurements offer a scalable fixed implementation with local circuits and
product reconstruction, demonstrated through eight qubits
\cite{Stricker2022,You2025}.  Their Pauli second moment is \(3^w\) at weight
\(w\), making high-weight Pauli targets costly to estimate.
A global SIC instead provides an isotropic minimal-IC benchmark with
\(d^2\) rank-one outcomes in dimension \(d\)
\cite{Renes2004,Scott2006}.  Exact SICs have been established
unconditionally in only finitely many dimensions; their existence for every
\(d=2^n\) remains open
\cite{ScottGrassl2010,Fuchs2017,Horodecki2022Open,ApplebyFlammiaKopp2025}.
Even a known fiducial need not admit shallow qubit preparation.
General realization tools include Naimark dilation
\cite{Neumark1940,Peres1990} and group-covariant constructions
\cite{DAriano2004,Decker2005}; probabilistic \cite{Singal2022} and adaptive
measurement schemes \cite{Cai2021,Ivashkov2024} offer alternative resource
tradeoffs.
A fixed analyzer therefore requires joint consideration of statistical
geometry and circuit realizability.
Table~\ref{tab:implementations} summarizes representative implementation
routes and their resource tradeoffs.

Here we construct a fixed analyzer for every \(n\), using a controlled
relaxation of exact SIC symmetry.  An \(n\)-qubit auxiliary register \(A\)
is prepared afresh for each copy in a known conjugate fiducial, before
contact with the unknown system \(S\).  Fixed, fiducial-independent Bell readout produces a
\(2n\)-bit Born outcome labeling the measurement effect and reusable shadow
snapshot.  With direct matched couplers, \(S\) undergoes one parallel
system--ancilla CNOT layer and local Hadamards before readout, with no
\(S\)--\(S\) multiqubit gates.  All fiducial-dependent preparation is confined
to \(A\).  Building on the general fiducial-plus-Bell interface
\cite{GuptaWeiss2025}, we supply explicit preparation and characterize its
statistical and circuit costs.  Figure~\ref{fig:main} contrasts this
interface with randomized-setting acquisition.

The analyzer uses the tensor-Pauli Weyl--Heisenberg (WH) family of
Ref.~\cite{ZhuWang2026}, which established its informational completeness
and projector-Gram spectrum.  It gives a rank-one minimal-IC measurement
with \(d^2=4^n\) outcomes for \(d=2^n\), correlated for \(n\geq2\).
We derive an analytic inverse of its Pauli-diagonal channel and exact
observable-estimation variances.  For any fixed Hermitian target \(O\),
writing \(O_0=O-\Tr(O)\id/d\), we prove
\begin{align*}
\mathbb E_{\psi\sim{\rm Haar}}
\operatorname{Var}_{\ketbra{\psi}{\psi}}[\widehat O]
&\leq 2(1+1/d)\Tr(O_0^2),\\
\sup_\rho\operatorname{Var}_\rho[\widehat O]
&\leq 2(d+1)\Tr(O_0^2).
\end{align*}
For \(n\geq3\), tight worst-state Pauli variances are \(O(1)\) for axial
strings containing only \(X/I\) or only \(Z/I\), and \(\Theta(d)\) for the
remaining nonidentity strings.  The bounded-cost axial sector spans every
weight, while local-Pauli and product-SIC shadows retain bounded costs for
fixed-weight mixed targets.  Global-Clifford shadows retain a stronger
dimension-independent state-uniform bound \cite{Huang2020}.
The same records also support minimal-IC tomography: one fixed setting
replaces the \(3^n\) standard local-Pauli settings, while the
\(d^2\)-outcome alphabet and the sample and storage costs of dense
reconstruction remain exponential.

Exact unitary fiducial preparation uses \(n-1\) arbitrary two-qubit gates
and no work qubits beyond \(A\).  A nearest-neighbor path has linear depth;
a balanced all-to-all schedule has asymptotically optimal two-qubit depth
\(\lceil\log_2n\rceil\).  Measurement-assisted alternatives, completed
before system contact, offer constant expected adaptive quantum depth
through heralding, or deterministic constant adaptive quantum depth with
additional auxiliaries (Sec.~\ref{sec:fiducial-preparation}).  Restricted
connectivity may add routing overhead.

\begin{table*}[t]
\caption{Fixed-analyzer implementation tradeoffs.  Here \(d\) is the
dimension of one encoded qudit, whereas \(n\) counts individually addressable
qubits.  A whole-POVM realization returns one of all outcomes from the same
per-copy device; an effect scan programs the effects separately.  The routes
are representative and need not be mutually exclusive.  The present
correlated family combines minimal IC outcomes and an analytic inverse with
an explicit all-\(n\) qubit circuit and, with direct matched couplers,
a single system-facing entangling layer.}
\label{tab:implementations}
\begin{ruledtabular}
\begingroup
\setlength{\tabcolsep}{4pt}
\renewcommand{\arraystretch}{1.13}
\footnotesize
\begin{tabular}{>{\raggedright\arraybackslash}p{3.05cm}
>{\raggedright\arraybackslash}p{2.85cm}
>{\raggedright\arraybackslash}p{4.75cm}
>{\raggedright\arraybackslash}p{6.10cm}}
implementation route & dimension / scope & map applied to each copy
& circuit implication\\\hline
effect-by-effect photonics \cite{Bent2015}
& SIC projectors and tomography at \(d=6,10\)
& \textbf{Effect scan:} program one SIC projector at a time and combine data
from \(d^2\) settings
& one encoded qudit; programmable projection measurements acquired
separately and combined for reconstruction\\
fixed single-qudit dilations
\cite{Tabia2012,Medendorp2011,Feng2025,Yan2026,Zhao2015,Bian2015,WangWalk2023}
& fixed multioutcome SICs at \(d=2,3,4\); quantum-walk SICs at \(d=2,3\)
& \textbf{Fixed dilation:} the same loop, multiport, or walk acts on every
copy; its terminal port or position is the outcome
& resources are expressed in native modes, ports, or walk steps;
translation to individually addressable qubits depends on the encoding\\
product qubit SIC \cite{Rehacek2004,Stricker2022,You2025,Fischer2022}
& demonstrated through \(n=8\); tensor-product minimal IC for every \(n\), with an explicit product inverse
& \textbf{Fixed local product measurement:} \(n\) one-qubit tetrahedral-SIC
circuits act independently and in parallel
& tensoring the optimized one-qubit ancilla circuit \cite{You2025} gives \(n\) disjoint
CNOTs in one layer; higher-level embeddings offer an alternative
\cite{Fischer2022}.  Product structure gives Pauli second moment
\(3^w\) at weight \(w\)\\
adaptive generalized measurements \cite{Cai2021,Ivashkov2024,AnderssonOi2008}
& general synthesis; SIC demonstrations on an encoded \(d=4\) qudit
\cite{Cai2021} and on up to two qubits \cite{Ivashkov2024}
& \textbf{Adaptive realization:} ancilla readout and reset, with
outcome-conditioned operations; hybrid schemes terminate in a Naimark block
& a complete logical POVM with a reused auxiliary qubit; sequential
control trades auxiliary width for measurement and feedback, with
target- and platform-dependent gate costs\\
\textbf{present tensor-Pauli family} \cite{ZhuWang2026}
& \textbf{every \(n\)}; rank-one minimal IC, correlated for \(n\geq2\), with
exact SICs at \(n=1,3\)
& \textbf{Fixed whole-POVM unitary analyzer:} prepare \(\ket{\phi_n^*}\) and
apply matched Bell readout; analytic Pauli-diagonal inverse
& \textbf{System interface:} one parallel CNOT layer, no \(S\)--\(S\) gates.
Balanced all-to-all total for \(n\geq2\): \(3n-3\) CNOTs, CNOT depth
\(2\lceil\log_2n\rceil\), before routing; at \(n=1\), one Bell CNOT\\
\end{tabular}
\endgroup
\end{ruledtabular}
\end{table*}

Together, these results show how auxiliary width and preparation can be
traded for reduced setting control while keeping statistical costs explicit.

\section{One fixed minimal measurement and its reusable record}
\label{sec:fixed-shadows}

\subsection{The fixed Weyl--Heisenberg minimal IC measurement}

For any $n$-qubit system, the Hilbert-space dimension is $d=2^n$.  We label
each Weyl--Heisenberg displacement by two bit strings
\(\bm a=(a_1,\ldots,a_n)\) and
\(\bm b=(b_1,\ldots,b_n)\) in \(\F_2^n\).  Up to an irrelevant Pauli phase,
define

\begin{equation}
D_{\bm a,\bm b}=Z(\bm a)X(\bm b)
=\bigotimes_{j=1}^n Z_j^{a_j}X_j^{b_j}.
\label{eq:field-to-pauli}
\end{equation}

The concrete fiducial used throughout is
\begin{align}
\ket{\phi_n}&=
\frac{\ket+^{\otimes n}+e^{i\theta_n}\ket{0^n}}{\sqrt{\mathcal N_n}},
&\mathcal N_n&=2+2^{1-n/2}\cos\theta_n, \label{eq:fid}\\
\theta_1&=5\pi/12,\qquad \theta_2=2\pi/3,
&\theta_n&=3\pi/4\quad(n\geq3).\nonumber
\end{align}

Its tensor-Pauli orbit defines

\begin{equation}
\Pi_{\bm a,\bm b}^{(n)}
=D_{\bm a,\bm b}\ketbra{\phi_n}{\phi_n}D_{\bm a,\bm b}^\dagger,
\qquad E_{\bm a,\bm b}^{(n)}=\frac1d\Pi_{\bm a,\bm b}^{(n)}.
\label{eq:fixed-povm}
\end{equation}

The Pauli twirl gives
\(\sum_{\bm a,\bm b}E_{\bm a,\bm b}^{(n)}=\id\), so the positive rank-one
operators in Eq.~\eqref{eq:fixed-povm} form a POVM.  Appendix~\ref{app:povm-normalization}
gives a direct Pauli-basis proof of this normalization.  We denote this fixed
POVM by
\begin{equation}
\mathsf E^{(n)}
:=\left\{E_{\bm a,\bm b}^{(n)}:
\bm a,\bm b\in\F_2^n\right\}.
\label{eq:fixed-measurement-set}
\end{equation}

Normalization alone does not imply informational completeness.  The orbit in
Eq.~\eqref{eq:fixed-povm} is the physical-qubit tensor-Pauli form of the
characteristic-two finite-field WH family studied in
Ref.~\cite{ZhuWang2026}, which establishes informational completeness and
determines its projector-Gram spectrum.  Since the POVM has exactly \(d^2\)
outcomes, it is minimal IC.  Section~\ref{sec:power-two} and
Appendices~\ref{app:pauli-channel-proof} and \ref{app:all-qubit-maps} provide
an independent tensor-Pauli derivation.

As a separate consequence, the equal-weight rank-one projectors form an
operator basis, so this is also a balanced informationally complete (BIC)
POVM.  In the associated Bell protocol, maximal violation can certify the
maximal $2\log_2d=2n$ device-independent private random bits
\cite{Farkas2026}.  The circuit developed here supplies the fixed measurement
component of that protocol.

Later we analyze how this one fixed POVM performs as an alternative to the
exponentially large ensembles of randomized measurement circuits used in
standard classical shadows and to the $3^n$ local-Pauli settings used in
conventional tomography.

\subsection{One outcome, one canonical snapshot}

Measuring one copy of an $n$-qubit state $\rho$ with $\mathsf E^{(n)}$ returns
a $2n$-bit outcome label $u=(\bm a,\bm b)$.  From now on, write
\(E_u:=E_{\bm a,\bm b}^{(n)}\) and
\(\Pi_u:=\Pi_{\bm a,\bm b}^{(n)}\).  To reuse this record for either shadow
estimation or tomography, we assign to it a matrix-valued estimator.  The
outcome occurs with Born probability
$p_\rho(u)=\Tr(\rho E_u)=d^{-1}\Tr(\rho\Pi_u)$.  If we associate the raw
projector $\Pi_u$ with this outcome, its average over repeated shots is
\begin{equation*}
\mathbb E_\rho[\Pi_u]
=\sum_u p_\rho(u)\Pi_u
=\mathcal M(\rho),
\end{equation*}
where the normalized frame channel of the POVM is

\begin{equation}
\mathcal M(A)=\frac1d\sum_u\Tr(\Pi_uA)\Pi_u.
\label{eq:fixed-frame}
\end{equation}

Thus $\mathcal M$ maps the input state $\rho$ to the average raw estimator
$\mathbb E_\rho[\Pi_u]$.  Minimal informational completeness makes
$\mathcal M$ invertible, so we correct each outcome by defining the canonical
snapshot
\begin{equation}
\widehat\rho_u=\mathcal M^{-1}(\Pi_u).
\label{eq:canonical-snapshot}
\end{equation}
Since $\mathcal M(\rho)=\sum_u p_\rho(u)\Pi_u$, the object inverted for each
outcome is $\Pi_u$.  Equivalently,
\begin{equation*}
\widehat\rho_u=d\,\mathcal M^{-1}(E_u).
\end{equation*}
When the measurement outcome $u$ is drawn according to the Born probabilities
$p_\rho(u)$, the corrected snapshot is unbiased:

\begin{equation}
\mathbb E_{u\sim p_\rho}[\widehat\rho_u]
=\sum_up_\rho(u)\mathcal M^{-1}(\Pi_u)
=\mathcal M^{-1}[\mathcal M(\rho)]=\rho.
\label{eq:fixed-shadow}
\end{equation}
Thus every Born outcome, after the fixed correction $\mathcal M^{-1}$, gives
an unbiased one-shot estimator of the full state, although it need not be
positive semidefinite.  Averaging these snapshots gives linear tomography,
whereas contracting them with selected observables gives classical-shadow
estimates.  Section~\ref{sec:power-two} evaluates the correction map
explicitly.

\subsection{Two estimators from the same records}

\paragraph{Observable prediction.}
For each Hermitian target $O$, one evaluates only the scalar
\begin{equation}
\widehat O(u)=\Tr(O\widehat\rho_u),
\qquad
\widehat O_N=\frac1N\sum_{s=1}^N\widehat O(u_s).
\label{eq:observable-estimator}
\end{equation}
Equation~\eqref{eq:fixed-shadow} gives
$\mathbb E_\rho\widehat O_N=\Tr(\rho O)$.  This is the classical-shadow route:
many observables can be evaluated from the same records without first
constructing a dense estimate of $\rho$.

\paragraph{Linear tomography.}
Let $N_u$ be the number of times outcome $u$ occurs and $f_u=N_u/N$ its
empirical frequency.  Averaging the same snapshots gives the closed-form
linear-inversion estimator
\begin{equation}
\widehat\rho_N
=\frac1N\sum_{s=1}^N\widehat\rho_{u_s}
=\sum_u f_u\,\mathcal M^{-1}(\Pi_u),
\qquad \mathbb E_\rho\widehat\rho_N=\rho.
\label{eq:linear-tomography}
\end{equation}
The two expressions are identical: one averages record-by-record snapshots,
or equivalently applies the dual frame to the empirical POVM probabilities.
Because the POVM is minimal IC, this unbiased linear inverse is unique.  For
an exact SIC, tightness reduces each snapshot to the familiar form
$\widehat\rho_u=(d+1)\Pi_u-\id$ \cite{Renes2004,Scott2006,ZhuEnglert2011}.
Our WH POVM is generally non-tight: different Pauli components are rescaled by
different inverse eigenvalues, rather than by a single universal SIC factor.
Its Pauli-basis inverse and the closed-form dual for our fiducial are derived
below.

\section{Measurement channel and canonical reconstruction}
\label{sec:power-two}

For the fiducial in Eq.~\eqref{eq:fid}, write
\(\mathcal M_{\phi_n}\equiv\mathcal M\) for the frame channel in
Eq.~\eqref{eq:fixed-frame}.  The Pauli-basis diagonalization below uses only
tensor-Pauli covariance; the chosen fiducial enters through its Pauli
expectation values.  Each outcome \(u\) is converted into the canonical
snapshot \(\widehat\rho_u=\mathcal M_{\phi_n}^{-1}(\Pi_u)\).  We evaluate this
map before discussing either its statistical performance or its physical
implementation.

\subsection{Measurement channel in the Pauli basis}

Let \(P_v\), with \(v=(\bm c,\bm e)\), denote the Hermitian tensor-Pauli
whose \(j\)th factor is \(I,X,Z,Y\) for
\((c_j,e_j)=(0,0),(0,1),(1,0),(1,1)\), respectively.  It differs from
\(D_{\bm c,\bm e}\) in Eq.~\eqref{eq:field-to-pauli} only by a phase.  These
operators satisfy
\(P_v^\dagger=P_v\) and
\(\Tr(P_vP_w)=d\delta_{v,w}\).  For each \(P_v\), define its expectation value
in the fiducial state \(\lvert\phi_n\rangle\) by
\begin{equation}
\chi_v:=\langle\phi_n|P_v|\phi_n\rangle .
\label{eq:fiducial-response}
\end{equation}
Because \(P_v\) is Hermitian, every \(\chi_v\) is real.

Pauli-diagonal reconstruction for Pauli-invariant unitary ensembles was
established by Bu et al.~\cite{Bu2024}.  We derive the corresponding form
for the fixed Pauli-covariant POVM below.

\begin{theorem}[Measurement channel and its inverse in the Pauli basis]
\label{thm:pauli-diagonal-reconstruction}
For every tensor-Pauli operator \(P_v\),
\begin{equation}
\mathcal M_{\phi_n}(P_v)=|\chi_v|^2P_v.
\label{eq:pauli-channel-eigenvalue}
\end{equation}
Consequently, for any operator \(A\), the channel has the first form below.
If every \(\chi_v\) is nonzero, then
\(\mathcal M_{\phi_n}\) is invertible and its inverse is the second:
\begin{align}
\mathcal M_{\phi_n}(A)
&=\frac1d\sum_v|\chi_v|^2\Tr(P_vA)P_v,\nonumber\\
\mathcal M_{\phi_n}^{-1}(A)
&=\frac1d\sum_v
\frac{\Tr(P_vA)}{|\chi_v|^2}P_v.
\label{eq:pauli-channel-closed-form}
\end{align}
\end{theorem}

Pauli covariance reduces all \(d^2\) snapshots to one reference dual.  With
\(D_u:=D_{\bm a,\bm b}\) and \(\Pi_0:=\ketbra{\phi_n}{\phi_n}\), define
\begin{align}
F_n&:=\mathcal M_{\phi_n}^{-1}(\Pi_0)
=\frac1d\sum_v\frac1{\chi_v}P_v,
\label{eq:fixed-dual}\\
\widehat\rho_u
&=\mathcal M_{\phi_n}^{-1}(\Pi_u)=D_uF_nD_u^\dagger.
\label{eq:fixed-dual-covariance}
\end{align}
Thus one computes the inverse only for the fiducial projector \(\Pi_0\);
conjugation by the observed displacement \(D_u\) gives every snapshot.
Appendix~\ref{app:pauli-channel-proof} proves the Pauli diagonalization, the
channel and inverse formulas, and the covariance used above.  Equation~\eqref{eq:pauli-channel-closed-form}
therefore reduces the inverse to coefficientwise scalar division; no
\(d^2\times d^2\) matrix is inverted.

\subsection{Closed-form dual for the chosen fiducial}

Equation~\eqref{eq:fixed-dual-covariance} reduces all \(d^2\) reconstruction
rules to one reference dual \(F_n\).  For \(n\geq3\), set
\(r_d:=\sqrt{2/d}\), \(\kappa_d:=1-r_d\), and
\(z_n:=e^{3\pi i/4}\).  Equation~\eqref{eq:fid} then gives
\(\mathcal N_n=2-r_d\).

\begin{theorem}[Closed-form canonical dual]
\label{thm:structured-reconstruction}
For the fiducial family in Eq.~\eqref{eq:fid}, the two exceptional reference
duals are
\begin{align}
F_1&=3\Pi_0-\id,\nonumber\\
F_2&=6\Pi_0+3\Pi_0^{\mathsf T}-2\id,
\end{align}
where the transpose is taken in the computational basis.  For \(n\geq3\),
the reference dual is
\begin{align}
F_n={}&\frac{\mathcal N_n}{r_d\kappa_d}
(\ketbra{+^n}{+^n}+\ketbra{0^n}{0^n})\nonumber\\
&+\frac{\mathcal N_n}{r_d^2}\left(
z_n^*\ketbra{+^n}{0^n}+z_n\ketbra{0^n}{+^n}\right)
-\frac{r_d}{\kappa_d}\id.
\label{eq:explicit-fixed-dual}
\end{align}
For every \(n\) and every outcome \(u\), the canonical snapshot is
\(\widehat\rho_u=D_uF_nD_u^\dagger\).
\end{theorem}

At \(n=1\) and \(n=3\), the Pauli--WH orbits are the tetrahedral and Hoggar
SICs, respectively \cite{Renes2004,ZhuWang2026}.  In both cases the result
reduces to the standard SIC rule
\(\widehat\rho_u=(d+1)\Pi_u-\id\), giving \(3\Pi_u-\id\) and
\(9\Pi_u-\id\), respectively.

For \(u=(\bm a,\bm b)\), write
\(\ket{\bm b}:=X(\bm b)\ket{0^n}\).  Tensor-Pauli conjugation acts on the two
product branches as
\begin{equation}
D_u\ket+^{\otimes n}=Z(\bm a)\ket+^{\otimes n},
\qquad
D_u\ket{0^n}=(-1)^{\bm a\cdot\bm b}\ket{\bm b}.
\label{eq:product-branch-conjugation}
\end{equation}
For \(n\geq3\), Eq.~\eqref{eq:product-branch-conjugation} shows directly that
tensor-Pauli conjugation preserves the four rank-one product-state terms in
\(F_n\).  Hence every snapshot has the same five-term structure.
Appendix~\ref{app:all-qubit-maps} obtains the fiducial-specific sectorwise
inverse, and Appendix~\ref{app:closed-form-dual-proof} converts it into this
product form and its explicit outcome-dependent version.

\begin{samepage}
The product form also makes the role of coherence transparent.  Apart from
the identity label \((\bm0,\bm0)\), the Pauli directions split into the
axial and mixed sets
\begin{align}
\mathsf{Ax}
&:=
\underbrace{\{(\bm0,\bm e):\bm e\neq\bm0\}}_{X/I\text{-only}}
\nonumber\\
&\quad\cup
\underbrace{\{(\bm c,\bm0):\bm c\neq\bm0\}}_{Z/I\text{-only}},
\nonumber\\
\mathsf{Mix}
&:=
\underbrace{\{(\bm c,\bm e):\bm c\neq\bm0,\ \bm e\neq\bm0\}}
_{\text{mixed directions}}.
\label{eq:axis-mixed-label-sets}
\end{align}
\end{samepage}
Together with the identity, the axial sector contains \(2d-1\) Pauli
directions, while the mixed sector contains the remaining \((d-1)^2\).
The WH orbit of the equal-weight mixture
\(\tau_{\rm inc}=(\ketbra{+^n}{+^n}+\ketbra{0^n}{0^n})/2\) spans only the
identity and axial directions and therefore has frame rank \(2d-1\).  The two cross outer
products in the coherent fiducial projector give nonzero Pauli coefficients
in the remaining \((d-1)^2\) directions, completing the IC frame.
Section~\ref{sec:observable-estimation} quantifies the resulting
direction-dependent reconstruction cost.

\section{Observable-estimation performance}
\label{sec:observable-estimation}

\subsection{Spectral input and standard benchmarks}
\label{sec:spectral-benchmarks}

For canonical classical shadows constructed from an IC POVM, inversion of
the measurement frame enters the single-shot raw second moments and controls
standard variance and sample-complexity bounds
\cite{Nguyen2022,Innocenti2023}.  The weakest nonidentity frame direction is
therefore the relevant spectral quantity for such bounds.

For the Pauli--WH POVM generated by the fiducial \(\ket{\phi_n}\) in
Eq.~\eqref{eq:fid}, this amplification can be read directly from the
projector-Gram spectrum computed in Ref.~\cite{ZhuWang2026}.  Define
\(G^\Pi_{u,u'}:=\Tr(\Pi_u\Pi_{u'})\).  Pauli covariance associates direction
\(P_v\) with the projector-Gram eigenvalue
\(\lambda_v:=d|\chi_v|^2\) (Appendix~\ref{app:gram-connection}).  Thus
Eq.~\eqref{eq:pauli-channel-eigenvalue} becomes
\begin{align*}
\mathcal M_{\phi_n}(P_v)&=\frac{\lambda_v}{d}P_v,\\
\mathcal M_{\phi_n}^{-1}(P_v)
&=\frac d{\lambda_v}P_v=\frac1{|\chi_v|^2}P_v.
\end{align*}
Hence \(d/\lambda_v\) is the inverse-frame amplification of the Pauli
direction \(P_v\).  For a Pauli target, this factor is exactly the
state-independent raw second moment of the canonical estimator; for a general
target, the inverse spectrum enters through a quadratic form and controls
bounds rather than determining every variance by itself.  The largest
nonidentity inverse factor is
\(d/\lambda_{\min}\), where
\begin{equation}
\lambda_{\min}=\min_{v\ne0}\lambda_v,
\qquad
\eta_n=\frac{d+1}{d}\lambda_{\min}.
\label{eq:eta}
\end{equation}
Maximizing \(\lambda_{\min}\) therefore minimizes the worst Pauli-direction
amplification.  This gives the phase-selection principle for
Eq.~\eqref{eq:fid}: within its equal-weight two-product-state family, the
listed phases maximize \(\lambda_{\min}\) over the relative phase
(Appendix~\ref{app:full-spectrum}).  An exact global SIC has
\(\lambda_{\min}=d/(d+1)\), so \(\eta=1\) is the unit benchmark.

\begin{theorem}[Exact spectral floor \cite{ZhuWang2026}]
\label{thm:spectrum}
For the family in Eq.~\eqref{eq:fid},
\begin{equation}
\lambda_{\min}=
\begin{cases}
2/3,&n=1,\\
4/9,&n=2,\\
\displaystyle\frac{2}{(2-\sqrt{2/d})^2},&n\geq3.
\end{cases}
\label{eq:lambda-min-exact}
\end{equation}
Consequently, \((\eta_1,\eta_2,\eta_3)=(1,5/9,1)\), and for all \(n\),
\begin{equation}
\eta_n\geq\frac12,\qquad
\eta_n\longrightarrow\frac12.
\label{eq:floor}
\end{equation}
The exact finite-\(n\) multiplicities are listed in
Appendix~\ref{app:full-spectrum}.
\end{theorem}

Thus, although the family is not a SIC except at \(n=1,3\), it retains at
least half the SIC spectral floor in every size.  Equivalently, its
worst-direction inverse-frame amplification is at most twice the SIC
benchmark:
\begin{equation}
\begin{array}{c|c}
\text{measurement}&d/\lambda_{\min}\\ \hline
\text{global SIC}&d+1\\
\text{product tetrahedral SIC}&3^n\\
\text{present global WH}&\leq2(d+1)
\end{array}
\label{eq:spectral-benchmark-table}
\end{equation}
The global-SIC and product-SIC values are standard
\cite{Stricker2022,You2025,Renes2004,Scott2006}.  The exact value of
\(\lambda_{\min}\) for the present construction is given in
Eq.~\eqref{eq:lambda-min-exact}.
The product construction is explicit for every \(n\), but its
weakest-direction amplification exceeds the global-SIC value by the factor
\(3^n/(d+1)\sim(3/2)^n\); equivalently,
\begin{equation}
\eta_{\rm prod}(n)=\frac{d+1}{3^n}\longrightarrow0.
\label{eq:product-floor}
\end{equation}
In contrast, the present WH family is explicit for every \(n\)-qubit system
and satisfies \(\eta_n\geq1/2\).  Its worst-direction inverse-frame
amplification is therefore at most \(1/\eta_n\leq2\) times the global-SIC
benchmark; this ratio equals one for \(n=1,3\) and approaches two as
\(n\to\infty\).  The comparison uses the global-SIC value only as a benchmark
and does not assume the existence of an exact SIC in every dimension \(d=2^n\)
\cite{Horodecki2022Open,ApplebyFlammiaKopp2025}.

\subsection{General Hermitian observables}

For a Hermitian target $O$, one record returns
$\widehat O(u)=\Tr(O\widehat\rho_u)$.

Writing $O_0=O-\Tr(O)\id/d$, define for the canonical snapshots the
state-dependent squared shadow norm and its worst-state counterpart by
\begin{equation*}
\norm{O_0}_{\mathrm{sh},\rho}^{2}
:=\mathbb E_\rho[\widehat O_0^{\,2}],
\qquad
\norm{O_0}_{\mathrm{sh}}^{2}
:=\sup_\rho\norm{O_0}_{\mathrm{sh},\rho}^{2},
\end{equation*}
where the supremum is over density operators.  Thus the shadow norm is a raw
second moment, while
\(\operatorname{Var}_\rho[\widehat O]
=\norm{O_0}_{\mathrm{sh},\rho}^{2}-\Tr(\rho O_0)^2\).

\begin{proposition}[General observable: average and worst-case variance]
\label{thm:general-observable-variance}
Let $O$ be any Hermitian observable and let $O_0$ be its traceless part
defined above.  Draw $\ket{\psi}$ from Haar measure on pure states and set
$\rho_\psi=\ketbra{\psi}{\psi}$.  For each fixed $\rho_\psi$,
$\operatorname{Var}_{\rho_\psi}[\widehat O]$ is the variance over the
measurement outcome.  Averaging this variance over the Haar-random choice of
$\ket{\psi}$ gives
\begin{equation}
\begin{aligned}
&\mathbb E_{\psi\sim\mathrm{Haar}}
\operatorname{Var}_{\rho_\psi}[\widehat O]\\
&\quad=\frac1d\Tr[O_0\mathcal M_{\phi_n}^{-1}(O_0)]
-\frac{\Tr(O_0^2)}{d(d+1)}\\
&\quad\leq2\left(1+\frac1d\right)\Tr(O_0^2).
\end{aligned}
\label{eq:haar-average-variance}
\end{equation}
Uniformly over all input states, the worst-case variance obeys
\begin{equation}
\sup_\rho\operatorname{Var}_\rho[\widehat O]
\leq2(d+1)\Tr(O_0^2).
\label{eq:worst-state-general-variance}
\end{equation}
\end{proposition}

The distinction between the two bounds is already visible in the raw second
moment.  It is linear in the input state, so its Haar average equals its value
at \(\id/d\); for an arbitrary state, \(\rho\leq\id=d(\id/d)\) bounds it by
\(d\) times that value.
The proof is given in Appendix~\ref{app:general-observable-proof}.

For arbitrary Hermitian targets, including fully global ones, the Haar-input
average in Eq.~\eqref{eq:haar-average-variance} has a
dimension-independent coefficient for all the global benchmarks considered
here.  An ideal global SIC, uniform sampling from a complete set of stabilizer
MUBs, and randomized global-Clifford measurements realize the same
depolarizing frame and give
\((1+1/(d+1))\Tr(O_0^2)\)
\cite{Huang2020,Renes2004,Scott2006,Zhang2024,WangCui2024}.  The complete
stabilizer MUB uses the minimum \(d+1\) Clifford measurement bases; it is an
exact finite realization of this frame, not a minimal-IC POVM.  The present
fixed \(d^2\)-outcome minimal-IC measurement is instead bounded by
\(2(1+1/d)\Tr(O_0^2)\).

State-uniformly, the comparison is different.  The full randomized
global-Clifford ensemble satisfies, for every Hermitian target,
\(\sup_\rho\operatorname{Var}_\rho[\widehat O_{\mathrm{Cl}}]
\leq3\Tr(O_0^2)\), whereas Eq.~\eqref{eq:worst-state-general-variance} gives
the upper bound \(2(d+1)\Tr(O_0^2)\) for the fixed WH measurement
\cite{Huang2020}.

A shallow randomized alternative trades circuit depth for a controlled
statistical penalty.  On a one-dimensional nearest-neighbor architecture
without ancillas, a relative-error
\(\varepsilon_{\mathrm{des}}\)-approximate unitary \(3\)-design reduces the
Clifford depth from \(O(n)\) to
\(O(\log(n/\varepsilon_{\mathrm{des}}))\).  For \(O\geq0\), the standard
reconstruction has variance at most
\(3\Tr(O_0^2)+10\varepsilon_{\mathrm{des}}[\Tr(O)]^2\)
and bias at most \(2\varepsilon_{\mathrm{des}}\Tr(O)\)
\cite{Schuster2025}.

The gap between averaged and state-uniform guarantees is also reflected in
outcome count.  The present POVM is minimal IC and has \(d^2\) outcomes,
whereas, among rank-one IC POVMs, a bound of the form
\(\sup_\rho\mathbb E_\rho[\widehat O_0^{\,2}]
\leq C\Tr(O_0^2)\), uniformly for all \(O_0\) with \(C\) independent of
\(d\), requires $\Omega(d^3)$ outcomes \cite{Yang2026}.

\begin{corollary}[State-dependent simultaneous estimation]
\label{cor:simultaneous-estimation}
Fix an input state $\rho$, let $O_1,\ldots,O_M$ be Hermitian targets, and put
\begin{equation*}
B_{\mathrm{var}}(\rho)
:=\max_k\operatorname{Var}_\rho[\widehat O_k].
\end{equation*}
Median-of-means postprocessing of independent records drawn from $\rho$
estimates all $M$
expectations to additive error $\epsilon$ with failure probability at most
$\delta$ using
\begin{equation}
N=O\!\left(
\frac{B_{\mathrm{var}}(\rho)\log(M/\delta)}{\epsilon^2}
\right).
\label{eq:sample-complexity}
\end{equation}
The quantum circuit and raw records are independent of the chosen targets;
Eq.~\eqref{eq:worst-state-general-variance} may be substituted for
$B_{\mathrm{var}}(\rho)$ when a state-independent shot count is required.
\end{corollary}

Proposition~\ref{thm:general-observable-variance} and
Corollary~\ref{cor:simultaneous-estimation} apply to arbitrary Hermitian
targets.  Since Pauli operators are eigenoperators of the measurement channel,
their variances admit an exact evaluation.

\subsection{Exact Pauli expectations}
\label{sec:pauli-performance}

For \(n\geq3\), we use the Pauli-sector split displayed after
Theorem~\ref{thm:structured-reconstruction}:
\begin{equation*}
\begin{aligned}
\{I\}\;&\mathbin{\dot\cup}
\underbrace{(X/I\text{-only})\mathbin{\dot\cup}(Z/I\text{-only})}
_{\mathsf{Ax},\;2(d-1)}\\
&\mathbin{\dot\cup}
\underbrace{(\text{mixed directions})}
_{\mathsf{Mix},\;(d-1)^2}.
\end{aligned}
\end{equation*}
For a Pauli label \(v=(\bm c,\bm e)\), the axial class
\(\mathsf{Ax}\) has exactly one of \(\bm c,\bm e\) nonzero, whereas the mixed
class \(\mathsf{Mix}\) has both nonzero.  Equivalently, a mixed string
contains a \(Y\), or contains an \(X\) and a \(Z\) on different qubits.

For $u=(\bm a,\bm b)$ and $v=(\bm c,\bm e)$, the commutation sign and the
canonical Pauli estimator are
\begin{align}
s_v(u)&:=(-1)^{\bm a\cdot\bm e+\bm b\cdot\bm c},
&D_uP_vD_u^\dagger&=s_v(u)P_v,\nonumber\\
\widehat P_v(u)
&:=\Tr(P_v\widehat\rho_u)=\frac{s_v(u)}{\chi_v},
&s_v(u)&\in\{+1,-1\},\nonumber\\
\widehat P_v(u)^2
&=\frac1{|\chi_v|^2}=\frac{d}{\lambda_v}=:C(P_v).
\label{eq:exact-pauli-second-moment}
\end{align}
Thus the outcome changes only the sign of the estimate, not its magnitude.
Consequently,
$C(P_v)=\mathbb E_\rho[\widehat P_v^{\,2}]$ is independent of the input
state.  The sign-resolved derivation is given in
Appendix~\ref{app:pauli-second-moment-proof}.

\begin{theorem}[Exact classwise Pauli variances for $n\geq3$]
\label{thm:global-pauli-variance}
Define
\begin{align}
A_d&:=\left(\frac{2-\sqrt{2/d}}{1-\sqrt{2/d}}\right)^2
\in(4,9],\quad A_d\to4,\nonumber\\
M_d&:=(\sqrt{2d}-1)^2<2d,\quad
M_d=\Theta(d)=\Theta(2^n).
\label{eq:axis-mixed-moments}
\end{align}
For a fixed state $\rho$, write $m_v(\rho):=\Tr(\rho P_v)$.  Then every
nonidentity Pauli direction has the following exact variance, raw second
moment, and tight worst-state variance:
\begin{equation}
\begin{aligned}
\operatorname{Var}_\rho(\widehat P_v)
&=C(P_v)-m_v(\rho)^2,\\
C(P_v)&=\mathbb E_\rho[\widehat P_v^{\,2}]
=\sup_\sigma\operatorname{Var}_\sigma(\widehat P_v)\\
&=\begin{cases}
A_d, & P_v\in\mathsf{Ax},\\
M_d, & P_v\in\mathsf{Mix}.
\end{cases}
\end{aligned}
\label{eq:global-pauli-variance}
\end{equation}
The supremum in each row is attained at $\sigma=I/d$; the identity has zero
variance.  The axial row contains the $2(d-1)$ nonidentity $X/I$-only and
$Z/I$-only strings, while the mixed row contains the remaining $(d-1)^2$
directions.  At $n=3$, the two values coincide at $A_8=M_8=9$; for $n>3$,
$M_d>A_d$.  Consequently,
\begin{equation*}
\sup_{\rho,\,P\ne I}\operatorname{Var}_\rho(\widehat P)
=M_d=\Theta(2^n),\qquad n\geq3.
\end{equation*}
\end{theorem}

The same variance identity also covers the two smaller cases.  At $n=1$, all
three nonidentity Paulis have $C(P)=3$.  At $n=2$, $\bm c\cdot\bm e$ is the
parity of the number of $Y$ factors: the even- and odd-$Y$ sectors contain
nine and six directions with $C(P)=9$ and $3$, respectively.  In both cases,
the worst-state value is attained at $I/d$.

The classwise law above contrasts with the weight law of local measurements.
For the $n$-fold product tetrahedral SIC \cite{Stricker2022} and uniformly
randomized local-Pauli shadows \cite{Huang2020}, the canonical estimators obey
\begin{equation}
\mathbb E_\rho[\widehat P^2]=3^{w(P)},\qquad
\operatorname{Var}_\rho(\widehat P)
=\mathbb E_\rho[\widehat P^2]-\Tr(\rho P)^2.
\label{eq:product-pauli-variance}
\end{equation}
Their cost depends only on the Pauli weight $w(P)$ and is unchanged when
$X$, $Y$, and $Z$ are exchanged on a fixed support.  The present WH cost
depends instead on the axial--mixed class and not on weight.  Thus, for
$n>3$, $X^{\otimes n}$ and $Z^{\otimes n}$ have second moment
$A_d=O(1)$ here but $3^n$ under the local schemes, whereas a weight-one
$Y_j$ has second moment $M_d=\Theta(2^n)$ here but only $3$ locally.  Neither
scheme uniformly dominates the other: the fixed WH measurement favors the
axial sector across all weights, while the local schemes favor low-weight
mixed directions.  Appendix~\ref{app:product-sic-baselines} gives a
tensor-product derivation of Eq.~\eqref{eq:product-pauli-variance}.

Combining the exact moments with
Corollary~\ref{cor:simultaneous-estimation} gives, for $n\geq3$,
\begin{equation*}
\begin{aligned}
N_{\mathsf{Ax}}
&=O\!\left(\frac{n+\log(1/\delta)}{\epsilon^2}\right),\\
N_{\mathrm{all\ Pauli}}
&=O\!\left(
\frac{2^n[n+\log(1/\delta)]}{\epsilon^2}
\right).
\end{aligned}
\end{equation*}
The first bound estimates all $2(d-1)$ axial strings, which span every
Pauli weight; each target has dimension-independent variance, and the factor
$n$ is only logarithmic in the number of targets.  The second bound controls
all $4^n-1$ nonidentity Paulis: $M_d=\Theta(d)$ is the mixed-sector
worst-state cost and $\log(4^n)=\Theta(n)$ controls the full target set.  The
same fixed-measurement records support either task, and the queried Paulis may
be chosen after acquisition.  By contrast, a single Pauli specified in
advance can be measured directly in its eigenbasis using
$O(\log(1/\delta)/\epsilon^2)$ copies.

For the complete Pauli family, global-Clifford shadows have
$\mathbb E_\rho[\widehat P^2]=d+1$ for every nonidentity $P$ and hence the
same $O(d[n+\log(1/\delta)]/\epsilon^2)$ scaling \cite{Huang2020}.  A
complete stabilizer MUB realizes this channel using the minimum $d+1$
Clifford bases \cite{Zhang2024,WangCui2024}.  The statistical scaling is
therefore comparable, but the architectures differ: Clifford shadows sample
a global stabilizer basis on each shot, and the MUB implementation selects
among $d+1$ such bases, whereas the present scheme repeats one fixed
minimal-IC $d^2$-outcome POVM and records one $2n$-bit outcome.

The factor $d$ is unavoidable in the single-copy, no-quantum-memory model.
Even an adaptive protocol that chooses an arbitrary POVM on each fresh copy
requires $\Omega(d/\epsilon^2)$ copies to estimate the complete Pauli family
with constant success probability \cite{ChenCotler2021,ChenGongYe2024}.
Thus the fixed-WH bound has the optimal dependence on dimension, apart from
the $\Theta(n)$ target-set logarithm in the upper bound.  Allowing joint
two-copy measurements changes the model: data-dependent multistage protocols
can remove the exponential dependence on $n$, with $O(n/\epsilon^4)$ copies
information-theoretically and
$O(n\log(n/\epsilon)/\epsilon^4)$ copies for a triply efficient Clifford
implementation \cite{ChenGongYe2024,King2025}.

These formulas concern individual Pauli estimators.  For a Pauli sum
$O=\sum_v c_vP_v$, cross-covariances also contribute, so the individual costs
cannot simply be added; general observables are covered by
Sec.~\ref{sec:observable-estimation}.  Informational completeness also permits
full tomography from the same records.

\section{Full tomography from the same records}
\label{sec:tomography}

Direct fidelity estimation against a known ideal target
\cite{FlammiaLiu2011,daSilva2011} and cross-platform overlap estimation
using randomized local measurements \cite{ElbenVerification2020} avoid
full reconstruction.  For full tomography, sample requirements depend on
the reconstruction loss and measurement model \cite{Haah2017}; here we
evaluate canonical single-copy linear inversion under Hilbert--Schmidt loss.

Because the fixed POVM is informationally complete, the same records also
give the unbiased linear estimator in Eq.~\eqref{eq:linear-tomography}.  WH
covariance makes all canonical snapshots unitarily equivalent and hence gives
an exact finite-sample Frobenius error.  We compare it below with uniformly
randomized local-Pauli tomography as well as the fixed product SIC.

\begin{corollary}[Exact linear-tomography error]
\label{thm:tomography-mse}
For any density operator $\rho$ and $N$ independent and identically
distributed (i.i.d.) outcomes obtained by applying
the fixed circuit to identically prepared copies of $\rho$,
\begin{equation}
\mathbb E\widehat\rho_N=\rho,
\qquad
\mathbb E\|\widehat\rho_N-\rho\|_F^2
=\frac{H_n-\Tr(\rho^2)}N,
\label{eq:tomography-mse}
\end{equation}
where the squared snapshot norm is independent of the outcome and, for
$n\geq3$, has the exact value
\begin{align}
H_n&:=\Tr(\widehat\rho_u^2)
=\frac1d\left(1+\sum_{v\ne0}\frac d{\lambda_v}\right),\nonumber\\
&=\frac1d\left[1+2(d-1)A_d+(d-1)^2M_d\right].
\label{eq:snapshot-hs-norm}
\end{align}
The spectral floor gives
\begin{equation}
H_n\leq2(d^2+d-1)-\frac1d=O(d^2).
\label{eq:snapshot-hs-bound}
\end{equation}
Together with $M_d=\Theta(d)$, the exact expression shows
$H_n=\Theta(d^2)$, so the Hilbert--Schmidt MSE is
$\Theta(4^n/N)$.
\end{corollary}

Canonical linear inversion gives
\begin{equation}
H_{\rm SIC}=H_{\rm Cl}=d^2+d-1,
\qquad
H_{\rm Pauli}=H_{\rm prod}=5^n,
\label{eq:tomography-baselines}
\end{equation}
for an ideal global SIC, global Clifford, uniformly randomized local-Pauli
acquisition, and the $n$-fold product tetrahedral SIC, respectively
\cite{Huang2020,Renes2004,Scott2006}.  The last equality follows because both
local schemes use one-qubit canonical factors of the form $3\Pi-\id$, whose
squared Hilbert--Schmidt norm is $5$.  Since
$H_n<2H_{\rm SIC}$ and $H_n<H_{\rm prod}$ for every $n\geq3$, the present
coefficient stays within a factor of two of the $\Theta(4^n)$ global benchmark
and below both local $5^n$ coefficients.  More explicitly,
\begin{equation*}
\frac{H_n}{H_{\rm Pauli}}
=\frac{H_n}{H_{\rm prod}}
=\frac{H_n}{5^n}
\sim 2\left(\frac45\right)^n.
\end{equation*}
Thus full tomography remains exponential, but the canonical sample-coefficient
advantage over either local baseline grows asymptotically as
$\tfrac12(5/4)^n$.  Table~\ref{tab:tomography-contexts}
separates this statistical comparison from measurement-context complexity.
The Pauli ensemble has $3^n$ joint settings and $6^n$ possible
setting--outcome records, although each shot still requires only $n$ local
basis choices.  The finite-$n$ formulas are proved in
Appendix~\ref{app:tomography-proof}.

\begin{center}
\begin{minipage}{\columnwidth}
\captionof{table}{Canonical i.i.d. linear tomography.  The coefficient $H$
enters $\mathbb E\|\widehat\rho_N-\rho\|_F^2=
(H-\Tr\rho^2)/N$.  A Pauli label includes the sampled setting and its Born
outcome; the two fixed minimal-IC measurements return the outcome alone.}
\label{tab:tomography-contexts}
\footnotesize
\setlength{\tabcolsep}{3.5pt}
\renewcommand{\arraystretch}{1.08}
\begin{tabular*}{\columnwidth}{@{\extracolsep{\fill}}lccc@{}}
\toprule
acquisition & settings & raw labels & $H$\\
\midrule
random Pauli & $3^n$ random & $6^n$ & $5^n$\\
product SIC & $1$ & $4^n$ & $5^n$\\
fixed WH & $1$ & $4^n$ & $H_n\sim2\cdot4^n$\\
\bottomrule
\end{tabular*}

\vspace{3pt}
\begin{tabular*}{\columnwidth}{@{\extracolsep{\fill}}rccc@{}}
\toprule
$n$ & Pauli / product SIC & present $H_n$ & ideal global SIC\\
\midrule
$1$ & $5$ & $5$ & $5$\\
$2$ & $25$ & $25$ & $19$\\
$3$ & $125$ & $71$ & $71$\\
$4$ & $625$ & $317.19$ & $271$\\
$5$ & $3125$ & $1482.11$ & $1055$\\
\bottomrule
\end{tabular*}
\end{minipage}
\end{center}

Explicitly forming $\widehat\rho_N$ requires $\Theta(d^2)=\Theta(4^n)$
storage.  Each individual snapshot, however, has a compact algebraic
representation that can be contracted directly with structured targets
without constructing a dense matrix.

\section{Structured snapshots and efficient postprocessing}
\label{sec:reconstruction}

Computational-basis sparsity provides one route to efficient qudit-shadow
postprocessing when direct matrix-element access and the target trace are
available \cite{WangDDB2025}.

Equations~\eqref{eq:fixed-dual-covariance},
\eqref{eq:explicit-fixed-dual}, and
\eqref{eq:product-branch-conjugation} turn the algebraic inverse into an
efficient per-record rule: every snapshot is a sum of four product-state
outer products and the identity, without constructing a dense \(d\times d\)
matrix.

This representation has two immediate consequences for exact single-record
postprocessing.  First, the four rank-one terms are outer products between
product stabilizer states.  For a
Pauli string \(P_v\), \(\Tr(P_v\widehat\rho_u)\) is computed directly from the
binary parity in Eq.~\eqref{eq:exact-pauli-second-moment} in \(O(n)\) time,
and a \(K\)-term Pauli sum costs \(O(Kn)\).  For a stabilizer-state projector,
the four rank-one terms require only a constant number of phase-sensitive
stabilizer overlaps, while the identity contribution is fixed by the trace.
Standard overlap algorithms therefore give the conservative costs \(O(n^3)\) and
\(O(Kn^3)\) for one projector and a \(K\)-term stabilizer-projector
decomposition, respectively
\cite{AaronsonGottesman2004,GarciaMarkovCross2014}.

These Pauli and stabilizer-decomposition target classes also admit exact
polynomial-time evaluation for global-Clifford shadows: a single-record
snapshot is \((d+1)\ketbra{s}{s}-\id\), where \(\ket{s}\) is a stabilizer
state, so tableau algorithms apply
\cite{Huang2020,AaronsonGottesman2004}.  The costs listed below are for the
present snapshots.

Second, the same representation is a sum of five product operators: the four
rank-one outer products and the identity.  Each term is a bond-one MPO, which
yields the uniform tensor-network guarantee below.

\begin{proposition}[Constant-bond snapshot evaluation]
\label{prop:structured-postprocessing}
For every \(n\geq3\), the canonical snapshot has an exact MPO representation
of bond dimension at most five.  If \(O\) is given exactly as an open-boundary
MPO of bond dimension \(\chi\), then the single-record value
\(\Tr(O\widehat\rho_u)\) can be computed exactly using \(O(n\chi^2)\)
arithmetic operations and \(O(\chi)\) additional working memory beyond
storage of the MPO tensors.
\end{proposition}
\begin{proof}
Each of the five product operators is a bond-one MPO.  Taking direct sums
of their virtual bonds gives an exact MPO of bond dimension at most five.
Contracting each term separately with the target MPO from left to right
propagates a boundary vector of dimension at most \(\chi\).  Each site costs
\(O(\chi^2)\) arithmetic operations and \(O(\chi)\) working memory.
Summing the five contractions therefore costs \(O(n\chi^2)\) operations
and \(O(\chi)\) additional working memory.
\end{proof}

Table~\ref{tab:record-cost} summarizes both postprocessing routes.  In every
row the dense canonical snapshot is avoided.

\begin{center}
\begin{minipage}{\columnwidth}
\captionof{table}{Classical cost per record after receiving the \(2n\)-bit
outcome.  The dense canonical snapshot is never constructed.}
\label{tab:record-cost}
\begin{ruledtabular}
\begin{tabular}{@{}ll@{}}
target representation & per-record cost\\
\hline
Pauli string & \(O(n)\)\\
\(K\)-term Pauli sum & \(O(Kn)\)\\
stabilizer-state projector & \(O(n^3)\)\\
\(K\)-term stabilizer-projector sum & \(O(Kn^3)\)\\
bond-\(\chi\) MPO & \(O(n\chi^2)\)\\
\end{tabular}
\end{ruledtabular}
\end{minipage}
\end{center}

The MPO row has broader scope: it applies to every bond-\(\chi\) target,
including those without a short Pauli or stabilizer decomposition.  General
global-Clifford snapshots have no analogous uniform guarantee.  Computing
certain probabilities
\(p=|\langle\phi_1\otimes\cdots\otimes\phi_n|G\rangle|^2\), where \(\ket G\)
is a stabilizer graph state, is \(\#\mathrm{P}\)-hard
\cite{Ghosh2023Complexity}.  The product-state projector
\(O_\phi=\bigotimes_j\ketbra{\phi_j}{\phi_j}\) is a bond-one MPO.
Substituting the global-Clifford snapshot
\(\widehat\rho_G=(d+1)\ketbra{G}{G}-\id\) gives
\(\Tr(O_\phi\widehat\rho_G)=(d+1)p-1\).  Thus computing this quantity exactly
is \(\#\mathrm{P}\)-hard in the worst case even for bond-one targets.  By
contrast, Proposition~\ref{prop:structured-postprocessing} gives a uniform
\(O(n\chi^2)\) algorithm for every bond-\(\chi\) MPO target here.

We now give the fixed quantum circuit that generates these records.

\FloatBarrier
\section{From the Pauli--WH orbit to one fixed circuit}
\label{sec:bell}

Let \(S\) carry the unknown \(n\)-qubit state, and let \(A\) be an \(n\)-qubit
ancilla initialized in \(\ket{0^n}\).  Every shot applies the same
premeasurement unitary to \(SA\), followed by computational-basis measurement
of all \(2n\) qubits.  The unitary has two stages.  First, a fiducial-dependent
preparation acting only on \(A\) maps \(\ket{0^n}_A\) to the conjugate
fiducial \(\ket{\phi_n^*}_A\).  Second, a fiducial-independent Bell-basis
transform factorizes over the matched pairs \((S_j,A_j)\).  With the
convention used below, the readouts of \(S\) and \(A\) are \(\bm a\) and
\(\bm b\), respectively, so the \(2n\)-bit outcome directly labels the
Pauli--WH effect \(E_{\bm a,\bm b}^{(n)}\).  Thus neither the measurement
setting nor the circuit changes between shots.

The common Bell construction below applies to an arbitrary normalized fiducial
and realizes the corresponding Pauli--WH POVM.  The only fiducial-dependent
stage is the exact preparation of \(\ket{\phi_n^*}\), constructed and costed in
Sec.~\ref{sec:fiducial-preparation}.

\subsection{Generic Bell interface}

Let \(\ket\phi\in\C^d\) be any normalized fiducial.  The premeasurement
unitary contains the following two ordered blocks.

\emph{(i) Fiducial preparation.}
Choose any unitary satisfying

\begin{equation}
U_{\rm prep}(\phi^*)\ket{0^n}_A=\ket{\phi^*}_A,
\label{eq:prep-def}
\end{equation}

where complex conjugation is taken in the computational basis.

\emph{(ii) Bell-basis transform.}
Independently of \(\ket\phi\), define

\begin{equation}
U_{\rm Bell}=\left(\bigotimes_{j=1}^nH_{S_j}\right)
\left(\prod_{j=1}^n\operatorname{CNOT}_{S_j\to A_j}\right),
\label{eq:bell}
\end{equation}

The matched CNOTs act on disjoint pairs.  With direct $S_j$--$A_j$ couplers,
they form one parallel layer, followed by one parallel layer of Hadamards.
For fixed \(\phi\), the complete premeasurement unitary is the
ordered composition of blocks (i) and (ii):

\begin{equation}
U_{\rm IC}(\phi)=
U_{\rm Bell}\bigl[\id_S\otimes U_{\rm prep}(\phi^*)\bigr].
\label{eq:uic}
\end{equation}

For the fiducial family in Eq.~\eqref{eq:fid}, we write
\(U_{\rm IC}^{(n)}:=U_{\rm IC}(\phi_n)\), as in Fig.~\ref{fig:main}.  The
following proposition identifies the POVM induced by computational-basis
readout after \(U_{\rm IC}(\phi)\).  It specializes the standard
fiducial-state-plus-generalized-Bell construction
\cite{DAriano2004,GuptaWeiss2025} to the tensor-Pauli convention used here.

\begin{proposition}[Fixed-circuit realization of a Pauli--WH POVM]
\label{prop:generic-bell}
Let $\ket\phi\in\C^d$ be normalized.  Initialize $S$ in the state $\rho$ and
$A$ in $\ket{0^n}$.  After applying $U_{\rm IC}(\phi)$ and measuring all
$2n$ qubits in the computational basis, let $\bm a$ and $\bm b$ denote the
outcomes on $S$ and $A$, respectively.  Then
\begin{equation}
p(\bm a,\bm b\mid\rho)
=\frac1d\Tr\!\left[
\rho D_{\bm a,\bm b}\ketbra\phi\phi
D_{\bm a,\bm b}^\dagger\right].
\label{eq:born-uic}
\end{equation}
For $\phi=\phi_n$, Eq.~\eqref{eq:fixed-povm} reduces this expression to
$p(\bm a,\bm b\mid\rho)=\Tr[\rho E_{\bm a,\bm b}^{(n)}]$.  Thus a single
fixed premeasurement unitary followed by computational-basis readout
implements the corresponding $d^2=4^n$-outcome Pauli--WH POVM.
\end{proposition}

Appendix~\ref{app:bell-proof} explicitly maps the two-register readout
\((\bm a,\bm b)\) to the WH displacement \(D_{\bm a,\bm b}\) and proves
Eq.~\eqref{eq:born-uic} for an arbitrary input state \(\rho\).  Since
\(U_{\rm Bell}\) is fiducial independent, any exact circuit
preparing \(\ket{\phi^*}\) realizes the same POVM and reconstruction map;
different choices affect only the preparation resources, such as gate count,
depth, connectivity, and ancillary workspace.

\FloatBarrier
\section{Exact fiducial preparation and resource tradeoffs}
\label{sec:fiducial-preparation}

It remains to prepare the known ancilla state \(\ket{\phi_n^*}\).  For the two
unitary routes, the construction is easiest
to read in three layers: (i) encode the two product branches in one root
qubit, (ii) expand that one-qubit memory along a path or a tree, and (iii)
attach the common Bell analyzer.  The first two layers prepare
\(\ket{\phi_n}\).  Complex conjugation only sends
\(\theta_n\mapsto-\theta_n\), so the same topology and resource counts prepare
\(\ket{\phi_n^*}\); below we therefore analyze \(\ket{\phi_n}\) and postpone
the already-fixed Bell block.

There is only one parameterized local rule in the expansion: a two-qubit
split turns one logical memory into two smaller memories while preserving the
coherent branch label.  Each memory occupies one physical output qubit; its
label (m) records the number of output factors still encoded, not the number
of qubits occupied by the memory.  A nearest-neighbor path applies this rule
serially, whereas a balanced tree applies disjoint splits in parallel.
Figure~\ref{fig:memory-schedules} puts the rule and both schedules side by
side.  The optional measurement-assisted LCU route uses a different control
model and admits both heralded and deterministic adaptive realizations.
Table~\ref{tab:hierarchy} compares these resource points; each is followed by
the same Bell-analysis block.

\subsection{Exact preparation by binary memory splitting}

\subsubsection{Nearest-neighbor path}

The path construction introduces no separate memory qubit.  Before step $j$,
qubits $1,\ldots,j-1$ already have their final values in each coherent branch,
qubit $j$ of $A$ carries the encoding of the two product tails that remain,
and all qubits to its right are in $\ket0$.  The adjacent unitary on
$(j,j+1)$ fixes qubit $j$ as the next $\ket+$ or $\ket0$ output and transfers
the shortened encoding to qubit $j+1$, which is also an output qubit of $A$.
Linearity preserves the relative coefficient $e^{i\theta_n}$.  Thus all
intermediate information stays within $A$, and no work qubit is used.

\begin{theorem}[Ancilla-free nearest-neighbor preparation]
\label{thm:sequential}
For every $n\geq2$, there exist an explicitly preparable one-qubit state
$\ket{\eta_n}$ and explicitly constructible adjacent two-qubit unitaries
$G_j$, acting on qubits $(j,j+1)$, such that
\begin{equation}
G_{n-1}\cdots G_2G_1
\bigl(\ket{\eta_n}_1\ket{0^{n-1}}_{2,\ldots,n}\bigr)
=\ket{\phi_n}.
\label{eq:sequential}
\end{equation}
The circuit uses one single-qubit initialization and $n-1$ adjacent two-qubit
gates, with two-qubit depth $n-1$ and no work qubits, measurements, or
feedforward.  Compiling the initialization and first split jointly gives
$1+2(n-2)=2n-3$ CNOTs and the same CNOT depth.
\end{theorem}

\begin{proof}
\emph{Memory encoding and one path step.}
For $m\geq1$, define the one-qubit memory states representing the two
unfinished $m$-factor product tails by
\begin{equation}
\begin{aligned}
c_m&:=2^{-m/2},& r_m&:=\sqrt{1-c_m^2},\\
\ket{\mu_0^{(m)}}&:=\ket0,&
\ket{\mu_+^{(m)}}&:=c_m\ket0+r_m\ket1.
\end{aligned}
\label{eq:memory-states}
\end{equation}
Here $m$ is the number of output factors still to be generated.  The two
encoding states have overlap $c_m$, exactly matching the overlap of the product
tails $\ket{0^{\otimes m}}$ and $\ket{+^{\otimes m}}$.
First prepare on qubit $1$
\begin{equation}
\ket{\eta_n}=\frac{\ket{\mu_+^{(n)}}+
e^{i\theta_n}\ket{\mu_0^{(n)}}}{\sqrt{\mathcal N_n}}.
\label{eq:moving-memory-initial}
\end{equation}
By the definition of $\mathcal N_n$ in Eq.~\eqref{eq:fid},
$\ket{\eta_n}$ is normalized.

For $m\geq2$, prescribe one shortening step by
\begin{align}
U_m\ket{\mu_0^{(m)}}\ket0
&=\ket{\mu_0^{(1)}}\ket{\mu_0^{(m-1)}},\nonumber\\
U_m\ket{\mu_+^{(m)}}\ket0
&=\ket{\mu_+^{(1)}}\ket{\mu_+^{(m-1)}}.
\label{eq:moving-memory-isometry}
\end{align}
Since $\ket{\mu_0^{(1)}}=\ket0$ and
$\ket{\mu_+^{(1)}}=\ket+$, this step writes one physical output factor and
passes the shortened encoding to the next qubit of $A$.  The input overlap is
$c_m$, and the output overlap is
$2^{-1/2}c_{m-1}=c_m$.  The two pairs therefore have the same Gram matrix, so
the prescribed map is an isometry on their span and extends to a two-qubit
unitary.  Appendix~\ref{app:path-block-compilation} gives an explicit unitary
completion.
At chronological step $j$, the unfinished tail contains $m=n-j+1$ output
factors; define the physical gate on sites $(j,j+1)$ by
\begin{equation*}
G_j:=U_{n-j+1}^{(j,j+1)},\qquad j=1,\ldots,n-1.
\end{equation*}

For the path, applying $G_1=U_n^{(1,2)}$ produces the root state below;
we suppress the untouched factor $\ket{0^{n-2}}$ on qubits $3,\ldots,n$:
\begin{equation}
\begin{aligned}
\ket{\xi_{1,n-1}^{(n)}}
&:=G_1\bigl(\ket{\eta_n}_1\ket0_2\bigr)\\
&=\frac1{\sqrt{\mathcal N_n}}\Bigl(
\ket+_1\ket{\mu_+^{(n-1)}}_2\\[-2pt]
&\hspace{2.0em}
+e^{i\theta_n}\ket0_1\ket{\mu_0^{(n-1)}}_2\Bigr).
\end{aligned}
\label{eq:first-sequential-step}
\end{equation}
The first output factor is now fixed in both coherent branches, while qubit
$2$ carries the encoding of the remaining $n-1$ factors.

\emph{Loop invariant and termination.}
Iterating the same rule gives, after $k$ gates and for
$1\leq k\leq n-1$,
\begin{equation}
\begin{aligned}
\ket{\psi_k}=\frac{1}{\sqrt{\mathcal N_n}}\Bigl(&
\ket+^{\otimes k}\ket{\mu_+^{(n-k)}}\\[-2pt]
&+e^{i\theta_n}\ket0^{\otimes k}
\ket{\mu_0^{(n-k)}}\Bigr)\ket0^{\otimes(n-k-1)}.
\end{aligned}
\label{eq:moving-memory-invariant}
\end{equation}
The three displayed factors occupy, from left to right, qubits $1{:}k$,
$k+1$, and $k+2{:}n$; the final factor is absent when $k=n-1$.
Equation~\eqref{eq:first-sequential-step} is the case $k=1$.  Applying
$G_{k+1}=U_{n-k}^{(k+1,k+2)}$ and
Eq.~\eqref{eq:moving-memory-isometry} proves the step from $k$ to $k+1$ for
$1\leq k\leq n-2$.
For $k=n-1$, $\ket{\mu_+^{(1)}}=\ket+$ and
$\ket{\mu_0^{(1)}}=\ket0$.  The two completed components are therefore
$\ket+^{\otimes n}$ and $e^{i\theta_n}\ket0^{\otimes n}$, and their coherent
sum is Eq.~\eqref{eq:fid}, proving Eq.~\eqref{eq:sequential}.  The one-CNOT
root block in Appendix~\ref{app:path-block-compilation} directly prepares
Eq.~\eqref{eq:first-sequential-step} from $\ket{00}$; it replaces the separate
initialization and $G_1$.  The same appendix gives a two-CNOT completion of
each later split, yielding the CNOT count stated in the theorem.
\end{proof}

\subsubsection{Balanced interaction tree}

Here $m$ labels the unresolved block size, whereas
$\ket{\mu_s^{(m)}}$ occupies only one qubit.  The left-to-right path
repeatedly partitions this block size as $m=1+(m-1)$.  To expose parallelism,
choose instead $m=a+b$ with $a,b\geq1$ and apply $U_{a,b}$ to the memory and
an unused output qubit.  The result is two single-qubit memories associated
with child blocks of sizes $a$ and $b$.  Explicitly,
\begin{align}
U_{a,b}\ket{\mu_0^{(m)}}\ket0
&=\ket{\mu_0^{(a)}}\ket{\mu_0^{(b)}},\nonumber\\
U_{a,b}\ket{\mu_+^{(m)}}\ket0
&=\ket{\mu_+^{(a)}}\ket{\mu_+^{(b)}}.
\label{eq:binary-memory-splitting}
\end{align}
Indeed, the input overlap is $c_m$, while the output overlap is
$c_ac_b=c_{a+b}=c_m$; the two prescribed pairs have the same Gram matrix and
therefore define an isometry that extends to a two-qubit unitary.  The path is
the special choice $(a,b)=(1,m-1)$, which recovers
Eq.~\eqref{eq:moving-memory-isometry} with $U_m=U_{1,m-1}$.
Choosing $a$ and $b$ as nearly equal as possible gives the parallel schedule
below.

\begin{figure*}[t!]
\centering
\includegraphics[width=0.80\textwidth]{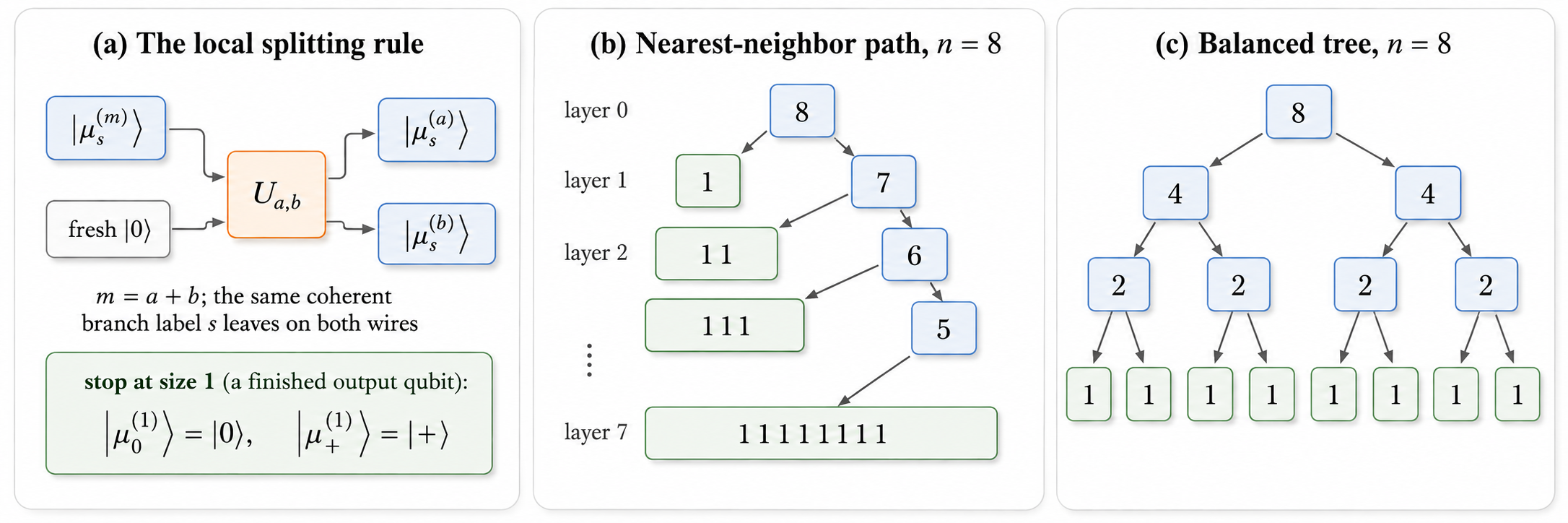}
\caption{Binary-memory splitting for \(n=8\).  Blue boxes carry one-qubit
encodings \(\ket{\mu_s^{(k)}}\); green boxes labeled \(1\) are completed
output qubits.
(a) \(U_{a,b}\) preserves \(s\) while splitting \(m=a+b\); ``fresh'' is an
unused output qubit of \(A\).  (b) Repeated \((1,m-1)\) splits form a
nearest-neighbor path; each green bar collects the outputs completed so far.
(c) Balanced splits on disjoint pairs run in parallel.  Both schedules use
only the \(n\) output qubits.}
\label{fig:memory-schedules}
\end{figure*}

\begin{theorem}[Ancilla-free balanced-tree preparation]
\label{thm:balanced-memory-tree}
With arbitrary one- and two-qubit gates and balanced-tree connectivity,
$\ket{\phi_n}$ can be prepared exactly on its $n$-qubit register using no
work qubits.  Initialize one of the $n$ output qubits in $\ket{\eta_n}$ from
Eq.~\eqref{eq:moving-memory-initial} and the other $n-1$ output qubits in
$\ket0$.  Whenever a qubit carries a one-qubit encoding for $m>1$ output
factors not yet written, choose
\begin{equation}
a=\lfloor m/2\rfloor,
\qquad b=\lceil m/2\rceil.
\label{eq:balanced-split-sizes}
\end{equation}
Apply $U_{a,b}$ to that qubit and one output qubit of $A$ still in $\ket0$.
After the split, the original wire carries the one-qubit memory for the
size-$a$ child block, while the previously unused output wire carries the
corresponding memory for the size-$b$ child block.  Perform all splits on
disjoint pairs in parallel and repeat until every unresolved block has size
one.  The circuit uses $n-1$ two-qubit gates and has two-qubit depth
$\lceil\log_2n\rceil$.
\end{theorem}

\begin{proof}
\emph{Recursive rule and correctness.}
We prove correctness by strong induction on the number $m$ of output factors
still to be generated.  Start with one qubit carrying the memory for an
unresolved block of size $m$ and $m-1$ output qubits in $\ket0$.  We claim
that the recursive splits implement the two branchwise maps
\begin{align}
\ket{\mu_0^{(m)}}\ket{0^{m-1}}&\longmapsto\ket{0^m},\nonumber\\
\ket{\mu_+^{(m)}}\ket{0^{m-1}}&\longmapsto\ket+^{\otimes m}.
\label{eq:balanced-tree-invariant}
\end{align}
For $m=1$, these are just
$\ket{\mu_0^{(1)}}=\ket0$ and $\ket{\mu_+^{(1)}}=\ket+$.
For $m>1$, choose $a$ and $b$ from
Eq.~\eqref{eq:balanced-split-sizes}.  Equation~\eqref{eq:binary-memory-splitting}
maps that memory to two single-qubit memories for child blocks of sizes $a$
and $b$.  Assign $a-1$ of the remaining zero qubits to the first child and
$b-1$ to the second.  The induction hypothesis expands these two disjoint
blocks into $a$ and $b$ completed factors, respectively, and the two
expansions can run in parallel.  This proves both maps above.

Equation~\eqref{eq:balanced-tree-invariant} describes the result after the
entire recursion has terminated.  Starting from
$\ket{\eta_n}\ket{0^{n-1}}$, the $\ket{\mu_+^{(n)}}$ component becomes
$\ket+^{\otimes n}$, while the $\ket{\mu_0^{(n)}}$ component becomes
$\ket{0^n}$.  Linearity preserves their relative coefficient, so the final
output is
\begin{equation}
\ket{\eta_n}\ket{0^{n-1}}
\xrightarrow{\text{all recursive splits}}
\frac{\ket+^{\otimes n}+e^{i\theta_n}\ket{0^n}}{\sqrt{\mathcal N_n}}
=\ket{\phi_n}.
\end{equation}

\emph{Gate count and depth.}
Each split replaces one encoding by two and uses one previously untouched
output qubit.  Starting with one unresolved block and ending with $n$
size-one blocks therefore requires $n-1$ splits.  All of these qubits belong
to $A$, so no work qubit is used.

Let $L_{\rm tree}(m)$ be the number of two-qubit layers needed to resolve a
block of size $m$.  The two child blocks are disjoint and can be expanded in
parallel, so
\begin{equation}
L_{\rm tree}(1)=0,\qquad
L_{\rm tree}(m)=1+\max\{L_{\rm tree}(a),L_{\rm tree}(b)\},
\end{equation}
for $m>1$.  Since $b=\lceil m/2\rceil\geq a$, strong induction gives
\begin{equation}
L_{\rm tree}(m)=1+\left\lceil\log_2 b\right\rceil
=\lceil\log_2m\rceil.
\end{equation}
Here we used
$1+\lceil\log_2\lceil m/2\rceil\rceil=\lceil\log_2m\rceil$, which also
holds when $m$ is not a power of two.
\end{proof}

\emph{Eight-qubit illustration.}
To make the recursion concrete, we now spell out the $n=8$ instance step by
step.
Label the eight output qubits of $A$ as $1,\ldots,8$.  All eight are present
from the start.  An ``unused'' qubit below means one of these output qubits
still in $\ket0$, not an additional work qubit.

\emph{Initialization.}
This is the $n=8$ instance of the one-qubit initialization in
Eq.~\eqref{eq:moving-memory-initial}.  Starting from $\ket{0^8}_A$, prepare
qubit $1$ in
\begin{equation*}
\ket{\eta_8}
=\frac{(1/16+e^{i\theta_8})\ket0
+(\sqrt{255}/16)\ket1}{\sqrt{\mathcal N_8}}.
\end{equation*}
A single-qubit gate suffices because $\ket{\eta_8}$ is normalized.  Leave the
other seven qubits in $\ket0$.  The initial state for the recursive splits is
therefore
\begin{equation}
\begin{aligned}
\ket{\Psi_0}:={}&\ket{\eta_8}_1\ket0_2\cdots\ket0_8\\
={}&\frac{1}{\sqrt{\mathcal N_8}}\Bigl(
\ket{\mu_+^{(8)}}_{1}
+e^{i\theta_8}\ket{\mu_0^{(8)}}_{1}\Bigr)\\[-2pt]
&{}\otimes\ket0_{2}\cdots\ket0_{8}.
\end{aligned}
\label{eq:eight-qubit-initial-state}
\end{equation}
Here both $\ket{\mu_+^{(8)}}$ and $\ket{\mu_0^{(8)}}$ are \emph{one-qubit}
memory states representing unresolved eight-factor product tails; neither is
an eight-qubit register.

\emph{Layer 1: $8\to4+4$.}
Use qubit $5$, which is still in $\ket0$, and apply $U_{4,4}$ to qubits
$(1,5)$.
Equation~\eqref{eq:binary-memory-splitting} acts on the two branches as
\begin{align*}
\ket{\mu_0^{(8)}}_{1}\ket0_{5}
&\longmapsto
\ket{\mu_0^{(4)}}_{1}\ket{\mu_0^{(4)}}_{5},\\
\ket{\mu_+^{(8)}}_{1}\ket0_{5}
&\longmapsto
\ket{\mu_+^{(4)}}_{1}\ket{\mu_+^{(4)}}_{5}.
\end{align*}
The other six output qubits remain in $\ket0$ and are suppressed below.
Applying these two rules to $\ket{\Psi_0}$ gives
\begin{equation}
\ket{\Psi_1}=\frac{1}{\sqrt{\mathcal N_8}}\Bigl(
\ket{\mu_+^{(4)}}_{1}\ket{\mu_+^{(4)}}_{5}
+e^{i\theta_8}
\ket{\mu_0^{(4)}}_{1}\ket{\mu_0^{(4)}}_{5}\Bigr).
\label{eq:eight-qubit-layer-one}
\end{equation}

\emph{Layer 2: each $4\to2+2$.}
Apply $U_{2,2}$ simultaneously to qubit pairs $(1,3)$ and $(5,7)$.  Thus each
size-four memory is split using a distinct output qubit still in $\ket0$,
while the other four output qubits remain in $\ket0$ and are suppressed.  The
result is
\begin{equation}
\begin{aligned}
\ket{\Psi_2}=\frac{1}{\sqrt{\mathcal N_8}}\Bigl(&
\bigotimes_{j\in\{1,3,5,7\}}\ket{\mu_+^{(2)}}_{j}\\[-2pt]
&+e^{i\theta_8}
\bigotimes_{j\in\{1,3,5,7\}}\ket{\mu_0^{(2)}}_{j}\Bigr).
\end{aligned}
\label{eq:eight-qubit-layer-two}
\end{equation}

\emph{Layer 3: each $2\to1+1$.}
Finally apply $U_{1,1}$ in parallel to
qubit pairs $(1,2)$, $(3,4)$, $(5,6)$, and $(7,8)$.  Every qubit now contains
one completed factor:
\begin{align}
\ket{\Psi_3}
&=\frac{1}{\sqrt{\mathcal N_8}}\left(
\bigotimes_{j=1}^{8}\ket{\mu_+^{(1)}}_{j}
+e^{i\theta_8}
\bigotimes_{j=1}^{8}\ket{\mu_0^{(1)}}_{j}\right)
\nonumber\\
&=\frac{\ket+^{\otimes8}+e^{i\theta_8}\ket0^{\otimes8}}
{\sqrt{\mathcal N_8}}
=\ket{\phi_8}.
\label{eq:eight-qubit-balanced-example}
\end{align}
The block sizes have therefore evolved as
$8\to4+4\to2+2+2+2\to1+\cdots+1$.  Each layer acts linearly on the two
coherent terms, so their normalization and relative phase are preserved.
The preparation consists of one single-qubit initialization followed by
$1+2+4=7$ two-qubit gates in three two-qubit layers.  Thus its
two-qubit depth is three, and it uses no qubit outside $A$.
This is precisely the $n=8$ instance of the recursive construction proved
above.

\subsubsection{Compilation and common CNOT ledger}

The path and balanced-tree constructions above act entirely within the
$n$-qubit output register $A$.  Neither invokes a fanout primitive, and
neither uses clean or dirty work qubits beyond $A$.  At each split, the input
qubit in $\ket0$ is one of the final output qubits of $A$ that has not yet
been written; after the split it remains part of the output.  Thus the tree
hides no ancillary workspace.  It remains to compile the abstract two-qubit
splits into CNOTs and one-qubit gates.

At the abstract level, either schedule contains $n-1$ split gates.  Compiling
all of them independently as one-to-two-qubit isometries would use
$2(n-1)$ CNOTs \cite{Iten2016}.  Appendix~\ref{app:path-block-compilation}
gives one explicit two-CNOT template for every $U_{a,b}$, including balanced
splits with $a,b>1$.  The root is cheaper: before it, the input is
one known product state and is not yet correlated with any completed output.
We therefore fuse the initialization and first split and compile their fixed
two-qubit output state, rather than the full root unitary.  Let
$(a,b)=(1,n-1)$ for the path and
$(\lfloor n/2\rfloor,\lceil n/2\rceil)$ for the balanced tree.  In either
case, write their joint target as
\begin{equation}
\ket{\xi_{a,b}^{(n)}}:=\frac{
\ket{\mu_+^{(a)}}\ket{\mu_+^{(b)}}
+e^{i\theta_n}\ket{00}}{\sqrt{\mathcal N_n}}.
\label{eq:root-two-qubit-state}
\end{equation}
Its coefficient-matrix determinant is
$e^{i\theta_n}r_ar_b/\mathcal N_n\neq0$, so $\ket{\xi_{a,b}^{(n)}}$ is entangled.
Appendix~\ref{app:path-block-compilation} constructs its direct preparation
from $\ket{00}$ using exactly one CNOT and one-qubit gates.  This root block
replaces jointly---it does not follow---the initialization and first split.
For $n\geq2$, the remaining $n-2$ nonroot splits use two CNOTs each, giving
\begin{equation}
C_{\rm prep}
=\underbrace{1}_{\substack{\text{initialization and}\\\text{root split}}}
+\underbrace{2(n-2)}_{\substack{\text{$n-2$ nonroot splits,}\\
\text{two CNOTs each}}}
=2n-3.
\label{eq:memory-prep-cnot-count}
\end{equation}
For the path, these CNOTs are sequential, so the CNOT depth is $2n-3$.
For the balanced tree, let $L=\lceil\log_2n\rceil$.  The joint first block
has CNOT depth one, and each of the remaining $L-1$ tree levels becomes two
parallel CNOT layers.  Its CNOT depth is therefore
$1+2(L-1)=2L-1$.  For $n=1$, no preparation CNOT is needed.

\subsubsection{Why the unitary schedules are optimal}

\begin{proposition}[Optimal two-qubit-gate count]
\label{prop:balanced-gate-optimality}
For $n\geq2$, every unitary circuit that starts from a product state and
prepares $\ket{\phi_n}$ exactly requires at least $n-1$ two-qubit gates, even
with any number of additional qubits initialized in a product state.  The
two-qubit-gate count in
Theorem~\ref{thm:balanced-memory-tree} is therefore optimal.
\end{proposition}

\begin{proof}
Across every nontrivial bipartition $S\mid\bar S$, the two branch vectors
$\ket+_S^{\otimes|S|},\ket0_S^{\otimes|S|}$ are linearly independent, as are
their counterparts on $\bar S$.  Hence $\ket{\phi_n}$ has Schmidt rank two
across every such cut.  In the interaction graph whose edges are the
two-qubit gates, all $n$ output qubits must therefore lie in one connected
component; otherwise the output would factor across a nontrivial cut.  A
connected graph containing $n$ specified vertices has at least $n-1$ edges.
\end{proof}

\begin{proposition}[Logarithmic unitary-depth lower bound]
\label{prop:balanced-depth-lower-bound}
For $n\geq2$, consider circuits starting from a product state on the output
and work registers, with arbitrary one-qubit gates, layers of disjoint
two-qubit gates, all-to-all connectivity, and any number of work qubits.
Exact preparation of $\ket{\phi_n}$ requires two-qubit depth
\begin{equation}
D\geq\frac12\log_2n.
\label{eq:depth-lower-bound}
\end{equation}
The balanced tree is therefore asymptotically depth optimal.
\end{proposition}

\begin{proof}
The key observation is that the target $\ket{\phi_n}$ has nonzero connected $Z$ correlation
between every pair of output qubits.  To verify this, fix two distinct outputs
$i$ and $j$, let $Z_i$ denote Pauli $Z$ on output $i$, and set
$x_n=2c_n\cos\theta_n$ and $\nu_n=(1+x_n)/(2+x_n)$.  Direct evaluation gives
\begin{align}
\langle Z_i\rangle_{\phi_n}
&=\langle Z_iZ_j\rangle_{\phi_n}=\nu_n,\nonumber\\
\operatorname{Cov}_{\phi_n}(Z_i,Z_j)
&=\nu_n(1-\nu_n)>0.
\label{eq:all-pair-Z-covariance}
\end{align}
For $n\geq2$, the phases in Eq.~\eqref{eq:fid} give
$x_n\in[-1/2,0)$ and hence $\nu_n\in[1/3,1/2)$.
We now turn this all-pairs correlation into a light-cone bound.  Let
$\mathcal B_i$ be the backward input light cone of output $i$.  At depth $D$,
$|\mathcal B_i|\leq2^D$.  If $\mathcal B_i$ and $\mathcal B_j$ were disjoint,
the product input would give zero connected correlation; hence
$\mathcal B_i\cap\mathcal B_j\neq\varnothing$ for every pair.  Let
$\mathcal F_a$ denote the forward output light cone of input wire $a$.  Fix
$i$.  Every output $j$ shares some ancestor $a\in\mathcal B_i$ with $i$, so
$j\in\mathcal F_a$.  Since $|\mathcal F_a|\leq2^D$, all $n$ outputs obey
\begin{equation}
n\leq\sum_{a\in\mathcal B_i}|\mathcal F_a|\leq4^D.
\end{equation}
This proves Eq.~\eqref{eq:depth-lower-bound}.
\end{proof}

\subsubsection{Small-system endpoints}

For $n=1$, the complete tetrahedral-SIC circuit uses one arbitrary preparation
gate
$U_{\rm prep}(\phi_1^*)$, one Hadamard, and one CNOT.  One CNOT is minimal for
a deterministic implementation of a qubit SIC POVM using one ancilla qubit
and computational-basis readout of both qubits \cite{You2025}.

For $n=3$, the ancilla-free CNOT complexity of preparing the Hoggar fiducial is
exactly three, assuming arbitrary one-qubit gates and all-to-all CNOT
connectivity.  Figure~\ref{fig:hoggar-circuit} gives a three-CNOT construction,
and Appendix~\ref{app:hoggar} proves that no circuit with at most two CNOTs can
prepare the state.  Adding one Bell-analysis CNOT for each of the three
system--ancilla pairs gives a six-CNOT implementation of the three-qubit
Hoggar SIC POVM.

\subsection{Optional measurement-assisted LCU}
\label{subsec:lcu}

The balanced tree already gives the strongest measurement-free resource point
used in this work.  Mid-circuit measurement and feedforward have also been
used for dynamic generalized measurements and constant-depth long-range
entangling gates \cite{Ivashkov2024,Baumer2025}.  Related measurement-assisted
constructions prepare arbitrary states at constant quantum depth with
exponential auxiliary resources \cite{Zi2025}, or exploit permutation
symmetry \cite{Luo2026} and suitable matrix-product-state structure
\cite{Smith2024,Li2026}.  Here they allow the two product
branches to be combined by a short heralded LCU procedure.  Following the
two-term LCU idea \cite{ChildsWiebe2012}, a
branch-label qubit coherently selects $H^{\otimes n}$ or
$e^{i\theta_n}\id$; interference produces the target in one branch:
\begin{equation}
\frac{H^{\otimes n}+e^{i\theta_n}\id}{2}\ket{0^n}
=\frac{\sqrt{\mathcal N_n}}{2}\ket{\phi_n}
=\sqrt{p_n}\ket{\phi_n}.
\label{eq:two-term-lcu}
\end{equation}
Here $p_n=\mathcal N_n/4$ is the probability of that branch.

For the coherent reference implementation, let $A$ be the $n$-qubit output
register, $F$ a single flag qubit carrying the branch label, and $W$ an
$(n-1)$-qubit fanout register.  Write
$\ket{\mathbf b}_W:=\ket b^{\otimes(n-1)}$ for $b\in\{0,1\}$.  For $n=1$,
$W$ is empty and the fanout unitary below is the identity.

\begin{figure}[t]
\centering
\definecolor{lcuBlue}{HTML}{356D9B}
\definecolor{lcuGold}{HTML}{B77A24}
\definecolor{lcuGreen}{HTML}{3E7C59}
\begin{tikzpicture}[
x=1cm,y=1cm,font=\scriptsize,>=Latex,
lcustep/.style={draw=lcuBlue!75,fill=lcuBlue!5,rounded corners=2pt,
minimum width=7.25cm,text width=6.75cm,minimum height=7.5mm,align=center,
inner xsep=3pt,inner ysep=2pt,line width=0.5pt},
control/.style={lcustep,draw=lcuGold!80!black,fill=lcuGold!6},
success/.style={lcustep,draw=lcuGreen!80!black,fill=lcuGreen!6},
arrow/.style={-{Latex[length=1.35mm,width=0.95mm]},draw=black!65,
line width=0.55pt,shorten <=1pt,shorten >=1pt}
]
\node[lcustep] (s1) at (0,0)
{1. \(U_{\rm flag}(\theta_n)\) on \(F\): create the two branch amplitudes};
\node[control] (s2) at (0,-1.02)
{2. \(U_{\rm fan}\) on \(FW\): distribute the classical-basis branch label};
\node[lcustep] (s3) at (0,-2.04)
{3. \(U_{\rm sel}\): write \(\ket+^{\otimes n}\) or \(\ket{0^n}\) on \(A\)};
\node[control] (s4) at (0,-3.06)
{4. \(U_{\rm fan}^{\dagger}\): return every workspace qubit in \(W\) to \(\ket0\)};
\node[success] (s5) at (0,-4.08)
{5. \(H_F\), then measure \(F\):\quad
\(F=0\Rightarrow\ket{\phi_n}\) \((p_n\geq3/8)\);\quad
\(F=1\Rightarrow\) discard and retry};
\draw[arrow] (s1) -- (s2);
\draw[arrow] (s2) -- (s3);
\draw[arrow] (s3) -- (s4);
\draw[arrow] (s4) -- (s5);
\end{tikzpicture}
\caption{Coherent reference decomposition of the optional heralded LCU
preparation, read from top to bottom.  The adaptive realization below compiles
steps 2--4 into one exact dynamic common-control primitive.  In the coherent
decomposition, the fanout workspace carries the branch label only between
steps 2 and 4 and is reset before the flag is measured.  The unknown state has
not yet interacted with \(A\), so failed ancilla preparations can be discarded
without consuming a copy of \(\rho\).}
\label{fig:lcu-flow}
\end{figure}
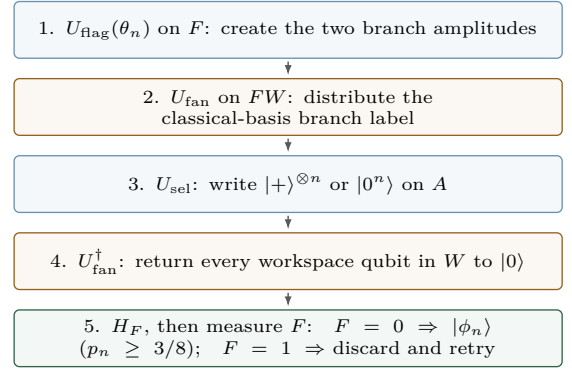

\begin{lemma}[Constant-success branch synthesis]
\label{lem:branch-synthesis}
For every $n\geq1$, with arbitrary one-qubit gates and all-to-all CNOT
connectivity, there is an explicit unitary $V_n$ on $FAW$ such that
\begin{align}
V_n\ket{0^{2n}}
=\frac12\Bigl[&\ket0_F
\bigl(\ket+^{\otimes n}+e^{i\theta_n}\ket{0^n}\bigr)_A
\nonumber\\[-2pt]
{}+&\ket1_F
\bigl(\ket+^{\otimes n}-e^{i\theta_n}\ket{0^n}\bigr)_A
\Bigr]\ket{\mathbf0}_W .
\label{eq:lcu}
\end{align}
Consequently, measuring $F$ yields $F=0$ and leaves $A$ exactly in
$\ket{\phi_n}$ with probability $p_n=\mathcal N_n/4\geq3/8$; $W$ is reset in either
branch.  The circuit uses $n$ additional initialized work qubits---the flag
qubit $F$ and the $(n-1)$-qubit register $W$---together with $3n-2$ CNOTs,
$O(n)$ one-qubit gates, and CNOT depth
$2\lceil\log_2n\rceil+1$.
\end{lemma}

\begin{proof}
Figure~\ref{fig:lcu-flow} gives the register-level roadmap for the coherent
reference implementation.  Initialize $FAW$ in
$\ket0_F\ket{0^n}_A\ket{\mathbf0}_W$; the five steps below are listed in
physical time order.
\begin{enumerate}[label=(\arabic*),leftmargin=*,nosep]
\item \emph{Create the branch superposition.}  Apply
$U_{\rm flag}(\theta_n)$ to $F$, where
\begin{equation*}
U_{\rm flag}(\theta):=\operatorname{diag}(1,e^{i\theta})_F H_F.
\end{equation*}
It produces
\begin{equation*}
U_{\rm flag}(\theta_n)\ket0_F
=\frac{\ket0_F+e^{i\theta_n}\ket1_F}{\sqrt2},
\end{equation*}
with $A$ and $W$ still in $\ket{0^n}_A\ket{\mathbf0}_W$.

\item \emph{Fan out the branch label.}  For $b\in\{0,1\}$, the relevant
action of the coherent fanout on its initialized workspace is
\begin{equation*}
U_{\rm fan}\ket b_F\ket{\mathbf0}_W
=\ket b_F\ket{\mathbf b}_W .
\end{equation*}
This is the standard binary-tree cat-state (repetition-code) encoding
primitive: it coherently encodes the computational-basis branch label rather
than cloning an arbitrary input state.  Under all-to-all connectivity, a
binary tree of $n-1$ CNOTs and depth $\lceil\log_2n\rceil$ implements
$U_{\rm fan}$ \cite{MooreNilsson2001}, giving
\begin{equation*}
\frac{\ket0_F\ket{0^n}_A\ket{\mathbf0}_W
+e^{i\theta_n}\ket1_F\ket{0^n}_A\ket{\mathbf1}_W}{\sqrt2}.
\end{equation*}

\item \emph{Select the two product branches.}  Let $F$ control $A_1$ and
$W_j$ control $A_{j+1}$ for $1\leq j\leq n-1$.  Denote the parallel product
of the $n$ zero-controlled Hadamards by $U_{\rm sel}$.  Since each pair obeys
$\ket0\ket0\mapsto\ket0\ket+$ and
$\ket1\ket0\mapsto\ket1\ket0$, this step gives
\begin{equation*}
\frac{\ket0_F\ket+_A^{\otimes n}\ket{\mathbf0}_W
+e^{i\theta_n}\ket1_F\ket{0^n}_A\ket{\mathbf1}_W}{\sqrt2}.
\end{equation*}

\item \emph{Erase the copied label.}  Applying $U_{\rm fan}^{\dagger}$
resets $W$ without undoing the two branches already written to $A$:
\begin{equation*}
\frac{\ket0_F\ket+_A^{\otimes n}
+e^{i\theta_n}\ket1_F\ket{0^n}_A}{\sqrt2}
\ket{\mathbf0}_W.
\end{equation*}

\item \emph{Interfere the branches.}  A final $H_F$ interferes the branch
amplitudes, routing their sum---the LCU numerator in
Eq.~\eqref{eq:two-term-lcu}---to $F=0$ and their difference to $F=1$.
\end{enumerate}
The complete branch-synthesis unitary is therefore
\begin{equation}
V_n:=H_FU_{\rm fan}^\dagger U_{\rm sel}U_{\rm fan}
U_{\rm flag}(\theta_n),
\label{eq:Vn-blocks}
\end{equation}
where the rightmost factor acts first and identities on untouched registers
are implicit.  Applying the final $H_F$ to the state in step~(4) gives
Eq.~\eqref{eq:lcu}.  Its $F=0$ component is
$\sqrt{\mathcal N_n}\ket{\phi_n}/2$ and therefore has probability
\begin{equation}
p_n=\frac{\mathcal N_n}{4}\geq\frac38.
\label{eq:pn}
\end{equation}
Explicitly, $p_1=(3+\sqrt3)/8$, $p_2=3/8$, and
$p_n=1/2-2^{-(n+3)/2}$ for $n\geq3$.

The forward and reverse fanout trees contribute $2(n-1)$ CNOTs and
depth $2\lceil\log_2n\rceil$; the parallel layer $U_{\rm sel}$ contributes
$n$ further CNOTs and one layer.  Indeed,
$H=R_y(\pi/4)ZR_y(-\pi/4)$, so a controlled Hadamard is locally equivalent
to a controlled $Z$, and hence to one CNOT; changing the control from one to
zero requires only local $X$ gates.  This proves the stated resources.
\end{proof}

\begin{table*}[t]
\caption{Preparation resources for the two unitary schedules and two adaptive
LCU realizations.  Extra qubits are counted beyond the $n$-qubit fiducial
register $A$, so total preparation width is $n$ plus the second-column entry.
The heralded value is $n+1$, rather than the coherent reference's $n$, because
dynamic fanout uses $n$ measured auxiliaries in place of the $(n-1)$-qubit
workspace $W$; this extra qubit is a space-for-depth cost.
The deterministic adaptive bound is a conservative exact-primitives composition
\cite{Baumer2025,Buhrman2024,Takahashi2016}.  The unitary counts apply for
$n\geq2$, and the common $n$-pair Bell-analysis layer is excluded.  NN denotes
nearest-neighbor.}
\label{tab:hierarchy}
\centering
\scriptsize
\renewcommand{\arraystretch}{1.10}
\begin{tabular}{>{\raggedright\arraybackslash}p{3.65cm}
>{\raggedright\arraybackslash}p{2.00cm}
>{\raggedright\arraybackslash}p{3.25cm}
>{\raggedright\arraybackslash}p{3.25cm}
>{\raggedright\arraybackslash}p{3.85cm}}
\toprule
route and model & extra qubits beyond $A$ & gate-count bound & depth bound & success and measurements\\
\midrule
nearest-neighbor path (Thm.~\ref{thm:sequential})\newline
measurement-free; NN $A$ rail
& $0$
& $n-1$ two-qubit gates; $2n-3$ CNOTs
& $n-1$ two-qubit; $2n-3$ CNOT
& deterministic; no mid-circuit measurement\\
balanced tree (Thm.~\ref{thm:balanced-memory-tree})\newline
measurement-free; tree connectivity
& $0$
& $n-1$ two-qubit gates; $2n-3$ CNOTs
& $\lceil\log_2n\rceil$ two-qubit; $2\lceil\log_2n\rceil-1$ CNOT
& deterministic; no mid-circuit measurement\\
\addlinespace[3pt]
heralded LCU (Prop.~\ref{prop:heralded})\newline
adaptive; alternating 1D
\cite{Baumer2025}
& $n+1$
& $3n-1$ CNOTs/attempt; expected at most $\frac83(3n-1)$
& 5 CNOT layers/attempt; $O(1)$ expected adaptive; $O(\log n)$ classical
& nondeterministic; $p_n\geq3/8$; $\mathbb E[N_{\rm attempts}]\leq8/3$; 2 measurement rounds/attempt\\
deterministic adaptive LCU (Prop.~\ref{prop:deterministic-adaptive})\newline
adaptive; NN grid
& $O(n\log n)$
& $O(n\log n)$ quantum gates
& $O(1)$ adaptive; $O(\log n)$ classical
& deterministic; $p_{\rm succ}=1$; constant measurement rounds; no retry\\
\bottomrule
\end{tabular}
\end{table*}

\subsubsection{Constant-depth heralded realization}

Here adaptive quantum depth counts only the quantum layers between measurement
and feedforward rounds, excluding classical processing, reset, and queuing
latency.  The flag outcome $F$ heralds exact ancilla preparation: failed
attempts are discarded and repeated before $A$ interacts with $\rho$, and only
a successful ancilla is coupled to $\rho$.  Because acceptance is independent
of $\rho$, heralding leaves both the logical POVM and its reconstruction map
unchanged and cannot bias the Born distribution.  The internal measurement
outcomes are implementation records, not externally sampled shadow settings.

\begin{proposition}[Exact heralded constant adaptive quantum depth]
\label{prop:heralded}
Our construction uses the one-round dynamic-fanout primitive of
Ref.~\cite{Baumer2025} on its alternating 1D layout.  In each attempt,
conditional on obtaining $F=0$, it prepares $\ket{\phi_n}$ exactly.  Each
attempt has $O(n)$ quantum size, $O(1)$ adaptive quantum depth, one round of
mid-circuit measurements followed by the final measurement of $F$, and
$O(\log n)$ classical parity depth.  Repeating with a freshly reinitialized
ancilla register until success gives
\begin{equation}
\mathbb E[N_{\rm attempts}]=p_n^{-1}\leq\frac83.
\end{equation}
Thus expected adaptive quantum depth and total gate count are $O(1)$ and
$O(n)$, respectively.
\end{proposition}

In Eq.~\eqref{eq:Vn-blocks}, the only blocks whose depth grows with $n$ are
$U_{\rm fan}$ and $U_{\rm fan}^{\dagger}$.  On the initialized workspace, the
former coherently broadcasts the branch bit,
\begin{equation*}
\begin{aligned}
&(\alpha\ket0+\beta\ket1)_F\ket{0^{n-1}}_W
\longmapsto \alpha\ket0_F\ket{0^{n-1}}_W \\
&\hspace{7.4em}{}+\beta\ket1_F\ket{1^{n-1}}_W,
\end{aligned}
\end{equation*}
which is a GHZ-type encoding; the inverse erases the copied label.  The
measurement-and-feedforward construction of Ref.~\cite{Baumer2025} implements
this fanout exactly at constant CNOT depth using measured auxiliaries and
Pauli feedforward, while $U_{\rm sel}$ is already one parallel layer.  A
literal replacement in Eq.~\eqref{eq:Vn-blocks} would therefore use the
dynamic-fanout primitive twice.

For the present Hadamard targets, however, one call suffices.  When $W$ starts
in $\ket{\mathbf0}$, the full sandwich obeys
\begin{equation*}
\begin{aligned}
&U_{\rm fan}^{\dagger}U_{\rm sel}U_{\rm fan}
\ket b_F\ket\psi_A\ket{\mathbf0}_W \\
&\qquad={}
\ket b_F\bigl(H_A^{\otimes n}\bigr)^{1-b}
\ket\psi_A\ket{\mathbf0}_W,
\quad b\in\{0,1\}.
\end{aligned}
\end{equation*}
Thus it is a zero-controlled $H^{\otimes n}$.  To expose its fanout form, set
$L=R_y(\pi/4)H$ and $R=HR_y(-\pi/4)$.  A Hadamard on $A_j$
controlled by $F=1$ obeys
$C_{F=1}(H_{A_j})=(I_F\otimes L_{A_j})
\mathrm{CNOT}_{F\to A_j}(I_F\otimes R_{A_j})$.  Our control is on
$F=0$, which is obtained by flipping $F$ before and after the standard
one-controlled operation:
$C_{F=0}(U)=X_FC_{F=1}(U)X_F$.  Hence
\begin{equation*}
\begin{aligned}
C_{F=0}(H_A^{\otimes n})={}&
X_F(I_F\otimes L_A^{\otimes n}) \\
&\times\left(\prod_{j=1}^{n}\mathrm{CNOT}_{F\to A_j}\right)
(I_F\otimes R_A^{\otimes n})X_F .
\end{aligned}
\end{equation*}
The single-qubit gates in each outer layer act in parallel, and the middle
product is exactly one common-control fanout.  The one-round primitive of
Ref.~\cite{Baumer2025} then uses $n$ measured auxiliaries, $3n-1$ CNOTs, and
five CNOT layers on its alternating 1D layout.  A balanced XOR tree computes
the Pauli corrections in $O(\log n)$ classical depth.  After feedforward,
$H_F$ and the measurement of $F$ herald the LCU branch.  Hence one attempt has
peak width $2n+1$, one mid-circuit measurement round followed by the final
flag measurement, and expected CNOT count
$(3n-1)/p_n\leq\frac83(3n-1)$.  The corrected channel, Eq.~\eqref{eq:lcu}, and
$p_n$ are unchanged.

\subsubsection{Deterministic completion}

The heralded procedure succeeds after a random number of attempts, with
constant expected adaptive quantum depth.  Exact amplitude amplification makes
the preparation deterministic.
With additional auxiliary qubits, measurement and feedforward also keep the
completed preparation at constant adaptive quantum depth.

\begin{proposition}[Deterministic preparation at constant adaptive quantum depth]
\label{prop:deterministic-adaptive}
For every $n\geq2$, the local alternating quantum--classical computation
(LAQCC) model of Ref.~\cite{Buhrman2024}, with exact arbitrary one-qubit
rotations, admits an exact deterministic preparation of $\ket{\phi_n}$ with
$O(1)$ adaptive quantum depth.  The construction uses $O(n\log n)$ initialized
auxiliary qubits beyond $A$ and $O(n\log n)$ quantum gates, a constant number
of measurement rounds, and $O(\log n)$ classical feedforward depth.
\end{proposition}

The construction combines the two LCU branches coherently by one exact
amplification step.  Its branch-control operations and all-zero selective
phase admit constant-depth implementations using measurement and feedforward
\cite{Baumer2025,Buhrman2024,Takahashi2016}.
Appendix~\ref{app:lcu-alternatives} proves the amplification identity and gives
the selective-phase construction and resource count.  Preparation finishes
before $A$ contacts $S$, so the same Bell interface and reconstruction apply.
Relative to the measurement-free balanced tree, this route trades additional
auxiliary qubits and classical feedback for lower quantum depth; the two
resource points are compared in Table~\ref{tab:hierarchy}.

\subsection{Comparison and end-to-end resources}

Table~\ref{tab:hierarchy} places four resource points on one ledger: exact
two-qubit and CNOT resources for the strict-unitary schedules, asymptotic
resources for the heralded and deterministic adaptive LCU realizations, and
the universal Bell readout kept outside every preparation count.

For $n\geq2$ and before topology-dependent routing, either unitary preparation
followed by the common Bell layer contains $n-1$ preparation isometries and
$n$ Bell-analysis CNOTs, hence $2n-1$ two-qubit gates in total.  At readout,
the fiducial-assisted analyzer acts on exactly the $2n$ measured qubits $SA$;
no additional online workspace remains.  The optimized root isometry uses one
CNOT, the remaining $n-2$ preparation isometries use two each, and the Bell
layer contributes $n$.  The resulting pre-routing count is therefore
\begin{equation}
C_{\rm full}=\underbrace{1+2(n-2)}_{\rm preparation}
+\underbrace{n}_{\rm Bell}=3n-3
\label{eq:full-measurement-cnot-count}
\end{equation}
CNOTs.  With direct matched $S_j$--$A_j$ couplers, the CNOT depth is $2n-2$
for the path and $2\lceil\log_2n\rceil$ for the tree.  At $n=1$, the complete
tetrahedral-SIC measurement uses one Bell-analysis CNOT.

The two schedules expose a connectivity tradeoff without a width tradeoff:
the path uses only neighboring couplings within $A$, whereas the tree requires
long-range pairs.  On a single interleaved line, additional routing may still
be needed to combine either preparation with the matched Bell layer.  The
adaptive implementations use preparation measurements internally, but only
an exactly prepared ancilla is coupled to the unknown state; any transport
needed for that coupling lies outside the preparation-only resources in
Table~\ref{tab:hierarchy}.  Every implementation yields the same $4^n$
external outcomes $(\bm a,\bm b)$ and uses the same reconstruction map.

\FloatBarrier
\section{From-circuit numerical validation and finite-shot benchmarks}
\label{sec:shadow-demo}

\begin{figure*}[t]
\centering
\includegraphics[width=\textwidth]{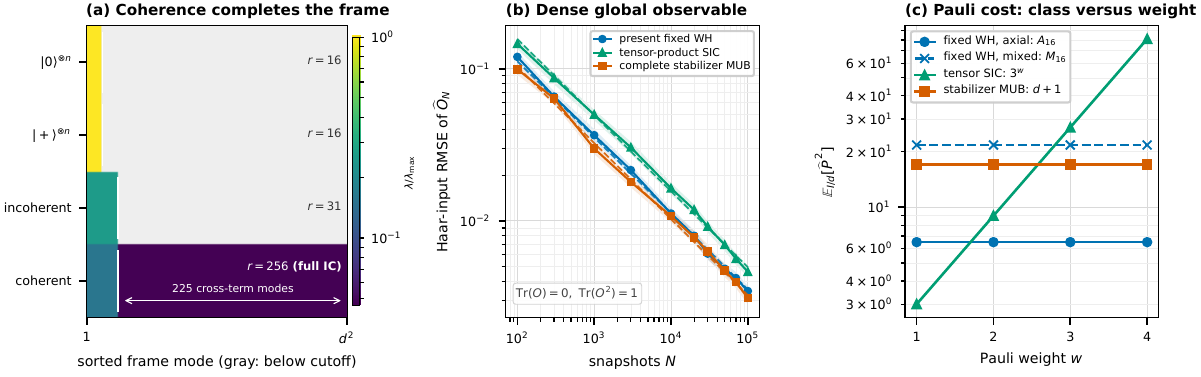}
\caption{From-circuit numerical validation at $n=4$.
(a) Common-normalized numerical frame spectra for the two product branches,
their incoherent mixture, and the coherent fiducial.  The rank ladder
$16,16,31,256$ shows that coherent cross terms add 225 directions and complete
the frame; gray entries lie below the $10^{-10}\lambda_{\max}$ rank cutoff.
(b) Haar-input root-mean-square error (RMSE) for one fixed
centered dense target with $\Tr(O^2)=1$.  The three measurements use the same
256 Haar-random input states and an independent outcome stream for every state
and measurement.  Markers are empirical RMSEs of the canonical sample means,
bands are pointwise 95\% percentile-bootstrap intervals over the 256
input--record trials, and dashed curves are numerical-frame predictions rather
than fits.  For the complete stabilizer MUB, this Haar prediction equals those
of uniform global-Clifford shadows and an ideal global SIC.  (c) Maximally-mixed
single-shot second moments grouped by Pauli weight.  Filled markers summarize
the weight-binned means over all 30 axial strings for the present WH
measurement, tensor SIC, and complete MUB; blue crosses give the corresponding
means over all 225 mixed WH strings.  The dashed blue line is the
independent closed-form value $M_{16}=(\sqrt{32}-1)^2$, while the axial WH
markers equal $A_{16}$.  Uniform global Clifford and an ideal global SIC share
the complete-MUB value $d+1$ at $I/d$.}
\label{fig:shadow-demo}
\end{figure*}

The numerics ask three questions not assumed by the analytic derivation:
whether coherence completes the frame, whether raw records show the predicted
finite-sample behavior, and whether the reconstructed costs exhibit the
class--weight tradeoff.  Panels~(a) and (c) are deterministic frame
diagnostics; panel~(b) continues through Born sampling and canonical sample
means.  We choose $n=4$, the first anisotropic case after the $n=3$ Hoggar SIC
and a size at which direct dense frame inversion remains practical.

Because every route in Table~\ref{tab:hierarchy} implements the same POVM, it
suffices to simulate one.  We independently compile $\ket{\phi_4^*}$ as a
nearest-neighbor bond-dimension-two circuit on the four qubits of $A$, with no
work qubits, and append the Bell layer.  For each computational-basis input
$\ket{x}$, the circuit gives
\begin{equation*}
A_{u,x}=\bra{u}U_{\rm IC}\ket{x,0^n}.
\end{equation*}
The row $A_u:=(A_{u,x})_x$ determines $E_u=A_u^\dagger A_u$ and hence
$\mathcal M_{\rm num}(A)=\sum_u\Tr(E_uA)E_u/\Tr E_u$.  We invert this frame
numerically and apply the same procedure to independently constructed
tensor-SIC and complete-MUB effects.  The MUB has the same depolarizing frame,
and hence the same Haar benchmark below, as uniform global Clifford
\cite{Stricker2022,Zhang2024,WangCui2024}.
No analytic WH spectrum, snapshot, or variance formula enters a numerical
marker or RMSE band.  The ideal global SIC is only a tight-frame benchmark; no
dimension-16 SIC fiducial or circuit is inserted \cite{Renes2004,Scott2006}.

Panel~(a) repeats the reconstruction for $\ket{0^n}$, $\ket{+^n}$, their
equal-weight incoherent mixture, and the coherent fiducial.  Their ranks are
$16,16,31,256$: the mixture retains only the two axial sectors, whereas the
coherent cross terms supply the 225 missing mixed directions.  These are
diagnostic controls, not alternative implementations.

Panel~(b) fixes the dense observable before sampling.  The same 256 independent
Haar states are used for all measurements, with a separate 100,000-outcome
Born stream for each state and measurement; each plotted $N$ is a nested
prefix.  Markers are RMSEs over these 256 end-to-end trials, and bands are
pointwise percentile-bootstrap intervals from 1000 resamples.  Dashed curves
use the numerically reconstructed maximally-mixed second moment and
$\mathbb E_\psi\langle\psi|O|\psi\rangle^2
=\Tr(O^2)/[d(d+1)]$, not a fit to the markers.  The complete-MUB variance is
$18/17$, shared here by uniform global Clifford and an ideal global SIC.  The
raw-record RMSE is consistent with $N^{-1/2}$ scaling.

Panel~(c) reads the single-shot moments directly from the numerical frames at
$I/d$.  For all 30 axial strings, the present value
$A_{16}\simeq6.49$ is weight independent, compared with $3^w$ for the tensor
SIC and 17 for the complete MUB.  All 225 mixed WH strings give
$M_{16}=(\sqrt{32}-1)^2\simeq21.69$, including each weight-one $Y_j$.
Uniform global Clifford and an ideal global SIC give 17 at $I/d$.  Thus the
present measurement beats the tensor SIC from weight two onward in the axial
sector, but pays a larger mixed-sector cost; this is a classwise tradeoff, not
uniform dominance.

These comparisons place the present measurement on a broader
hardware--statistical frontier.  It is attractive when shot-to-shot setting
control is costly and the target family contains axial or structured global
observables.
Product SICs remain natural for low-weight mixed targets, whereas
global-Clifford shadows provide the stronger dimension-independent
state-uniform guarantee at the price of a setting ensemble.  Locally entangled
shadows and target-set-specific strategies under bounded interaction range
occupy other points on this frontier
\cite{XuPauli2026,IppolitiEntangled2024}.

\section{Discussion and outlook}
\label{sec:discussion}

\discussionheading{Precompiled analyzer and system-facing interface.}
In the present construction, no shot-dependent measurement setting is applied
to the unknown system $S$: every copy encounters the same system-facing
analyzer, and the Born outcome itself labels the shadow snapshot.  The
fiducial register $A$ is freshly prepared before interacting with $S$, with
all fiducial-dependent processing confined to the analyzer side.  Once $A$
is ready, each system qubit $S_j$ undergoes only the matched Bell-analysis
operations $\operatorname{CNOT}_{S_j\to A_j}$ and $H_{S_j}$, followed by
terminal $Z$-basis measurement.  With direct matched couplers and ignoring
routing overhead, the system-facing entangling depth is one, and no
$S$--$S$ gate is required.  The nonlocal structure of the POVM is supplied
by the correlations prepared in $A$ together with the fixed pairwise Bell
interface, rather than by a global entangling circuit acting directly on $S$.

Randomized local-Pauli and global-Clifford shadows sample from large,
exponentially growing measurement ensembles.  Local Pauli measurements have $3^n$
possible settings, each requiring only local basis changes.  Even when each
setting is shallow, a long acquisition may invoke many distinct physical
control sequences; global-Clifford settings generally also involve multiqubit
gates on $S$ \cite{Huang2020}.  Precomputing the setting labels does not remove
the need to execute the corresponding physical configurations, although
caching and batching may reduce the associated overhead.  Repeating the fixed
system-facing interface eliminates shot-dependent reconfiguration of the
system-facing measurement setting during acquisition, at the cost of an
additional $n$-qubit analyzer register and per-shot analyzer preparation.

The analyzer architecture and system-facing control configuration remain fixed
across shots, while $A$ is reset and freshly prepared before each interaction
with $S$.  A simple platform-dependent measure of the operational cost is
the time required to reach a target precision,
\begin{equation*}
T_\epsilon^{(x)}\simeq
N_\epsilon^{(x)}
\bigl(\tau_{\rm shot}^{(x)}+\bar\tau_{\rm update}^{(x)}\bigr)
+T_{\rm cal}^{(x)},
\qquad x\in\{\mathrm{rand},\mathrm{fix}\}.
\end{equation*}
Here $N_\epsilon$ is the required number of shots, $\bar\tau_{\rm update}$
denotes the average setting-update overhead,
and $T_{\rm cal}^{(x)}$ accounts for the calibration cost of the corresponding
acquisition scheme.  For the fixed analyzer,
$\bar\tau_{\rm update}^{(\mathrm{fix})}=0$ after initialization, while
$\tau_{\rm shot}^{(\mathrm{fix})}$ includes fresh preparation of $A$, the fixed
Bell-analysis layer, and readout.

\discussionheading{Small-system tomography.}
The same fixed circuit provides informationally complete tomography with a
single measurement setting, replacing the $3^n$ local-Pauli settings
conventionally used for $n$-qubit tomography while retaining the minimal
$d^2=4^n$ outcomes of a rank-one IC measurement.  Beyond this reduction in
setting complexity, the Hilbert--Schmidt MSE coefficient is also exponentially smaller
under the canonical i.i.d. linear inversion of Sec.~\ref{sec:tomography}.
Randomized local-Pauli tomography and a fixed tensor-product qubit SIC measurement both
have Hilbert--Schmidt MSE coefficient $5^n$, whereas the present measurement
has coefficient $H_n$, with
\begin{equation*}
H_n<5^n\quad(n\geq3),\qquad H_n\sim2\cdot4^n.
\end{equation*}
Equivalently, $H_n/5^n\sim2(4/5)^n$, so the sample requirement at fixed
Hilbert--Schmidt MSE is smaller by an asymptotic factor of order $(5/4)^n$
than for these two benchmarks.  For $n=3$, the Hoggar SIC gives
$H_3=71$, compared with $5^3=125$ for the local-Pauli and product-SIC
benchmarks.  Table~\ref{tab:tomography-contexts} summarizes the finite-size
comparison.

Tensor-product qubit SICs already provide scalable fixed minimal-IC
measurements, but their Pauli statistics mirror those of randomized
local-Pauli shadows: the second moment grows as \(3^w\) with operator
weight \(w\).  The present correlated measurement instead keeps
axial-Pauli costs bounded at every weight and, for arbitrary Hermitian
targets, has the same dimension-independent scaling as global-Clifford
shadows in the Haar-averaged sense.  These gains require correlated
fiducial preparation and come with larger costs for low-weight mixed-Pauli
targets.

\discussionheading{From minimal IC to shallow SIC measurements.}
A SIC provides a rank-one minimal-IC benchmark with equal pairwise overlaps
between distinct measurement projectors, but its existence has not been
established in every power-of-two dimension
\cite{Horodecki2022Open,ApplebyFlammiaKopp2025}.  Existence and implementation
are separate questions: knowing a SIC fiducial does not by itself provide a
shallow circuit for preparing it on an $n$-qubit device.

The present construction gives an explicit minimal-IC measurement for every
$n$, with shallow fiducial preparation and exact SIC symmetry at $n=1,3$.
For other $n$, the projector overlaps are not all equal, but the statistical
bounds established above still apply.  An interesting question is whether
WH-covariant measurements with SIC symmetry, or uniformly near-SIC
projector overlaps and comparable statistical performance, can be realized
for all $n$ with shallow
qubit-compatible preparation.  This asks for more than SIC existence alone.

\discussionheading{Beyond minimality.}
For general targets, the present measurement has a dimension-independent
Haar-input coefficient, matching global-Clifford shadow scaling, but an
$O(d)$ state-uniform coefficient.  Relaxing minimality could improve the
latter while retaining a fixed analyzer.  For rank-one IC measurements,
Ref.~\cite{Yang2026} shows that a dimension-independent uniform
raw-second-moment bound for all traceless, Hilbert--Schmidt-normalized
observables requires $\Omega(d^3)$ outcomes, and achieves this scaling using
a collection of $\Theta(d^2)$ measurement bases.

Here we consider a different implementation question: whether the same
outcome scaling can be realized within the WH--Bell architecture.  A natural
$d^3$-outcome candidate uses $d$ WH fiducials coherently indexed by an
$n$-qubit label register, together with the $n$-qubit fiducial and system
registers, giving a $3n$-qubit fixed analyzer.  If the resulting combined
orbit forms an equal-weight complex projective $3$-design, the measurement
yields a dimension-independent, state-uniform variance bound for
Hilbert--Schmidt-normalized targets \cite{Huang2020}.  Whether such an
all-$n$ WH family exists, and whether its coherent fiducial preparation can
be compiled into a shallow $3n$-qubit circuit with computationally efficient
reconstruction, are interesting questions for further study.

\par\addvspace{0.6\baselineskip}
These results highlight measurement-setting control as an
experimental resource complementary to conventional metrics such as circuit
width, gate count, depth, and sample complexity.  Randomized protocols realize
part of their measurement complexity through shot-dependent setting selection
and implementation, whereas a fixed analyzer shifts part of that complexity
into precompiled measurement geometry,
analyzer-state preparation, and a fixed system-facing interface reused across shots.
More broadly, this raises the question of where the complexity of quantum
measurement should reside.  How much randomness and control must be supplied
anew on each shot, and how much can instead be compiled once into the geometry
and circuitry of the measuring device itself?

\begin{acknowledgments}
This work was supported by the National Natural Science Foundation of China
under Grants No.~42330707 and No.~42530108.
\end{acknowledgments}

\clearpage
\appendix

\section{Normalization of the WH orbit}
\label{app:povm-normalization}

For an $n$-qubit system, $d=2^n$.  This appendix proves the WH-orbit
normalization for an arbitrary normalized fiducial $\ket\phi$; the fixed POVM
in Eq.~\eqref{eq:fixed-povm} is obtained by setting
$\ket\phi=\ket{\phi_n}$.  We first make the two kinds of labels explicit.
Write
\begin{equation*}
u=(\bm a,\bm b),\qquad v=(\bm c,\bm e),
\end{equation*}
where all four bold symbols belong to $\F_2^n$.  The label $u$ identifies an
orbit element, so we abbreviate
$D_u=D_{\bm a,\bm b}$, $\Pi_u=\Pi_{\bm a,\bm b}$, and
$E_u=E_{\bm a,\bm b}$.  The independent label $v$ will identify a Pauli-basis
element.  Each label has $d^2=4^n$ possible values.

Every effect $E_u=d^{-1}\Pi_u$ is positive because $\Pi_u$ is a rank-one
projector.  To show that the effects form a POVM, it therefore remains only
to prove
\begin{equation*}
\sum_uE_u=\id.
\end{equation*}

\emph{Step 1: Expand an operator in the Pauli basis.}
Let $P_v=P_{\bm c,\bm e}$ be the Hermitian tensor Pauli proportional to
$Z(\bm c)X(\bm e)$.  Its irrelevant phase is chosen so that $P_v$ is
Hermitian.  The $d^2$ operators $\{P_v\}$ form an orthogonal basis, with
$P_0=\id$ and
\begin{equation*}
\Tr(P_vP_w)=d\delta_{vw}.
\end{equation*}
Every operator $A$ can therefore be written as
\begin{equation}
A=\frac1d\sum_v\Tr(P_vA)P_v.
\label{eq:appendix-pauli-expansion}
\end{equation}

\emph{Step 2: Average one Pauli component over all displacements.}
Recall that $u=(\bm a,\bm b)$ labels the conjugating displacement and
$v=(\bm c,\bm e)$ labels the Pauli being conjugated.  The single-qubit
relation $ZX=-XZ$ gives
\begin{equation*}
D_uP_vD_u^\dagger
=(-1)^{\bm a\cdot\bm e+\bm b\cdot\bm c}P_v
=:s_v(u)P_v.
\end{equation*}
Thus $s_v(u)=+1$ when $D_u$ and $P_v$ commute and $s_v(u)=-1$ when they
anticommute.  The notation $s_v(u)$ is only shorthand for the explicit binary
phase displayed above.

For a fixed $\bm e$, the sum
$\sum_{\bm a}(-1)^{\bm a\cdot\bm e}$ equals $d$ if $\bm e=\bm0$.  If
$\bm e\ne\bm0$, exactly half of the bit strings $\bm a$ give
$\bm a\cdot\bm e=0$ and half give $\bm a\cdot\bm e=1$, so the signs cancel
and the sum is zero.  The same argument applies to the sum over $\bm b$.
Consequently,
\begin{equation}
\begin{aligned}
\sum_us_v(u)
&=\sum_{\bm a,\bm b}
(-1)^{\bm a\cdot\bm e+\bm b\cdot\bm c}\\
&=
\left[\sum_{\bm a}(-1)^{\bm a\cdot\bm e}\right]
\left[\sum_{\bm b}(-1)^{\bm b\cdot\bm c}\right]\\
&=d^2\delta_{v0}.
\end{aligned}
\label{eq:appendix-character-sum}
\end{equation}
Here $\delta_{v0}=1$ exactly when $\bm c=\bm e=\bm0$.  Thus every
nonidentity Pauli component cancels in the displacement average; only
$P_0=\id$ survives.

\emph{Step 3: Average the full operator.}
Inside each conjugated operator $D_uAD_u^\dagger$, replace $A$ by its
expansion in Eq.~\eqref{eq:appendix-pauli-expansion}.  Then apply
Eq.~\eqref{eq:appendix-character-sum}:
\begin{equation}
\begin{aligned}
\frac1{d^2}\sum_uD_uAD_u^\dagger
&=\frac1{d^3}\sum_v\Tr(P_vA)P_v\sum_us_v(u)\\
&=\frac{\Tr A}{d}\id.
\end{aligned}
\label{eq:appendix-pauli-twirl}
\end{equation}
The last equality keeps only $v=0$, for which $P_0=\id$ and
$\Tr(P_0A)=\Tr A$.

\emph{Step 4: Apply the twirl to the fiducial projector.}
Take $A=\ketbra\phi\phi$.  Because $\ket\phi$ is normalized,
$\Tr A=1$, and its conjugates are precisely the orbit projectors
$\Pi_u=D_u\ketbra\phi\phi D_u^\dagger$.
Equation~\eqref{eq:appendix-pauli-twirl} therefore gives
\begin{equation}
\frac1{d^2}\sum_u\Pi_u=\frac1d\id,
\qquad
\sum_u\Pi_u=d\id,
\qquad
\sum_uE_u=\id.
\end{equation}
Thus the factor $1/d$ is exactly the physical POVM normalization.  This
argument uses only normalization of $\ket\phi$ and Pauli covariance; it does
not by itself imply informational completeness.  For the family generated by
$\ket{\phi_n}$,
Ref.~\cite{ZhuWang2026} establishes both informational completeness and
uniform spectral stability.

\section{Measurement-channel and reconstruction proofs}
\label{app:reconstruction-proofs}

\subsection{Pauli diagonalization and covariance}
\label{app:pauli-channel-proof}

We first prove Theorem~\ref{thm:pauli-diagonal-reconstruction} and the
covariance formulas in Eqs.~\eqref{eq:fixed-dual} and
\eqref{eq:fixed-dual-covariance}.  We use the conjugation sign
\(s_v(u)\in\{\pm1\}\) introduced in Appendix~\ref{app:povm-normalization},
so that \(D_uP_vD_u^\dagger=s_v(u)P_v\).  From the explicit sign formula
there, the exponents add modulo two, so
\(s_v(u)s_w(u)=s_{v\oplus w}(u)\).  Equation~\eqref{eq:appendix-character-sum}
then gives
\begin{equation}
\sum_u s_v(u)s_w(u)=d^2\delta_{v,w}.
\label{eq:pauli-character-orthogonality}
\end{equation}
By the Pauli expansion in Eq.~\eqref{eq:appendix-pauli-expansion} and the
definition of \(\chi_v\),
\begin{align}
\Pi_0
&=\frac1d\sum_v\Tr(P_v\Pi_0)P_v
=\frac1d\sum_v\chi_vP_v,\nonumber\\
\Pi_u=D_u\Pi_0D_u^\dagger
&=\frac1d\sum_v s_v(u)\chi_vP_v.
\label{eq:orbit-pauli-expansion}
\end{align}
Taking the Hilbert--Schmidt inner product with \(P_w\) and using
\(\Tr(P_vP_w)=d\delta_{v,w}\) then gives
\begin{equation}
\Tr(\Pi_uP_w)=s_w(u)\chi_w.
\label{eq:orbit-pauli-overlap}
\end{equation}
Because \(P_v\) is Hermitian, \(\chi_v=\Tr(\Pi_0P_v)\) is real.
Substituting Eqs.~\eqref{eq:orbit-pauli-expansion} and
\eqref{eq:orbit-pauli-overlap} directly into the definition of
\(\mathcal M_{\phi_n}\), and using
Eq.~\eqref{eq:pauli-character-orthogonality} gives
\begin{align}
\mathcal M_{\phi_n}(P_v)
&=\frac1d\sum_u\Tr(\Pi_uP_v)\Pi_u\nonumber\\
&=\frac{\chi_v}{d^2}\sum_{u,w}s_v(u)s_w(u)\chi_wP_w
=\chi_v^2P_v=|\chi_v|^2P_v.
\label{eq:appendix-pauli-diagonalization}
\end{align}
Writing \(A=\sum_v a_vP_v\), with
\(a_v=\Tr(P_vA)/d\), Eq.~\eqref{eq:appendix-pauli-diagonalization} acts on
each Pauli coefficient separately:
\begin{align*}
\mathcal M_{\phi_n}(A)&=\sum_v|\chi_v|^2a_vP_v,\\
\mathcal M_{\phi_n}^{-1}(A)&=\sum_v\frac{a_v}{|\chi_v|^2}P_v.
\end{align*}
These are the two formulas in Eq.~\eqref{eq:pauli-channel-closed-form}.  The
inverse exists exactly when every \(\chi_v\) is nonzero; equivalently, no
Pauli direction is erased, and the orbit is IC.

Assume henceforth that every \(\chi_v\) is nonzero.  For \(A=\Pi_0\),
Eq.~\eqref{eq:orbit-pauli-expansion} shows that the Pauli
coefficient is \(a_v=\chi_v/d\).  Since \(\chi_v\) is real,
\(\chi_v/|\chi_v|^2=1/\chi_v\), including when \(\chi_v<0\).  Substituting
these coefficients into the inverse gives
\begin{align*}
F_n=\mathcal M_{\phi_n}^{-1}(\Pi_0)
&=\frac1d\sum_v\frac{\chi_v}{|\chi_v|^2}P_v\\
&=\frac1d\sum_v\frac1{\chi_v}P_v.
\end{align*}

It remains to pass from the reference outcome to an arbitrary \(u\).  For
\(A=\sum_va_vP_v\), Pauli conjugation gives
\(D_uAD_u^\dagger=\sum_vs_v(u)a_vP_v\).  Applying the coefficientwise inverse
to this expansion proves covariance:
\begin{align}
\mathcal M_{\phi_n}^{-1}(D_uAD_u^\dagger)
&=\sum_v\frac{s_v(u)a_v}{|\chi_v|^2}P_v\nonumber\\
&=D_u\mathcal M_{\phi_n}^{-1}(A)D_u^\dagger.
\label{eq:appendix-inverse-covariance}
\end{align}
Finally, the orbit definition gives \(\Pi_u=D_u\Pi_0D_u^\dagger\).  Applying
Eq.~\eqref{eq:appendix-inverse-covariance} with \(A=\Pi_0\) yields
\begin{align*}
\widehat\rho_u
&=\mathcal M_{\phi_n}^{-1}(\Pi_u)\\
&=\mathcal M_{\phi_n}^{-1}(D_u\Pi_0D_u^\dagger)\\
&=D_u\mathcal M_{\phi_n}^{-1}(\Pi_0)D_u^\dagger
=D_uF_nD_u^\dagger.
\end{align*}
This proves Eqs.~\eqref{eq:fixed-dual} and
\eqref{eq:fixed-dual-covariance}.

\subsection{Fiducial overlaps and channel eigenvalues}
\label{app:all-qubit-maps}

It remains to evaluate these overlaps.  Write
\(P_{\bm c,\bm e}=\gamma_{\bm c,\bm e}D_{\bm c,\bm e}\), where
\(D_{\bm c,\bm e}=Z(\bm c)X(\bm e)\) is phase free and
\(|\gamma_{\bm c,\bm e}|=1\).  The corresponding phase-free overlap
\(\widetilde\chi_{\bm c,\bm e}
:=\langle\phi_n|D_{\bm c,\bm e}|\phi_n\rangle\) satisfies
\(\chi_{\bm c,\bm e}=\gamma_{\bm c,\bm e}
\widetilde\chi_{\bm c,\bm e}\), so
\(|\chi_{\bm c,\bm e}|^2=|\widetilde\chi_{\bm c,\bm e}|^2\).
With \(z_n=e^{i\theta_n}\), direct expansion of Eq.~\eqref{eq:fid} gives
\begin{align}
\widetilde\chi_{\bm c,\bm e}
&=\langle\phi_n|D_{\bm c,\bm e}|\phi_n\rangle\nonumber\\
&=\frac{\delta_{\bm c,\bm0}+\delta_{\bm e,\bm0}
+d^{-1/2}[z_n(-1)^{\bm c\cdot\bm e}+z_n^*]}{\mathcal N_n}.
\label{eq:sector-characteristic}
\end{align}

For \(n=1\), direct substitution gives
\(|\widetilde\chi_v|^2=1/3\) for all \(v\ne0\).  Hence
\(\mathcal M_{\phi_1}\) is the identity on the scalar component and multiplies
the traceless component by \(1/3\).  Therefore
\begin{align}
\mathcal M_{\phi_1}(A)
&=\frac{A+\Tr(A)\id}{3},\nonumber\\
\mathcal M_{\phi_1}^{-1}(A)
&=3A-\Tr(A)\id.
\label{eq:n1-channel-inverse}
\end{align}
Taking \(A=\Pi_0\) gives \(F_1=3\Pi_0-\id\).

For \(n=2\), Eq.~\eqref{eq:sector-characteristic} gives
\begin{equation}
|\widetilde\chi_{\bm c,\bm e}|^2=
\begin{cases}
1/9,&(\bm c,\bm e)\ne(\bm0,\bm0),\
\bm c\cdot\bm e=0,\\
1/3,&\bm c\cdot\bm e=1.
\end{cases}
\label{eq:n2-response-sectors}
\end{equation}
The Hermitian Pauli convention satisfies
\(P_{\bm c,\bm e}^{\mathsf T}
=(-1)^{\bm c\cdot\bm e}P_{\bm c,\bm e}\).  Thus the symmetric traceless
sector has channel eigenvalue \(1/9\), while the antisymmetric sector has
eigenvalue \(1/3\).  Writing
\[
A_\pm=\frac{A\pm A^{\mathsf T}}2,\qquad
A_0=\frac{\Tr(A)}4\id,
\]
gives
\begin{align}
\mathcal M_{\phi_2}(A)
&=\frac{2A-A^{\mathsf T}+2\Tr(A)\id}{9},\nonumber\\
\mathcal M_{\phi_2}^{-1}(A)
&=6A+3A^{\mathsf T}-2\Tr(A)\id.
\label{eq:n2-channel-inverse}
\end{align}
Taking \(A=\Pi_0\) gives
\(F_2=6\Pi_0+3\Pi_0^{\mathsf T}-2\id\).

For \(n\geq3\), we use the axial and mixed sets defined in
Eq.~\eqref{eq:axis-mixed-label-sets}.  For any operator \(A\), define its
scalar, axial, and mixed Pauli components by
\begin{align}
A_0&:=\frac{\Tr A}{d}\id,\qquad
c_v(A):=\frac1d\Tr(P_vA),\label{eq:pauli-coefficients}\\
A_{\mathsf{Ax}}&:=\sum_{v\in\mathsf{Ax}}c_v(A)P_v,\qquad
A_{\mathsf{Mix}}:=A-A_0-A_{\mathsf{Ax}}.
\label{eq:axis-mixed-pauli-expansion}
\end{align}

To evaluate the channel on these sectors, set
\begin{equation}
\begin{aligned}
r_d&:=\sqrt{\frac2d},&\mathcal N_n&=2-r_d,
&\kappa_d&:=1-r_d,\\
\alpha_d&:=\frac{\kappa_d^2}{\mathcal N_n^2},&
\beta_d&:=\frac{2}{d\mathcal N_n^2}.
\end{aligned}
\label{eq:axis-mixed-eigenvalues}
\end{equation}
Now \(z_n=e^{3\pi i/4}\).
On \(\mathsf{Ax}\), exactly one Kronecker delta in
Eq.~\eqref{eq:sector-characteristic} equals one and
\(\bm c\cdot\bm e=0\), giving
\(\widetilde\chi_v=(1-r_d)/\mathcal N_n\).  On \(\mathsf{Mix}\), both deltas vanish,
and the remaining numerator has modulus \(r_d\) for either value of
\(\bm c\cdot\bm e\).  Hence the channel eigenvalue is \(1\) in the identity
direction, \(\alpha_d\) on \(\mathsf{Ax}\), and \(\beta_d\) on
\(\mathsf{Mix}\).  Grouping the Pauli expansion into these three sectors gives
\begin{align}
\mathcal M_{\phi_n}(A)
&=A_0+\alpha_dA_{\mathsf{Ax}}+\beta_dA_{\mathsf{Mix}},\nonumber\\
\mathcal M_{\phi_n}^{-1}(A)
&=A_0+\alpha_d^{-1}A_{\mathsf{Ax}}
+\beta_d^{-1}A_{\mathsf{Mix}}.
\label{eq:compact-channel}
\end{align}
Taking \(A=\Pi_u\) gives
\begin{equation}
\widehat\rho_u=\frac{\id}{d}
+\alpha_d^{-1}(\Pi_u)_{\mathsf{Ax}}
+\beta_d^{-1}(\Pi_u)_{\mathsf{Mix}}.
\label{eq:compact-snapshot}
\end{equation}
At \(n=3\), \(d=8\) and the two nonidentity values coincide at
\(1/9=1/(d+1)\).  Since
\(A_{\mathsf{Ax}}+A_{\mathsf{Mix}}=A-A_0\) and
\(dA_0=\Tr(A)\id\), the sector formulas reduce directly to
\begin{align*}
\mathcal M_{\phi_3}(A)
&=A_0+\frac{A-A_0}{d+1}
=\frac{A+\Tr(A)\id}{d+1},\\
\mathcal M_{\phi_3}^{-1}(A)
&=A_0+(d+1)(A-A_0)
=(d+1)A-\Tr(A)\id.
\end{align*}
These are exactly the standard SIC channel and inverse.  In particular,
\(\Tr(\Pi_u)=1\) gives
\(\mathcal M_{\phi_3}^{-1}(\Pi_u)=9\Pi_u-\id\).
All listed overlaps are nonzero, completing the direct proof of minimal
informational completeness.

\subsection{Closed-form dual and outcome-dependent snapshots}
\label{app:closed-form-dual-proof}

We now convert the sector inverse in Eq.~\eqref{eq:compact-snapshot} into the
product form stated in Theorem~\ref{thm:structured-reconstruction}.  For
\(n\geq3\) and outcome \(u=(\bm a,\bm b)\), set
\begin{align*}
z_n&=e^{i\theta_n},&
\ket{x_{\bm a}}&=Z(\bm a)\ket+^{\otimes n},\\
\ket{\bm b}&=X(\bm b)\ket{0^n},&
\omega_{\bm a,\bm b}&=(-1)^{\bm a\cdot\bm b}.
\end{align*}
First, direct expansion of the orbit projector gives
\begin{equation}
\Pi_{\bm a,\bm b}
=\frac{\ketbra{x_{\bm a}}{x_{\bm a}}+\ketbra{\bm b}{\bm b}
+z_n^*\omega_{\bm a,\bm b}\ketbra{x_{\bm a}}{\bm b}
+z_n\omega_{\bm a,\bm b}\ketbra{\bm b}{x_{\bm a}}}{\mathcal N_n}.
\label{eq:orbit-four-products}
\end{equation}
Second, we isolate its axial and mixed Pauli components.  Define complete
dephasing in the \(Z\) and \(X\) bases by
\begin{align}
\Delta_Z(A)
&:=\sum_{\bm y\in\{0,1\}^n}
\ketbra{\bm y}{\bm y}A\ketbra{\bm y}{\bm y},\nonumber\\
\Delta_X(A)
&:=H^{\otimes n}\Delta_Z
\!\left(H^{\otimes n}AH^{\otimes n}\right)H^{\otimes n}.
\label{eq:axis-dephasing}
\end{align}
Equivalently, \(\Delta_X\) is complete dephasing in the \(X\)-eigenbasis
\(\{\ket{x_{\bm y}}\}_{\bm y}\).  The two maps retain, respectively, the
identity together with the \(Z\)-only and \(X\)-only Pauli components.  Hence
\begin{align}
A_{\mathsf{Ax}}
&=\Delta_Z(A)+\Delta_X(A)-2\frac{\Tr A}{d}\id,\nonumber\\
A_{\mathsf{Mix}}
&=A-\Delta_Z(A)-\Delta_X(A)+\frac{\Tr A}{d}\id.
\label{eq:appendix-axis-mixed-projections}
\end{align}
With \(\kappa_d=1-r_d\), direct dephasing gives
\begin{align}
\Delta_Z(\Pi_{\bm a,\bm b})
&=\frac{\id/d+\kappa_d\ketbra{\bm b}{\bm b}}{\mathcal N_n},\nonumber\\
\Delta_X(\Pi_{\bm a,\bm b})
&=\frac{\id/d+\kappa_d\ketbra{x_{\bm a}}{x_{\bm a}}}{\mathcal N_n}.
\label{eq:dephased-orbit}
\end{align}
Equations~\eqref{eq:appendix-axis-mixed-projections} and
\eqref{eq:dephased-orbit} therefore determine
\((\Pi_{\bm a,\bm b})_{\mathsf{Ax}}\) and
\((\Pi_{\bm a,\bm b})_{\mathsf{Mix}}\) explicitly.  Finally, substituting
them into Eq.~\eqref{eq:compact-snapshot} gives
\begin{align}
\widehat\rho_{\bm a,\bm b}={}&\frac{\mathcal N_n}{r_d\kappa_d}
(\ketbra{x_{\bm a}}{x_{\bm a}}+\ketbra{\bm b}{\bm b})\nonumber\\
&+\frac{\mathcal N_n}{r_d^2}\left(
z_n^*\omega_{\bm a,\bm b}\ketbra{x_{\bm a}}{\bm b}+
z_n\omega_{\bm a,\bm b}\ketbra{\bm b}{x_{\bm a}}\right)
-\frac{r_d}{\kappa_d}\id .
\label{eq:explicit-snapshot}
\end{align}
Setting \(\bm a=\bm b=\bm0\) recovers the reference dual in
Eq.~\eqref{eq:explicit-fixed-dual}.  Together with the \(n=1,2\) calculations
in Appendix~\ref{app:all-qubit-maps}, this proves
Theorem~\ref{thm:structured-reconstruction}.

\section{Gram spectrum and observable-estimation proofs}
\label{app:observable-proofs}

\subsection{Projector-Gram connection}
\label{app:gram-connection}

The expectation-value calculation above diagonalizes the measurement channel.
To use
the projector-Gram spectrum in the performance comparison, we now show that
the same Pauli sign patterns also diagonalize the projector-Gram matrix.
Fix one Pauli operator \(P_v\), and let the displacement---equivalently,
outcome---label \(u\) run over all \(d^2\) values.  The conjugation rule
\[
D_uP_vD_u^\dagger=s_v(u)P_v
\]
assigns one sign to each \(u\): \(s_v(u)=+1\) when \(D_u\) commutes with
\(P_v\), and \(s_v(u)=-1\) when they anticommute.  Collecting these signs
produces the \(d^2\)-component column vector
\[
\mathbf s_v:=(s_v(u))_u\in\{+1,-1\}^{d^2}.
\]
As \(v\) runs over all \(d^2\) Pauli labels, the vectors
\(\mathbf s_v\) form the orthogonal character basis in
Eq.~\eqref{eq:pauli-character-orthogonality}; each has norm \(d\), so
\(\mathbf s_v/d\) is normalized.

The matrix \(G^\Pi\) records the pairwise overlaps of the orbit projectors.
Using Eq.~\eqref{eq:orbit-pauli-expansion} gives
\begin{align}
G^\Pi_{u,u'}
&=\Tr(\Pi_u\Pi_{u'})
=\frac1d\sum_w s_w(u)s_w(u')|\chi_w|^2,\nonumber\\
\left(G^\Pi\mathbf s_v\right)_u
&=\sum_{u'}G^\Pi_{u,u'}s_v(u')
=d|\chi_v|^2s_v(u).
\label{eq:appendix-projector-gram-spectrum}
\end{align}
The second line is ordinary matrix-vector multiplication.  After inserting
the first line, character orthogonality gives
\(\sum_{u'}s_w(u')s_v(u')=d^2\delta_{w,v}\): every term with \(w\neq v\)
cancels, leaving only \(w=v\).  Therefore, for each Pauli label \(v\), the
projector-Gram eigenpair is
\[
G^\Pi\mathbf s_v=\lambda_v\mathbf s_v,\qquad
\mathbf s_v=(s_v(u))_u,\qquad
\lambda_v=d|\chi_v|^2.
\]
Here \(G^\Pi\) is the \(d^2\times d^2\) matrix being diagonalized,
\(\mathbf s_v\) is its \(d^2\)-component eigenvector, and \(\lambda_v\) is
the corresponding scalar eigenvalue.  The Pauli label \(v\) indexes these
eigenpairs.  This is the bridge between the fiducial expectation values that
control reconstruction and the projector-Gram spectrum used in
Sec.~\ref{sec:spectral-benchmarks}.

\subsection{Phase selection and full finite-size projector-Gram spectrum}
\label{app:full-spectrum}

\paragraph{Phase selection.}
We use the axial and mixed label sets of
Eq.~\eqref{eq:axis-mixed-label-sets} in Appendix~\ref{app:all-qubit-maps}.
To explain the phases in Eq.~\eqref{eq:fid}, allow an arbitrary relative
phase \(\theta\) while keeping the two branch amplitudes equal.  Set
\(\mathcal N_\theta=2+2\cos\theta/\sqrt d\).  Replacing \(z_n\) by \(e^{i\theta}\)
in Eq.~\eqref{eq:sector-characteristic} and using the Gram relation in
Appendix~\ref{app:gram-connection} gives the three nonidentity eigenvalue branches
\begin{align*}
\lambda_{\rm Ax}(\theta)
&=\frac{(\sqrt d+2\cos\theta)^2}{\mathcal N_\theta^2},\\
\lambda_{\rm even}(\theta)
&=\frac{4\cos^2\theta}{\mathcal N_\theta^2},\qquad
\lambda_{\rm odd}(\theta)
=\frac{4\sin^2\theta}{\mathcal N_\theta^2}.
\end{align*}
The last two branches are the mixed directions with even and odd numbers
of \(Y\) factors; both are present for \(d\geq4\).
For \(d\geq8\), put \(x=|\cos\theta|\).  Their smaller eigenvalue obeys
\[
\min\{\lambda_{\rm even},\lambda_{\rm odd}\}
\leq f_d(x):=
\frac{\min\{x^2,1-x^2\}}{(1-x/\sqrt d)^2},
\]
with equality when \(\cos\theta=-x\).  The function \(f_d\) increases up
to \(x=1/\sqrt2\) and decreases thereafter: on the second interval its
derivative has the sign of \(1/\sqrt d-x<0\).
Hence the mixed-sector upper bound is attained at \(\theta=3\pi/4\).
At this phase,
\[
\frac{\lambda_{\rm Ax}}{\lambda_{\rm even}}
=\frac{(\sqrt d-\sqrt2)^2}{2}\geq1.
\]
The axial branch does not lower the floor, proving phase optimality for
\(d\geq8\).

For \(d=4\) and \(\cos\theta=-x\leq0\), the smaller of the axial and
even mixed branches is
\[
\frac{\min\{(1-x)^2,x^2\}}{(1-x/2)^2}\leq\frac49.
\]
It is maximized at \(x=1/2\); there the odd branch is \(4/3\), so
\(\theta=2\pi/3\) attains the optimal floor \(4/9\).
For \(\cos\theta\geq0\), the mixed-sector floor is at most
\(2/(2+1/\sqrt2)^2<4/9\).
Finally, for \(d=2\) the even mixed branch is absent, and
\(\theta=5\pi/12\) makes all three nonidentity eigenvalues \(2/3\).
This reaches the SIC upper bound: the nonidentity eigenvalues of any
pure-state Pauli orbit sum to \(d^2-d\), so their minimum cannot exceed
\((d^2-d)/(d^2-1)=d/(d+1)\).
These optimizations concern the relative phase in the stated
equal-weight family, not unrestricted fiducials.

\paragraph{Spectrum at the selected phases.}
Substituting these phases gives the finite-size spectrum of Theorem~8 in
Ref.~\cite{ZhuWang2026}.  Its nonidentity part is
\begin{equation}
\begin{cases}
\{(2/3)^{(3)}\},&d=2,\\
\{(4/9)^{(9)},(4/3)^{(6)}\},&d=4,\\
\{(8/9)^{(63)}\},&d=8,\\
\{\mu_d^{((d-1)^2)},\nu_d^{(2(d-1))}\},&d>8,
\end{cases}
\end{equation}
where
\begin{equation}
\mu_d=\frac2{(2-\sqrt{2/d})^2},\qquad
\nu_d=\frac{d(1-\sqrt{2/d})^2}{(2-\sqrt{2/d})^2}.
\end{equation}
For $d\geq8$, Eq.~\eqref{eq:sector-characteristic} assigns $\nu_d$ to the
$2(d-1)$ nonidentity X- and Z-axis labels and $\mu_d$ to the $(d-1)^2$ mixed
labels; the two values coincide at $d=8$.  At $d=4$,
Eq.~\eqref{eq:n2-response-sectors} instead gives the even- and odd-$Y$ parity
sectors with multiplicities nine and six.

For \(d\geq8\), the mixed sector attains the smallest eigenvalue, so
\begin{align*}
\lambda_{\min}&=d\beta_d=\frac{2}{\mathcal N_n^2},\\
\eta_n&=\frac{d+1}{d}\lambda_{\min}
=\frac{2(d+1)}{d\mathcal N_n^2}\longrightarrow\frac12.
\end{align*}
Together with \((\eta_1,\eta_2,\eta_3)=(1,5/9,1)\), these expressions give
Eqs.~\eqref{eq:lambda-min-exact} and \eqref{eq:floor}.
The spectrum is not asymptotically flat:
$\nu_d/\mu_d=(\sqrt d-\sqrt2)^2/2$ diverges.  The main-text claim is a
uniform lower-edge statement, not geometric convergence to a SIC.

\subsection{General-observable variance bounds}
\label{app:general-observable-proof}

We prove Proposition~\ref{thm:general-observable-variance}.  At
\(\rho=\id/d\), every rank-one effect has probability \(1/d^2\).  By
self-adjointness of the inverse and the frame identity,
\begin{align}
\mathbb E_{\id/d}\widehat O^2
&=\frac1{d^2}\sum_u
[\Tr(\Pi_u\mathcal M_{\phi_n}^{-1}(O))]^2\nonumber\\
&=\frac1d\Tr[O\mathcal M_{\phi_n}^{-1}(O)].
\label{eq:appendix-general-second-moment}
\end{align}
Since \(\mathcal M^{-1}(\id)=\id\), applying
Eq.~\eqref{eq:appendix-general-second-moment} to $O_0$ gives its maximally
mixed second moment.  Its mean vanishes at $\id/d$, so it is also the variance
there.  In the
normalized Pauli basis, the nonidentity inverse eigenvalues are
\(d/\lambda_v\), so this quantity is at most
\(\Tr(O_0^2)/\lambda_{\min}\).  Equations~\eqref{eq:eta} and
\eqref{eq:floor} give the stated upper bound.

The conditional second moment is linear in the input state, so its Haar
average equals its value at \(\id/d\), whereas
\begin{equation}
\mathbb E_{\psi\sim\mathrm{Haar}}
\langle\psi|O_0|\psi\rangle^2
=\frac{\Tr(O_0^2)}{d(d+1)}.
\end{equation}
Subtracting this squared mean proves
Eq.~\eqref{eq:haar-average-variance}.  Finally, rank-one
positivity gives
\(p_\rho(u)=\Tr(\rho\Pi_u)/d\leq1/d=d\,p_{\id/d}(u)\).  Hence
\begin{align}
\mathbb E_\rho[\widehat O_0^2]
&\leq d\mathbb E_{\id/d}[\widehat O_0^2]\nonumber\\
&\leq2(d+1)\Tr(O_0^2).
\end{align}
Since variance is at most the second moment, taking the supremum over $\rho$
proves Eq.~\eqref{eq:worst-state-general-variance}.  Applying a standard
median-of-means estimate with $B_{\mathrm{var}}(\rho)$ from
Corollary~\ref{cor:simultaneous-estimation} to each target and taking a union
bound proves Eq.~\eqref{eq:sample-complexity}.

\subsection{Exact Pauli second moments}
\label{app:pauli-second-moment-proof}

We derive Eq.~\eqref{eq:exact-pauli-second-moment} and prove
Theorem~\ref{thm:global-pauli-variance}.  Pauli conjugation gives
\(\Tr(P_v\Pi_u)=s_v(u)\chi_v\), while
\(\mathcal M_{\phi_n}^{-1}(P_v)=P_v/|\chi_v|^2\).  Since \(P_v\) is
Hermitian, \(\chi_v\) is real.  The Pauli-diagonal inverse is self-adjoint,
and therefore
\begin{equation}
\widehat P_v(u)
=\Tr[P_v\mathcal M_{\phi_n}^{-1}(\Pi_u)]
=\Tr[\mathcal M_{\phi_n}^{-1}(P_v)\Pi_u]
=\frac{s_v(u)}{\chi_v}.
\label{eq:appendix-pauli-estimator}
\end{equation}
Its square is \(1/|\chi_v|^2=d/\lambda_v\) for every outcome and hence for
every input state.  Subtracting the squared mean proves
Eq.~\eqref{eq:global-pauli-variance}.  The $n=1$ expectation value and
Eq.~\eqref{eq:n2-response-sectors} give the low-dimensional values stated
after the theorem.  For \(n\geq3\), inserting \(\alpha_d\) and \(\beta_d\) gives
Eq.~\eqref{eq:axis-mixed-moments}.  Here $r_d\leq1/2$, so the mixed value
dominates the axial value, with equality at $n=3$; the former is exactly
$(\sqrt{2d}-1)^2$, while the latter is at most $9$ and tends to $4$.
Together with the exceptional cases, these observations complete the proof.

\subsection{Product-SIC baselines}
\label{app:product-sic-baselines}

For one qubit the tetrahedral canonical snapshot is \(3\Pi_r-\id\).  Every
tetrahedral Bloch component is \(\pm1/\sqrt3\), so a nonidentity Pauli
estimator is \(\pm\sqrt3\).  The same snapshot has eigenvalues \(2,-1\) and
therefore squared Hilbert--Schmidt norm \(5\).  Tensorization gives
\(3^{w(P)}\) for a weight-\(w(P)\) Pauli and \(5^n\) for the squared snapshot
norm.  Born averaging and unbiasedness give
Eq.~\eqref{eq:product-pauli-variance}; inserting
\(H_{\rm prod}=5^n\) from Eq.~\eqref{eq:tomography-baselines} into the iid
calculation in Appendix~\ref{app:tomography-proof} gives the product-SIC
tomography error.

\section{Canonical tomography error}
\label{app:tomography-proof}

We prove Corollary~\ref{thm:tomography-mse}.  Unbiasedness is
Eq.~\eqref{eq:fixed-shadow}.  Independence and zero mean of
\(\widehat\rho_{u_s}-\rho\) remove all cross terms, leaving
\begin{equation}
\mathbb E\|\widehat\rho_N-\rho\|_F^2
=\frac1N\left[
\mathbb E\Tr(\widehat\rho_u^2)-\Tr(\rho^2)\right].
\label{eq:appendix-iid-tomography}
\end{equation}
WH covariance makes all snapshots unitarily conjugate and hence equal in
norm.  In the normalized Pauli basis, the squared coefficient of \(\Pi_u\)
in direction \(v\ne0\) is \(\lambda_v/d^2\), while inversion multiplies that
coefficient by \(d/\lambda_v\).  Parseval therefore gives
\begin{equation}
\Tr(\widehat\rho_u^2)
=\frac1d\left(1+\sum_{v\ne0}\frac d{\lambda_v}\right)=H_n.
\end{equation}
For the low-dimensional cases this yields
\begin{align}
H_1&=\frac{1+3\cdot3}{2}=5,\nonumber\\
H_2&=\frac{1+9\cdot9+6\cdot3}{4}=25,\nonumber\\
H_3&=\frac{1+63\cdot9}{8}=71.
\label{eq:low-dimensional-snapshot-norms}
\end{align}
For \(n\geq3\), counting \(2(d-1)\) axis and \((d-1)^2\) mixed directions
gives the second line of Eq.~\eqref{eq:snapshot-hs-norm}.  Bounding every
nonidentity term by \(1/\lambda_{\min}\) and using Eq.~\eqref{eq:floor}
yields Eq.~\eqref{eq:snapshot-hs-bound}; substituting into
Eq.~\eqref{eq:appendix-iid-tomography} proves the corollary.

For completeness, the global benchmarks in
Eq.~\eqref{eq:tomography-baselines} follow from the standard canonical
snapshot
\begin{equation}
R=(d+1)\Pi-\id,\qquad \Pi^2=\Pi,\quad \Tr\Pi=1.
\end{equation}
For a global SIC, \(\Pi\) is the outcome projector; the corresponding
linear-tomography MSE is also given explicitly in Eq.~(103) of
Ref.~\cite{Scott2006}.  For uniformly randomized global-Clifford
measurements, \(\Pi=U^\dagger\ketbra{b}{b}U\), and the same snapshot formula
is Eq.~(S16) of Ref.~\cite{Huang2020}.  In both cases,
\begin{align}
R^2&=(d^2-1)\Pi+\id,\nonumber\\
\Tr(R^2)&=d^2+d-1.
\end{align}
Since this norm is independent of the measurement record, unbiasedness
and Eq.~\eqref{eq:appendix-iid-tomography} give
\begin{equation}
\mathbb E\|\widehat\rho_N-\rho\|_F^2
=\frac{d^2+d-1-\Tr(\rho^2)}{N}.
\end{equation}
Thus \(H_{\rm SIC}=H_{\rm Cl}=d^2+d-1=4^n+2^n-1\) for
\(d=2^n\), under the same i.i.d. canonical linear-inversion convention.

For the local baselines, the product-SIC snapshot norm is \(5^n\), as
shown in Appendix~\ref{app:product-sic-baselines}.  Uniformly randomized
local-Pauli acquisition has the same norm because its single-qubit
canonical factors also have the form \(3\Pi-\id\).
We can now compare the coefficients.  Equation~\eqref{eq:snapshot-hs-bound}
gives \(H_n\leq2H_{\rm SIC}-1/d<2H_{\rm SIC}\).
For the local comparison, \(H_3=71<5^3\).  For \(n\geq4\), the same bound
gives \(H_n<2d^2+2d\leq5^n\), since
\(2(4/5)^n+2(2/5)^n\) is less than one at \(n=4\) and decreases with \(n\).
Hence \(H_n<H_{\rm Pauli}=H_{\rm prod}\) for every \(n\geq3\).

\section{Detailed proof of the fixed-circuit realization}
\label{app:bell-proof}

\paragraph{Conventions and readout labels.}
Throughout, $\bm a$ denotes the final $S$-register readout after the
Hadamards, while $\bm b$ denotes the $A$-register readout.  For
$\bm x,\bm a,\bm b\in\F_2^n$, Eq.~\eqref{eq:field-to-pauli} gives
\begin{align}
D_{\bm a,\bm b}&=Z(\bm a)X(\bm b),\nonumber\\
Z(\bm a)\ket{\bm x}&=(-1)^{\bm a\cdot\bm x}\ket{\bm x},
&X(\bm b)\ket{\bm x}&=\ket{\bm x\oplus\bm b}.
\label{eq:appendix-pauli-convention}
\end{align}
We use this representative for the exact amplitude identity below; multiplying
$D_{\bm a,\bm b}$ by a unit phase leaves its orbit projector and probability
unchanged.

\paragraph{Pure-state amplitude through the circuit.}
For a pure system input, the state after conjugate-fiducial preparation is
\begin{equation}
\ket\psi_S\ket{\phi^*}_A
=\sum_{\bm x,\bm y}\psi_{\bm x}\phi_{\bm y}^*
\ket{\bm x}_S\ket{\bm y}_A.
\label{eq:appendix-input-expansion}
\end{equation}
The conjugated resource supplies the coefficients $\phi_{\bm y}^*$ appearing
in the bra $\bra\phi$.

Each $\operatorname{CNOT}_{S_j\to A_j}$ maps
$\ket{x_j,y_j}$ to $\ket{x_j,y_j\oplus x_j}$.  Applying the product of all
$n$ matched CNOTs therefore gives
\begin{equation}
\sum_{\bm x,\bm y}\psi_{\bm x}\phi_{\bm y}^*
\ket{\bm x}_S\ket{\bm y\oplus\bm x}_A.
\label{eq:appendix-after-cnot}
\end{equation}
For the final $A$-register label $\bm b$, only terms satisfying
$\bm y\oplus\bm x=\bm b$ contribute, equivalently
$\bm y=\bm x\oplus\bm b$.  The corresponding unnormalized system component is
\begin{equation}
\sum_{\bm x}\psi_{\bm x}\phi_{\bm x\oplus\bm b}^*\ket{\bm x}_S.
\label{eq:appendix-after-b}
\end{equation}

The Walsh--Hadamard identity is
\begin{equation}
H^{\otimes n}\ket{\bm x}
=d^{-1/2}\sum_{\bm a}(-1)^{\bm a\cdot\bm x}\ket{\bm a}.
\label{eq:appendix-hadamard}
\end{equation}
Applying these Hadamards and selecting the final $S$-register label $\bm a$
gives the full-circuit amplitude
\begin{align}
&\bra{\bm a,\bm b}U_{\rm IC}(\phi)
(\ket\psi_S\otimes\ket{0^n}_A)\nonumber\\
&\quad=\bra{\bm a,\bm b}U_{\rm Bell}
(\ket\psi_S\otimes\ket{\phi^*}_A)\nonumber\\
&\quad=d^{-1/2}\sum_{\bm x}(-1)^{\bm a\cdot\bm x}
\psi_{\bm x}\phi_{\bm x\oplus\bm b}^*.
\label{eq:appendix-component-amplitude}
\end{align}

\paragraph{Identify the WH displacement.}
For the convention above, taking the adjoint reverses the order of the
Hermitian Pauli factors:
\begin{equation}
D_{\bm a,\bm b}^\dagger=X(\bm b)Z(\bm a),\qquad
D_{\bm a,\bm b}^\dagger\ket{\bm x}
=(-1)^{\bm a\cdot\bm x}\ket{\bm x\oplus\bm b}.
\label{eq:appendix-adjoint-action}
\end{equation}
Therefore
\begin{equation}
\bra\phi D_{\bm a,\bm b}^\dagger\ket\psi
=\sum_{\bm x}(-1)^{\bm a\cdot\bm x}
\psi_{\bm x}\phi_{\bm x\oplus\bm b}^*.
\label{eq:appendix-pauli-amplitude}
\end{equation}
Comparing with Eq.~\eqref{eq:appendix-component-amplitude} proves the required
amplitude identity.

\paragraph{Born probability and mixed inputs.}
For the pure state $\rho_\psi=\ketbra\psi\psi$, Born's rule gives
\begin{align}
p(\bm a,\bm b\mid\psi)
&=\frac1d\abs{\bra\phi D_{\bm a,\bm b}^\dagger\ket\psi}^2\nonumber\\
&=\frac1d\Tr\!\left[
\rho_\psi D_{\bm a,\bm b}\ketbra\phi\phi
D_{\bm a,\bm b}^\dagger\right].
\label{eq:appendix-pure-probability}
\end{align}
A mixed state has a decomposition
$\rho=\sum_k r_k\ketbra{\psi_k}{\psi_k}$ with $r_k\geq0$ and
$\sum_k r_k=1$.  Both the Born probability and the trace expression are
affine in $\rho$, so averaging the pure-state identity gives
\begin{equation}
p(\bm a,\bm b\mid\rho)
=\frac1d\Tr\!\left[
\rho D_{\bm a,\bm b}\ketbra\phi\phi
D_{\bm a,\bm b}^\dagger\right].
\label{eq:appendix-final-probability}
\end{equation}
This is Eq.~\eqref{eq:born-uic}.  For $\phi=\phi_n$,
Eq.~\eqref{eq:fixed-povm} identifies the operator inside the Born rule with
$E_{\bm a,\bm b}^{(n)}$, completing the proof of
Proposition~\ref{prop:generic-bell}.

\FloatBarrier

\section{Explicit CNOT compilation of the root and binary splits}
\label{app:path-block-compilation}

\emph{Root block: one CNOT.}
The optimization is a state preparation, not a one-CNOT implementation of the
full first-split unitary.  Let $a+b=n$, with $(a,b)=(1,n-1)$ for the path and
$(a,b)=(\lfloor n/2\rfloor,\lceil n/2\rceil)$ for the balanced tree.  Writing
the state in Eq.~\eqref{eq:root-two-qubit-state} as
$\ket{\xi_{a,b}^{(n)}}=\sum_{j,k=0}^1(X_{a,b})_{jk}\ket{j,k}$ gives the known
coefficient matrix
\begin{equation}
X_{a,b}=\frac1{\sqrt{\mathcal N_n}}
\begin{pmatrix}
e^{i\theta_n}+c_ac_b & c_ar_b\\
r_ac_b & r_ar_b
\end{pmatrix}.
\label{eq:root-coefficient-matrix}
\end{equation}
It obeys $\Tr(X_{a,b}^\dagger X_{a,b})=1$ and
\begin{equation}
\det X_{a,b}=\frac{e^{i\theta_n}r_ar_b}{\mathcal N_n},
\qquad \Delta_{a,b}:=|\det X_{a,b}|=\frac{r_ar_b}{\mathcal N_n}>0.
\label{eq:root-determinant}
\end{equation}
Take a singular-value decomposition
\begin{equation}
X_{a,b}=L_{a,b}\operatorname{diag}(\sigma_0,\sigma_1)R_{a,b}^\dagger,
\qquad \sigma_0\geq\sigma_1>0.
\label{eq:root-svd}
\end{equation}
Normalization and Eq.~\eqref{eq:root-determinant} give
$\sigma_{0,1}^2=(1\pm\sqrt{1-4\Delta_{a,b}^2})/2$.  Define
\begin{equation}
\begin{aligned}
\gamma_{a,b}&:=\arcsin(2\Delta_{a,b})\\
&=\arcsin\!\left(\frac{2r_ar_b}{\mathcal N_n}\right),
\qquad 0<\gamma_{a,b}\leq\frac{\pi}{2}.
\end{aligned}
\label{eq:root-schmidt-angle}
\end{equation}
Then $\sigma_0=\cos(\gamma_{a,b}/2)$ and
$\sigma_1=\sin(\gamma_{a,b}/2)$.  With
$R_y(\vartheta)=e^{-i\vartheta Y/2}$, an explicit root circuit is
\begin{equation}
\ket{\xi_{a,b}^{(n)}}
=(L_{a,b}\otimes R_{a,b}^*)
\operatorname{CNOT}_{1\to2}
[R_y(\gamma_{a,b})\otimes I]\ket{00}.
\label{eq:root-one-cnot-compilation}
\end{equation}
Here $L_{a,b}$ and $R_{a,b}^*$ are single-qubit unitaries obtained from the
SVD; the displayed CNOT is the sole two-qubit gate.
Indeed, the rotation and CNOT first produce
$\sigma_0\ket{00}+\sigma_1\ket{11}$.  Under local gates $A\otimes B$, a
coefficient matrix transforms as $X\mapsto AXB^T$, so the final local gates
produce $L_{a,b}\operatorname{diag}(\sigma_0,\sigma_1)R_{a,b}^\dagger$.
The complex conjugate on $R_{a,b}$ is therefore essential.  Since
$\det X_{a,b}\neq0$, the target is entangled and cannot be prepared using only
one-qubit gates; Eq.~\eqref{eq:root-one-cnot-compilation} is thus an exact and
optimal one-CNOT preparation.  For the path, it directly prepares
Eq.~\eqref{eq:first-sequential-step} from $\ket{00}$, replacing the separate
preparation of $\ket{\eta_n}$ and application of $G_1$.

\emph{All nonroot splits: one two-CNOT template.}
Let $a,b\geq1$ and $m=a+b$.  The two branch inputs in
Eq.~\eqref{eq:binary-memory-splitting} span
$\{\ket{00},\ket{10}\}$.  By linearity, the required action on these
orthonormal inputs is
\begin{align}
 U_{a,b}\ket{00}&=\ket{00},\nonumber\\
 U_{a,b}\ket{10}&=\frac1{r_m}\bigl(
 c_ar_b\ket{01}+r_ac_b\ket{10}+r_ar_b\ket{11}\bigr).
 \label{eq:binary-split-columns}
\end{align}
The second output is normalized since $r_b^2+r_a^2c_b^2=r_m^2$ and
is orthogonal to the first.  Define the two rotation angles
\begin{equation}
 \alpha_a:=\arccos c_a,\qquad
 \delta_{a,b}:=\arctan\!\left(\frac{r_ac_b}{r_b}\right).
 \label{eq:binary-split-angles}
\end{equation}
Thus $\cos\alpha_a=c_a$, $\sin\alpha_a=r_a$,
$\sin\delta_{a,b}=r_ac_b/r_m$, and $\cos\delta_{a,b}=r_b/r_m$.
An explicit unitary completion is
\begin{align}
 U_{a,b}={}&[R_y(\alpha_a)\otimes I]\,
 \operatorname{CNOT}_{2\to1}\nonumber\\
 &\times[R_y(-\alpha_a)\otimes R_y(-\delta_{a,b})]\nonumber\\
 &\times\operatorname{CNOT}_{1\to2}
 [I\otimes R_y(\delta_{a,b})].
 \label{eq:binary-split-factorization}
\end{align}
Figure~\ref{fig:binary-split-cnot} shows the gates in physical time order.

\begin{center}
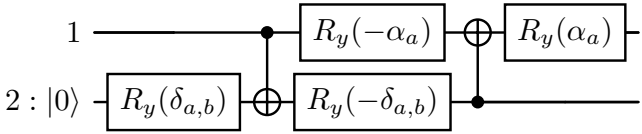

\begin{minipage}{\columnwidth}
\centering
\resizebox{\columnwidth}{!}{%
\begin{quantikz}[row sep={0.8cm,between origins},column sep=0.15cm]
 \lstick{$1$} & \qw & \ctrl{1}
 & \gate{R_y(-\alpha_a)} & \targ{}
 & \gate{R_y(\alpha_a)} & \qw \\
 \lstick{$2:\ket0$} & \gate{R_y(\delta_{a,b})} & \targ{}
 & \gate{R_y(-\delta_{a,b})} & \ctrl{-1}
 & \qw & \qw
\end{quantikz}%
}
\captionof{figure}{Two-CNOT circuit for the split $m=a+b$.
Wire 1 carries the one-qubit memory; wire 2 is an unused output qubit.
The same circuit implements both the path ($a=1$) and balanced-tree splits;
only the rotation angles change.}
\label{fig:binary-split-cnot}
\end{minipage}
\end{center}

To verify the circuit, the target rotation, first CNOT, and inverse target
rotation leave $\ket{00}$ unchanged and send
$\ket{10}$ to
$\sin\delta_{a,b}\ket{10}+\cos\delta_{a,b}\ket{11}$.
The remaining operation is the second CNOT conjugated by
$R_y(\alpha_a)$ on wire 1.  When wire 2 is one, it maps
$\ket1$ on wire 1 to $c_a\ket0+r_a\ket1$; when wire 2 is zero, it
acts as the identity.  Substitution gives
Eq.~\eqref{eq:binary-split-columns}.  In particular,
\begin{equation*}
 U_{a,b}(c_m\ket{00}+r_m\ket{10})
 =(c_a\ket0+r_a\ket1)\otimes(c_b\ket0+r_b\ket1),
\end{equation*}
which is the required coherent branch map.

For the path, $a=1$ gives $\alpha_1=\pi/4$ and
$\delta_{1,m-1}=\arctan(c_m/r_{m-1})$, recovering $U_m=U_{1,m-1}$.
Equation~\eqref{eq:binary-split-factorization} therefore supplies two CNOTs
and four one-qubit rotations for every nonroot split in either schedule.
The first split still uses the separate one-CNOT state preparation above,
so the preparation count $1+2(n-2)=2n-3$ and balanced CNOT depth
$2\lceil\log_2n\rceil-1$ are unchanged.

Two CNOTs are also necessary for this full coherent action on the reachable
input subspace.  Up to local gates, a one-CNOT circuit receives two orthogonal
first-qubit states with the same second-qubit factor.  Expanding either
control direction shows that the two output coefficient determinants vanish
simultaneously.  Here $U_{a,b}\ket{00}$ is a product state, whereas the
coefficient determinant of $U_{a,b}\ket{10}$ is
$-c_ar_ar_bc_b/r_m^2\neq0$.  One CNOT cannot implement both prescribed
columns.  This is a specialized real-isometry compilation
\cite{Iten2016,Vatan2004,WeiDi2012}.

\section{Hoggar fiducial: preparation lower bound and numerical circuit}
\label{app:hoggar}

\begin{table}[b]
\caption{One numerical Euler-angle realization of the ancilla-free
CNOT-optimal Hoggar fiducial preparation, rounded to six decimals for display.
The numerical calculations use the angles at full precision.}
\label{tab:hoggar-angles}
\begin{ruledtabular}
\scriptsize
\begin{tabular}{cc|c}
$\ell$ & qubit $j$ & $(a_{\ell j},b_{\ell j},c_{\ell j})$\\\hline
1 & 1 & $(-0.846651,\ 1.045827,\ -3.103113)$\\
1 & 2 & $(-1.258829,\ -2.482732,\ -2.932261)$\\
1 & 3 & $( 2.289262,\ -0.126380,\ -0.101227)$\\
2 & 1 & $(-2.266624,\ -2.170051,\  0.240365)$\\
2 & 2 & $( 2.092672,\  2.799899,\ -2.845832)$\\
2 & 3 & $(-0.549941,\ -0.809641,\  2.049816)$\\
3 & 1 & $(-2.987084,\ -2.685355,\ -1.727123)$\\
3 & 2 & $(-3.079891,\ -0.795702,\  2.443830)$\\
3 & 3 & $( 1.842244,\ -0.142857,\  2.655990)$\\
4 & 1 & $( 1.360092,\ -0.627407,\ -0.853500)$\\
4 & 2 & $(-3.058065,\ -0.762020,\  0.350864)$\\
4 & 3 & $(-3.004442,\  0.230627,\  1.255897)$
\end{tabular}
\end{ruledtabular}
\end{table}

At $d=8$, $\mathcal N_3=3/2$ and
\begin{equation}
\ket{\phi_3}=\frac1{\sqrt{12}}(-1+2i,1,1,1,1,1,1,1)^T.
\end{equation}
This is the standard Pauli-covariant Hoggar fiducial \cite{ZhuWang2026}.  In
the ancilla-free three-qubit unitary model with arbitrary one-qubit gates and
arbitrary CNOT connectivity, it requires exactly three CNOTs to prepare from a
product state.  Three suffice for any three-qubit pure state
\cite{Znid2008}.  For the lower bound, the unnumbered section
``Classification with respect to $\ket{000}$'' of Ref.~\cite{Znid2008} gives,
for its exact two-CNOT class up to local unitaries and a qubit permutation,
with $\braket a{a^\perp}=0$, the form
\begin{equation}
\cos\alpha\ket a\ket b\ket c+\sin\alpha\ket{a^\perp}\ket{b'}\ket{c'}.
\end{equation}
The zero-CNOT class is the degenerate product
case.  The one-CNOT class is biseparable; Schmidt-decomposing its entangled pair
and, if necessary, relabeling the qubits gives the same orthogonal-factor form.
Hence every state reachable with at most two CNOTs has a representative of
this form.  Tracing out the orthogonal-factor qubit gives
\begin{equation}
\rho_{23}=\cos^2\alpha\ketbra{bc}{bc}
+\sin^2\alpha\ketbra{b'c'}{b'c'},
\end{equation}
a separable mixture.  Thus at least one two-qubit reduced concurrence is zero.  Direct reduction of the Hoggar vector gives
\begin{equation}
C(\rho_{12})=C(\rho_{13})=C(\rho_{23})=\frac{\sqrt2}{3}>0.
\end{equation}
For an analytic check, tracing out one qubit gives the two subnormalized
conditional vectors
\begin{equation}
\begin{aligned}
v_0&=(a,b,b,b),& v_1&=(b,b,b,b),\\
a&=\frac{-1+2i}{\sqrt{12}},& b&=\frac1{\sqrt{12}}.
\end{aligned}
\end{equation}
Their Wootters matrix
\(\tau_{rs}=\langle v_r|(\sigma_y\otimes\sigma_y)|v_s^*\rangle\) is, up to
entrywise conjugation, \(c\left(\begin{smallmatrix}2&1\\1&0\end{smallmatrix}\right)\)
with \(\abs c=\sqrt2/6\).  Its two singular values differ by
\(2\abs c=\sqrt2/3\).  Permutation symmetry of the target gives the same result
for all three pairs.

Concurrence is invariant under local unitaries, and a qubit permutation only
relabels the three reduced pairs, so the preceding obstruction applies to every
representative in the cited at-most-two-CNOT classes.  The positive values
therefore exclude all zero-, one-, and two-CNOT classes.  This proves
CNOT-optimal preparation of the fiducial, not global optimality of a Naimark
unitary.  A separate question is whether the same three-qubit Hoggar SIC POVM
admits an exact six-qubit, strict-unitary Naimark realization with fewer than
six CNOTs under the same gate model.

An analytic three-CNOT preparation follows from the $n=3$ instance of the
general memory-splitting construction in Appendix~\ref{app:path-block-compilation}:
one CNOT prepares the root state $\ket{\xi_{1,2}^{(3)}}$, and two implement the
final $U_{1,1}$ split.  Figure~\ref{fig:hoggar-circuit} gives a separate numerical
three-CNOT compilation of the same fiducial, rather than an expansion of
that analytic circuit.
\begin{center}
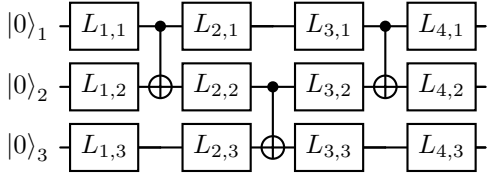

\begin{minipage}{\columnwidth}
\centering
\begin{quantikz}[row sep=0.18cm, column sep=0.13cm]
\lstick{$\ket{0}_1$} & \gate{L_{1,1}} & \ctrl{1} & \gate{L_{2,1}} & \qw
& \gate{L_{3,1}} & \ctrl{1} & \gate{L_{4,1}} & \qw \\
\lstick{$\ket{0}_2$} & \gate{L_{1,2}} & \targ{} & \gate{L_{2,2}} & \ctrl{1}
& \gate{L_{3,2}} & \targ{} & \gate{L_{4,2}} & \qw \\
\lstick{$\ket{0}_3$} & \gate{L_{1,3}} & \qw & \gate{L_{2,3}} & \targ{}
& \gate{L_{3,3}} & \qw & \gate{L_{4,3}} & \qw
\end{quantikz}
\captionof{figure}{Numerical ancilla-free three-CNOT compilation of the Hoggar
fiducial.  The preceding analytic argument proves that the CNOT count is
optimal.  Here
$L_{\ell,j}=R_z(a_{\ell j})R_y(b_{\ell j})R_z(c_{\ell j})$, with
$R_z(\vartheta)=e^{-i\vartheta Z/2}$; Table~\ref{tab:hoggar-angles}
lists the angles rounded to six decimals.}
\label{fig:hoggar-circuit}
\end{minipage}
\end{center}

\FloatBarrier
\section{Alternative coherent and adaptive LCU completions}
\label{app:lcu-alternatives}

The heralded branch synthesis of Lemma~\ref{lem:branch-synthesis} can be made
deterministic by one exact amplitude-amplification step.  We give two
implementations of the same identity: a fully unitary circuit and a
measurement-and-feedforward circuit.  We first prove the logical identity,
then compile its two selective phases in the two models.

\noindent\textit{Logical amplitude amplification.---}
Lemma~\ref{lem:branch-synthesis} gives
\begin{align}
\ket G&:=\ket0_F\ket{\phi_n}_A\ket{\mathbf0}_W,\nonumber\\
\ket v&:=V_n\ket{0^{2n}}
=\sqrt{p_n}\ket G+\sqrt{1-p_n}\ket B,
\label{eq:appendix-good-bad}
\end{align}
where $\ket B$ is the normalized $F=1$ branch and is orthogonal to $\ket G$.
The task is to remove its $\ket B$ amplitude coherently.
Figure~\ref{fig:exact-amplification-circuit}(a) previews the five operations
below in their physical time order.

First phase the good branch.  Because the two branches are distinguished by
$F$, define
\begin{equation}
S_G(\varphi)
:=\id_{FAW}+(e^{i\varphi}-1)
(\ketbra{0}{0}_F\otimes\id_{AW}).
\label{eq:SG-projector-form}
\end{equation}
For $p=p_n$ and $z=e^{i\varphi}$, the state becomes
$S_G(\varphi)\ket v=z\sqrt p\ket G+\sqrt{1-p}\ket B$; its overlap with
$\ket v$ is $1+p(z-1)$.

Next undo the preparation with $V_n^\dagger$, phase only its all-zero input,
and reapply $V_n$.  The required input phase is
\begin{equation}
S_0^{FAW}(\varphi)
:=\id_{FAW}+(e^{i\varphi}-1)
\ketbra{0^{2n}}{0^{2n}}.
\label{eq:S0-projector-form}
\end{equation}
Conjugating the input phase by $V_n$ gives
\begin{equation}
R_v(\varphi):=V_nS_0^{FAW}(\varphi)V_n^\dagger
=\id_{FAW}+(z-1)\ketbra{v}{v}.
\label{eq:conjugated-input-phase}
\end{equation}
Thus the last three physical operations phase the prepared direction
$\ket v$.  Applying them to the state above, the remaining bad-branch
coefficient is
\begin{align}
c_B&:=\bra B R_v(\varphi)S_G(\varphi)\ket v\nonumber\\
&=\sqrt{1-p}\{1+(z-1)[1+p(z-1)]\}\nonumber\\
&=\sqrt{1-p}\,z[1-4p\sin^2(\varphi/2)].
\label{eq:amplified-bad-coefficient}
\end{align}
The bad branch therefore vanishes if
\begin{equation}
\varphi_n:=2\arcsin\frac{1}{2\sqrt{p_n}},
\qquad
4p_n\sin^2(\varphi_n/2)=1.
\label{eq:appendix-lcu-phases}
\end{equation}
This phase is real whenever $p_n\geq1/4$, in particular under our stronger
bound $p_n\geq3/8$.  At $\varphi=\pi$, $S_G$ and $S_0$ are the two usual
Grover sign flips; the matched non-$\pi$ phase in
Eq.~\eqref{eq:appendix-lcu-phases} makes the single step exact
\cite{Hoyer2000}.  Since $c_B=0$ and the circuit is unitary, only a normalized
good component remains.  Hence, up to an irrelevant global phase,
\begin{align}
&\Bigl(V_nS_0^{FAW}(\varphi_n)V_n^\dagger\Bigr)
S_G(\varphi_n)V_n\ket{0^{2n}}\nonumber\\
&\hspace{4em}=e^{i\gamma_n}\ket G.
\label{eq:appendix-exact-lcu}
\end{align}

\noindent\textit{Unitary compilation of the selective phases.---}
The two selective phases have sharply different circuit costs.  The good-branch
phase $S_G$ is local because the branch label is stored in $F$.  Let
$P(\varphi):=\operatorname{diag}(1,e^{i\varphi})$.  This ordinary phase gate
phases $\ket1$, whereas $S_G$ must phase $F=0$.  Since $X$ exchanges $\ket0$
and $\ket1$,
\begin{equation}
\begin{aligned}
XP(\varphi)X
&=\operatorname{diag}(e^{i\varphi},1)
=\id+(e^{i\varphi}-1)\ketbra{0}{0},\\
S_G(\varphi)
&=[X_FP_F(\varphi)X_F]\otimes\id_{AW}.
\end{aligned}
\label{eq:SG-local-form}
\end{equation}
Thus Eq.~\eqref{eq:SG-projector-form} requires only local gates on $F$ and no
entangling gate.

By contrast, $S_0^{FAW}$ must coherently phase one joint basis state,
$\ket{0^{2n}}$.  Every $V_n$-type block begins and ends with
$W=\ket{\mathbf0}_W$, so immediately before $S_0$ it is enough to recognize
the all-zero string on $FA$.

Following the standard generalized multi-controlled-unitary notation
\cite{Barenco1995}, $\Lambda_R[U_T]$ denotes the gate that applies $U$ to a
target $T$ exactly when all $r=|R|$ control qubits are one:
\begin{equation}
\begin{aligned}
\Lambda_R[U_T]
&:={\bigl(\id_R-\ketbra{1^{r}}{1^{r}}_R\bigr)}\otimes\id_T\\
&\quad+\ketbra{1^{r}}{1^{r}}_R\otimes U_T,
\qquad r=|R|.
\end{aligned}
\label{eq:generalized-multi-control}
\end{equation}
For $U=X$, this is the generalized Toffoli gate
\cite{NieZiSun2024,KhattarGidney2025}.  With
$X_{FA}:=X_F\otimes X_A^{\otimes n}$, the required all-zero phase is therefore
\begin{align}
S_0^{FA}(\varphi)
&=X_{FA}\Lambda_A[P_F(\varphi)]X_{FA}\nonumber\\
&=\id_{FA}+(e^{i\varphi}-1)\ketbra{0^{n+1}}{0^{n+1}}.
\label{eq:S0-multicontrolled-form}
\end{align}
From right to left, the first $X_{FA}$ converts the zero controls to one
controls, $\Lambda_A[P_F(\varphi)]$ phases the component for which all $FA$
qubits are one, and the final $X_{FA}$ restores the basis labels.  Unlike the
local $S_G$, this is an $n$-controlled phase, i.e., a general
multicontrolled one-qubit unitary.  Exact general constructions attain
$O(n)$ size and $O(\log n)$ depth with one clean ancilla
\cite{NieZiSun2024}, while recent gate-count-oriented compilations provide
linear-size alternatives for general multicontrolled $U(2)$ gates
\cite{ZindorfBose2025}.  Here the $n-1$ clean qubits of $W$ are already
available, so we use the balanced-AND realization below; it keeps both the
logarithmic depth and an explicit CNOT ledger.  The circuit computes the joint
predicate, applies the phase, and uncomputes the predicate.  Measuring the
predicate would instead collapse the superposition needed for the cancellation
in Eq.~\eqref{eq:appendix-exact-lcu}.

No new tree register is needed: reuse the $n-1$ clean qubits of $W$ as the
internal nodes of a balanced AND tree with leaves $A_1,\ldots,A_n$.  Only the
root $w_\star$ enters the final controlled phase; the other $n-2$ work qubits
store intermediate conjunctions that permit logarithmic-depth computation and
uncomputation.  Let $U_{\wedge}$ denote this reversible AND-tree computation.
Then
\begin{align}
&X_{FA}U_{\wedge}^{\dagger}
\operatorname{CP}_{w_\star,F}(\varphi)
U_{\wedge}X_{FA}
\ket\psi_{FA}\ket{\mathbf0}_W\nonumber\\
&=S_0^{FA}(\varphi)\ket\psi_{FA}\ket{\mathbf0}_W,
\label{eq:S0-and-tree-form}
\end{align}
Here $\operatorname{CP}_{w_\star,F}(\varphi)$ applies $P_F(\varphi)$ when
$w_\star=1$; equivalently it is
$\operatorname{diag}(1,1,1,e^{i\varphi})_{w_\star,F}$.  This gate is symmetric
under exchanging its two qubits; we write the order $(w_\star,F)$ to emphasize
that $w_\star$ stores the AND result and controls the phase $P_F(\varphi)$ on
$F$.
Figure~\ref{fig:exact-amplification-circuit}(b) shows $U_\wedge$ computing
$w_\star=A_1\wedge\cdots\wedge A_n$, the controlled phase acting on
$(w_\star,F)$, and $U_\wedge^\dagger$ clearing $W$.
Thus the phase is applied exactly when all flipped $FA$ qubits are one.

\begin{figure}[t]
\centering
\includegraphics[width=\columnwidth]{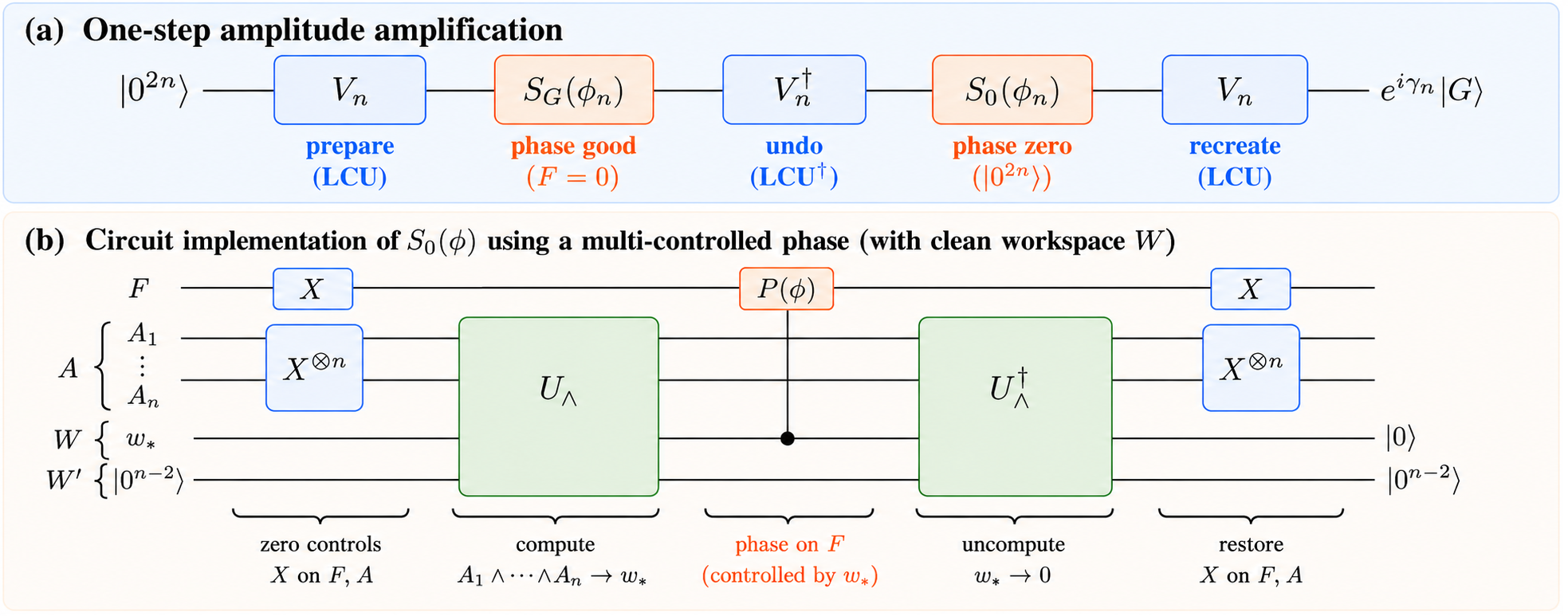}
\caption{Circuit view of exact amplitude amplification and its only nonlocal
selective phase.  Panel (a) shows Eq.~\eqref{eq:appendix-exact-lcu} in temporal
order.  Panel (b) expands $S_0$ on the invariant subspace with clean $W$:
$U_\wedge$ stores the conjunction in $w_\star$, which controls a phase on $F$,
and $U_\wedge^\dagger$ restores all of $W$ to zero.}
\label{fig:exact-amplification-circuit}
\end{figure}

For the resource ledger, let $L=\lceil\log_2n\rceil$.  The three $V_n$-type calls contribute
$3(3n-2)$ CNOTs and depth $3(2L+1)$; $S_G$ contributes no CNOT.  The AND tree
and its inverse contain $2(n-1)$ Toffolis, and the controlled phase contributes
two CNOTs.  Using an exact six-CNOT Toffoli gives
\begin{equation}
\begin{aligned}
C_{\rm prep}&=3(3n-2)+6[2(n-1)]+2=21n-16,\\
D_{\rm prep}&=3(2L+1)+2(6L)+2=18L+5.
\end{aligned}
\end{equation}
The initialized auxiliary register is only $FW$, of size
$1+(n-1)=n$; including $A$, the preparation width is $2n$.  These are
transparent standard-library upper bounds, not optimized counts.

\noindent\textit{Deterministic adaptive compilation.---}
We now prove Proposition~\ref{prop:deterministic-adaptive}.  In the local
alternating quantum--classical computation (LAQCC) model
\cite{Buhrman2024}, we trade initialized auxiliary qubits for depth:
measurement and feedforward compile the same exact five-block identity,
Eq.~\eqref{eq:appendix-exact-lcu}, from $O(\log n)$ coherent depth to $O(1)$
adaptive quantum depth.  This conservative construction uses $O(n\log n)$
initialized auxiliary qubits and $O(n\log n)$ gates.  No logical block is
changed, and no measurement outcome is postselected.

Among the five blocks, only the three $V_n$-type occurrences---$V_n$,
$V_n^\dagger$, and $V_n$---and the phase $S_0$ require nonlocal coherent
control; $S_G$ is local.  Each $V_n$-type block reduces to one exact dynamic
common-control fanout \cite{Baumer2025}, together with parallel one-qubit
gates, and therefore has constant adaptive quantum depth.  After this
replacement, the only remaining nonlocal operation is $S_0$.

As in the unitary compilation above, $W$ is in $\ket{\mathbf0}_W$ whenever
$S_0^{FAW}$ is applied.  The full-register phase therefore reduces on this
invariant subspace to the all-zero phase on $FA$.  We implement it without
measuring whether $FA$ is zero: coherently compute this predicate, apply a
phase, and erase the computation.  Introduce a clean marker qubit $T$ and
define
\begin{equation}
U_{\rm zero}:=X_{FA}\Lambda_{FA}[X_T]X_{FA}.
\label{eq:adaptive-zero-marker-mcx}
\end{equation}
This is an all-zero-controlled generalized Toffoli: for
$x\in\{0,1\}^{n+1}$,
\begin{equation}
U_{\rm zero}\ket{x}_{FA}\ket{t}_T
=\ket{x}_{FA}\ket{t\mathbin\oplus\mathbf1_{x=0}}_T .
\label{eq:adaptive-zero-marker}
\end{equation}
An exact measurement-and-feedback formulation of the multicontrolled-$X$
gate is summarized in Theorem~3 of Ref.~\cite{Zi2025}, following
Ref.~\cite{Takahashi2016}.
Since $\mathbf1_{x=0}=1\mathbin\oplus\operatorname{OR}(x)$, implementing
$U_{\rm zero}$ is equivalent, up to a local $X_T$, to computing a reversible
OR marker.  The OR primitive is only the physical compiler for the same
logical operation $U_{\rm zero}$.  The exact quantum OR construction of
Ref.~\cite{Takahashi2016} and its local adaptive compilation in
Ref.~\cite{Buhrman2024} provide such a measurement-and-feedforward
implementation.  Each marker call has $O(1)$ adaptive quantum depth, a
constant number of measurement rounds, and $O(\log n)$ classical feedforward
depth, using $O(n\log n)$ initialized internal auxiliaries and gates.  This
internal workspace---not the single logical marker $T$---accounts for the
quoted width.  Only the internal auxiliary qubits are measured: $T$, and hence
the value of the all-zero predicate, is kept coherent.  Feedforward applies
the outcome-dependent corrections, so the net operation on $FAT$ is exactly
$U_{\rm zero}$.

The selective phase is now obtained in three steps: compute the marker with
$U_{\rm zero}$, apply
$P_T(\varphi)=\operatorname{diag}(1,e^{i\varphi})_T$, and uncompute with
$U_{\rm zero}^\dagger=U_{\rm zero}$.  Explicitly,
\begin{align}
&U_{\rm zero}^{\dagger}
\bigl(\id_{FA}\otimes P_T(\varphi)\bigr)U_{\rm zero}
\ket{\psi}_{FA}\ket0_T\nonumber\\
&\qquad={}
\left[\id_{FA}+(e^{i\varphi}-1)\ketbra{0^{n+1}}{0^{n+1}}\right]
\ket{\psi}_{FA}\ket0_T .
\label{eq:adaptive-zero-phase}
\end{align}
Thus $T$ is returned to zero, while the all-zero component of $FA$ alone
acquires $e^{i\varphi}$.  This is exactly $S_0$ on the invariant subspace;
the predicate is neither observed nor postselected.

The full circuit therefore contains three constant-depth dynamic $V_n$-type
calls, the local phase $S_G$, and two constant-depth marker calls surrounding
one local phase.  The marker construction dominates the conservative resource
bound.  Because the constant number of calls reuse its internal workspace, the
complete preparation has $O(n\log n)$ peak initialized width and gate count,
$O(1)$ adaptive quantum depth, a constant number of measurement rounds, and
$O(\log n)$ classical feedforward depth.  Because every corrected block is
exact, the cancellation
in Eq.~\eqref{eq:appendix-exact-lcu} still gives $\ket G$ with unit probability:
there is no heralding measurement and no retry tail.  This is the deterministic
adaptive entry in Table~\ref{tab:hierarchy}.  All auxiliary-side processing is
completed before coupling the prepared fiducial to the unknown state.

\FloatBarrier

\end{document}